\documentclass[12pt,reqno,openany]{report} 
\usepackage[a4paper]{geometry}
\usepackage{amsfonts,amssymb,amsmath,amsthm,amscd}
\usepackage[nottoc]{tocbibind}

\def\C{\mathbb{C}}

\def\Z{\mathbb{Z}}
\def\N{\mathbb{N}}
\let\ds\displaystyle

\newcommand{\be}{\begin{equation}}
\newcommand{\ee}{\end{equation}}

\newcommand{\pa}{\partial}

\theoremstyle{plain}
\newtheorem{theorem}{Theorem}[chapter]
\newtheorem{lemma}{Lemma}[chapter]
\newtheorem{corollary}{Corollary}[chapter]
\newtheorem{remark}{Remark}[chapter]
\newtheorem{proposition}{Proposition}[chapter]
\newtheorem{conjecture}{Conjecture}[chapter]
\newtheorem{question}{Question}[chapter]
\newtheorem{op}{Open problem}[chapter]
\newtheorem{observation}{Observation}[chapter]

\theoremstyle{remark}
\newtheorem{definition}{Definition}[chapter]
\newtheorem{example}{Example}[chapter]

\newtheorem{exercise}{Exercise}[chapter]

\def\dsize{\displaystyle}
\def\p{\partial }

\begin{document}

\begin{center}

\vspace{100mm}

{\Huge{\sc Symmetry approach to  integrability   and non-associative \\[5mm] algebraic structures.}
}


\vspace{35pt}
{\bf \LARGE Vladimir Sokolov$^{a,b}$}\\[25pt]

{${}^a$\normalsize
 Landau Institute for Theoretical Physics, \\
142432 Chernogolovka (Moscow region), Russia\\}e-mail: {\normalsize
 \it
vsokolov@landau.ac.ru}\vspace{15pt}

{${}^b$\normalsize
 Sao Paulo University, \\
 Instituto de Matematica e Estatistica
 \\ 05508-090, Sao Paulo, Brazil} \vspace{10pt}

\vspace{3mm}

\end{center}

\tableofcontents

\date{}

\chapter{Introduction}

\qquad The symmetry approach to the 
 classification of integrable PDEs 
has been developed since 1979 by: A.~Shabat, A.~Zhiber,  N.~Ibragimov,  A.~Fokas,
V.~Sokolov, S.~Svinolupov, A.~Mikhailov,  R.~Yamilov, V.~Adler,
P.~Olver, J.~Sanders, J.P.~Wang, V.~Novikov, A.~Meshkov, D.~Demskoy, H.~Chen, Y.~Lee,
C.~Liu, I.~Khabibullin,   B.~Magadeev,  
 R.~Heredero,   V.~Marikhin,  M.~Foursov, S.~Startsev, M.~Balakhnev, and others. It is very efficient for PDEs with two independent variables and  under additional assumptions  can be applied to ODEs.
 
The basic definition of the symmetry approach is: 
 
\begin{definition} A differential equation is integrable if it possesses infinitely many
higher infinitesimal symmetries.
\end{definition}

The reader may ask: ``Why should equations integrable in one sense or another have higher symmetries?'' 

A not rigorous but instructive answer is the following. 
Linear differential equations have infinitely many higher
symmetries.  Integrable nonlinear differential equations known as of today are related to linear equations 
by some transformations. The same transformations produce higher
symmetries of  the nonlinear equation  from symmetries of the  linear one.

It turns out that the existence of higher symmetries allows one to find all integrable equations 
from a beforehand prescribed class of equations.
The first classification result in  frames of the symmetry approach was:
\begin{theorem}\label{ZS} {\rm  \cite{zibshab1}}  
The nonlinear hyperbolic equation
of the form
$$
u_{xy}=F(u)
$$
possesses higher symmetries iff (up to scalings and shifts)
$$
F(u)=e^{u}, \quad F(u)=e^u+e^{-u}, \quad \mbox{or} \quad F(u)=e^u+e^{-2u}.
$$
\end{theorem}

Several reviews \cite{sokshab}--\cite{int2} are devoted to the symmetry approach (see also \cite{ibrag}--\cite{Olv93}). In this book I mostly concentrate on the results which are not 
 covered by these papers and books.  
I  consider only ODEs and PDEs with two independent variables.   For integrable multi-dimensional equations of KP type, most of the methods under consideration are
not applicable.
 
The statements are formulated  in the most simple form  but often possible ways for generalization are pointed out.  In the proofs only essential points are mentioned while for technical details references are given.   The text contains many carefully selected examples, which give a perception  on the subject. A number of open problems are suggested.

The author is not a scrabble in original references.  Instead, some references to reviews, where an information of pioneer works can be found, are given.

The book  addresses both experts in algebra and in classical integrable systems. It is suitable for PhD students and can serve as an introduction to classical integrability for scientists with algebraic inclinations.

The exposition is based on a series of lectures  delivered by the author in USP (Sao Paulo, 2015). 

The contribution of my collaborators A. Mikhailov, A. Meshkov, S. Svinolupov, and  A. Shabat  to results presented in  this book is difficult to  overestimate. 
 
I am thankful to the first readers of the book  Yu. Bagderina, S. Carpentier, I. Freire, S. Konstantinou-Rizos, S. Startsev and  A. Zobnin who made many suggestions and found a lot of misprints which contributed to the improvement of the text. 

I am grateful to I. Shestakov, V. Kac, I. Krasilchchik, A. Maltsev, E. Ferapontov and V. Futorny for their attention and to FAPEST for the financial support (grants 2014/00246-2 and 2016/07265-8) of his visits to Brazil, where the book was written.  

I thank  my wife Olga Sokolova, who provided excellent conditions for
writing this book.
 
\bigskip

\section{List of basic notation} 

\subsection{Constants, vectors and matrices}

Henceforth, the field of constants is $\C$; $\bf u$ stands for $N$-dimensional vector, namely ${\bf u}=(u^1,\dots, u^N).$ Moreover, the  standard scalar product $\sum_{i=1}^{N} u^i\, v^i$ is denoted by $\langle {\bf u},\, {\bf v}\rangle$.

The associative algebra of order ``m'' square matrices is denoted by ${\rm Mat}_m$; the matrix  $\{u_{ij}\}\in {\rm Mat}_m$ is denoted by $\bf U.$  The unity matrix is denoted by ${\bf 1}$ or  ${\bf 1}_m$. The notation ${\bf U}^t$ stands for the matrix transpose of ${\bf U}$.

For the set of $n\times m$ matrices we use the notation ${\rm Mat}_{n,m}.$

\subsection{Derivations and differential operators} 
For ODEs the independent variable is denoted by $t$, whereas for PDEs we have two independent variables $t$ and $x$. Notation $u_t$ stands for the partial derivative of $u$ with respect to $t$. For the $x$-partial derivatives of $u$ the notation $u_x=u_1$, $u_{xx}=u_2,$ etc, is used.

The operator $\ds \frac{d}{d x}$ is often denoted  by $D$. For the differential operator $L=\sum_{i=0}^{k} a_i \,D^i$ we define the operator $L^{+}$ as
$$
L^+=\sum_{i=0}^{k} (-1)^i \,D^i\circ\, a_{i},
$$
where $\circ$ means that, in this formula, $a_i$ is the operator of multiplication by $a_i$. By $L_t$ we denote 
$$
L_t = \sum_{i=0}^{k} (a_i)_{t} \,D^i.
$$

\subsection{Differential algebra}
We denote by ${\cal F}$  a differential field. 
For our main considerations one can assume that elements of $ {\cal F}$ are rational functions of finite number of independent variables $u_i.$ However, very often we find some functions solving overdetermined systems of PDEs. In such a case we have to extend the basic field $\cal F$. We will avoid any formal description of such extensions hoping that in any particular case it is clear what we really need from $ {\cal F}$.
 
The principle derivation 
\be \label{DD}
D \stackrel{def}{=} \frac{\partial}{\partial x}+\sum_{i=0}^\infty u_{i+1} \frac{\partial}{\partial u_i},
\ee generates all independent variables $u_i$ starting from $u_0=u$. 

When we speak of solutions (or common solutions) of ODEs and PDEs, we mean {\it local}  solutions with generic initial date. 

\subsection{Algebra} 
We denote by $A(\circ)$ an $N$-dimensional algebra $A$ over $\C$ with an operation $\circ$.  A basis of $A$ is denoted by ${\bf e}_1, \dots , {\bf e}_N$, and the corresponding structural constants  by $C^{i}_{jk}$:
$$
{\bf e}_j \circ {\bf e}_k =  C^i_{jk}\, {\bf e}_i.  
$$
We denote by $U$ the element  
\begin{equation}\label{UU}
U=\sum^N_{i=1} u_i {\bf e}_i.
\end{equation}
In what follows we assume that the summation is carried out over repeated indices. 
We will use the following notation: 
\begin{equation}\label{as}
{\rm As}(X, Y, Z) = (X \circ Y) \circ Z - X \circ (Y \circ Z),
\end{equation}
\begin{equation}\label{br}
[X, Y, Z] = {\rm As}(X, Y, Z) - {\rm As}(Y, X, Z).
\end{equation}

By ${\cal G}$ and ${\cal A}$ we usually denote a Lie and an associative algebra, respectively.

\bigskip

\section{Infinitesimal symmetries}

Consider a dynamical system of ODEs
\begin{equation}
\frac{dy^{i}}{dt} = F_{i}(y^1,\ldots ,y^n),\ \ \ \ i=1, \ldots ,n
\, .  \label{dynsys}
\end{equation}
 
\begin{definition} The dynamical system 
\begin{equation}
\frac{dy^{i}}{d\tau} = G_{i}(y^1,\ldots ,y^n),\ \ \ \ i=1, \ldots
,n  \,   \label{dynsym}
\end{equation}
is called an (infinitesimal) symmetry for~\eqref{dynsys} iff~\eqref{dynsys} and~\eqref{dynsym}
are compatible.
\end{definition}

Informally speaking,  compatibility means that for any initial data ${\bf y}_0$ there exists a common solution 
${\bf y}(t,\tau)$ of equations~\eqref{dynsys} and~\eqref{dynsym} such that ${\bf y}(0,0)={\bf y}_0$. 
\bigskip

More rigorously, it means that $$X Y-Y X=0,$$ where 
\begin{equation}\label{ODEfield} X=\sum F_i \frac{\partial}{\partial y^i},   \qquad Y=\sum G_i \frac{\partial}{\partial y^i}.
\end{equation}  

Consider now an evolution 
equation
\begin{equation}\label{eveq}
u_t=F(u, u_x,  u_{xx}, \dots , u_n), \qquad u_i=\frac{\partial^i
u}{\partial x^i}.
\end{equation}
A higher (or generalized) infinitesimal symmetry (or a commuting flow) is an evolution equation
\begin{equation}\label{evsym}
u_{\tau}=G(u, u_x,  u_{xx}, \dots , u_m), \qquad m > 1, 
\end{equation}
which is compatible with~\eqref{eveq}. Compatibility means that \begin{equation}\label{ttau}
\frac{\partial}{\partial t}\frac{\partial u}{\partial \tau}=\frac{\partial}{\partial \tau}\frac{\partial u}{\partial t},
\end{equation}
where the partial derivatives are calculated in virtue of~\eqref{eveq} and~\eqref{evsym}. In other words, for any initial value $u_0(x)$ 
there exists a common solution $u(x,t,\tau)$ of equations~\eqref{eveq} and~\eqref{evsym} such that $u(x,0,0)=u_0(x)$.
For a more rigorous definition in terms of evolution vector fields see~\cite{Olv93} and/or Section~\ref{SectionIntegrabilityConditions}. 

\begin{remark} Infinitesimal symmetries \eqref{evsym} with $m \le 1$ correspond to one-parametric groups of point or contact transformations~{\rm \cite{Olv93}}. They are called {\rm classical symmetries}. We do not consider them in this paper in spite of the fact that they are related to an important class of differential substitutions for evolution equations~{\rm \cite{SvSok26}}. 
\end{remark}

\begin{example} Any equation \eqref{eveq} has the classical symmetry $u_{\tau}=u_x,$ which corresponds to the one-parametric group $x \to x+\lambda$ of shifts. 
\end{example}
 
\begin{example} \label{Example1.4} 
For any $m$ and $n$ the equation $u_{\tau}=u_m$ is a symmetry for the linear equation $u_t=u_n$. The symmetries for different $m$ are compatible with each other. Thus, we have an infinite hierarchy of equations such that any of them is a symmetry for others. 
\end{example}

\begin{example} \label{Example1.5}
The Burgers equation
\be \label{burgers} 
u_t=u_{xx}+2 u u_x
\ee
has the following third order symmetry
\be \label{burgersym} 
u_{\tau}=u_{xxx}+3 u u_{xx}+3 u_x^2+3 u^2 u_x.
\ee
\end{example}

\begin{example} \label{Example1.6}
The simplest higher symmetry for the Korteweg--de Vries (KdV) equation
\begin{equation} \label{kdv} 
u_t=u_{xxx}+ 6\,u\,u_x
\end{equation}
has the following form
\begin{equation}\label{kdvsym}
u_{\tau}=u_{5}+10 u u_{3}+20 u_1 u_2+30 u^2 u_1.
\end{equation}
\end{example}

\begin{remark}\label{rem12} The existence of higher symmetries is a strong indication that the equation~\eqref{eveq} is integrable. 
One can propose the following ``explanation'' of this fact.
According to Example~\ref{Example1.4}, the linear equation has infinitely many higher
symmetries. As a rule, an integrable nonlinear equation is related to a linear
one by some transformation. The same transformation produces a hierarchy of 
higher symmetries for nonlinear equation starting from the symmetries
of the corresponding linear equation.
\end{remark}

For instance, the Burgers equation is integrable because of the
Cole--Hopf substitution
\begin{equation}\label{colehopf}
u=\frac{v_x}{v},
\end{equation}
which relates~\eqref{burgers} to the linear heat equation $v_t=v_{xx}$.
Moreover, the same substitution maps the third order symmetry~\eqref{burgersym} of the Burgers equation
to $$v_{\tau}=v_{xxx},$$ etc.

The transformation that reduces the KdV equation to the linear equation $v_t=v_{xxx}$ is a non-linear generalization of the Fourier 
transform~\cite{Z, AblSeg81}. It is much more nonlocal than the Cole-Hopf substitution. Nevertheless, the inverse transformation can be applied 
to all symmetries $v_{\tau}=v_{2 n+1}$ of the linear equation to produce an infinite hierarchy of symmetries of odd order for the KdV equation. 

\subsubsection{Naive symmetry test}
By an example of fifth order equations we demonstrate how to state and to solve a simple
classification problem for integrable polynomial homogeneous equations.

The differential equation \eqref{eveq}
 is said to be $\lambda $-{\it homogeneous} of {\it weight} $\mu $ if it 
admits the one-parameter group of scaling symmetries
$$(x, \ t, \ u)\longrightarrow (\tau^{-1}x, \ \tau^{-\mu} t, \ \tau^{\lambda} u).$$
For $N$-component systems with unknowns $u^1,...,u^N$ the corresponding 
scaling group has a similar form
\begin{equation}\label{homo}(x,t,u^1,...,u^N)\longrightarrow (\tau^{-1} x, \ \tau^{-\mu} t, \ 
\tau^{\lambda_1} u^1,...,
\tau^{\lambda_N} u^N).\end{equation}

\begin{theorem}\label{schomo} {\rm \cite{sw}} Scalar $\lambda$-homogeneous 
polynomial equation with $\lambda>0$ may possess a homogeneous polynomial 
higher symmetry only if 
\begin{itemize}
\item Case 1: \qquad $\lambda=2$;
\item Case 2: \qquad $\lambda=1$;
\item Case 3: \qquad $\lambda=\frac{1}{2}.$
\end{itemize}
\end{theorem}

For example, the KdV equation \eqref{kdv} is 
homogeneous of weight $3$ for $\lambda =2$, its symmetry ~\eqref{kdvsym} has the same homogeneity. The mKdV equation 
 $u_t = u_{xxx} + u^2 u_x $ has the weight $3$ for $\lambda=1$ and for the 
Ibragimov-Shabat equation 
\begin{equation} \label{IbShab}
u_t = u_{xxx} + 3 u^2 u_{xx} + 9 u u_x^2 + 3 u^4 u_x 
\end{equation}
the weight is $3$ and $\lambda=\frac{1}{2}$.

 The general form of a fifth order polynomial
equation of  homogeneity $\lambda=2$ is given by
\begin{equation}\label{kdv5}
u_t=u_5+a_1 u u_3 + a_2 u_1 u_2+a_3 u^2 u_1,
\end{equation}
where $a_i$ are constants.
Let us find all equations \eqref{kdv5} having a polynomial homogeneous seventh order symmetry  
 of the form 
\begin{equation}\label{kdv7}
u_{\tau}=u_7+c_1 u u_5+ c_2 u_1 u_4+c_3 u_2 u_3+c_4 u^2 u_3+c_5 u u_1 u_2
+c_6 u_1^3+c_7 u^3 u_1.
\end{equation}

Compatibility condition \eqref{ttau} can be rewritten in the form $F_{\tau}-G_t=0$. When we eliminate $\tau$ and $t$-derivatives in virtue of~\eqref{kdv5} and~\eqref{kdv7} from this defining equation,
the left-hand side becomes
a polynomial of sixth degree in variables $u_1,\dots , u_{10}$. Linear terms are
absent. Equating coefficients of quadratic terms to zero, we find that
$$c_1= \frac{7}{5} a_1, \qquad c_2= \frac{7}{5}(a_1+a_2),
\qquad c_3= \frac{7}{5} (a_1+2 a_2).$$
The conditions arising from the cubic terms allow one to express
$c_4, c_5$ and $c_6$ in terms of
$a_1, a_2, a_3$. Moreover, it turns out that
$$a_3=-\frac{3}{10} a_1^2+\frac{7}{10} a_1 a_2-\frac{1}{5} a_2^2.$$
The fourth degree terms give us an explicit dependence of $c_7$ on
$a_1, a_2$ and also the main algebraic relation
$$(a_2-a_1)(a_2-2 a_1)(2 a_2-5 a_1)=0$$
for the coefficients $a_1$ and $a_2$.
Solving this equation we find that up to the scaling  $u \rightarrow
\lambda u$ only four different integrable cases are possible: the linear equation
$u_t=u_5,$ equations 
\begin{equation}\label{sk}
u_t=u_5+5 u u_3+ 5 u_1 u_2+5 u^2 u_1,
\end{equation}
\begin{equation}\label{kk}
u_t=u_5+10 u u_3+ 25 u_1 u_2+20 u^2 u_1,
\end{equation}
and \eqref{kdvsym}.
In each of these cases all terms of the fifth and the sixth degrees in the
defining equation are canceled automatically.
Equations \eqref{sk} and \eqref{kk} are well known   
\cite{SawKot74, Kau80}.  

\begin{question} The problem solved above looks artificial. Why do we consider the fifth order equations with the seventh order symmetry? 
\end{question}

Under the assumption that the right-hand side of equation \eqref{eveq}
is polynomial and homogeneous with $\lambda > 0$, it was proved in~\cite{sw} that it suffices to consider the following three cases:
\begin{itemize}
\item a second order equation with a third order symmetry; 
\item a third order equation with a fifth order symmetry; 
\item a fifth order equation with a seventh order symmetry.
\end{itemize} 
Other integrable equations belong to the hierarchies of such equations. This statement looks very credible and without any additional restrictions for the right-hand side of equation. A proof in the non-polynomial case is absent and this statement has a status of a conjecture well-known for experts. In fact, no counterexamples to this conjecture are known.

In Section~\ref{SectionIntegrabilityConditions} an advanced symmetry test will be considered.
No restrictions such as the polynomiality of the right-hand side of \eqref{eveq} or the fixation of symmetry  order are imposed there.

\section{First integrals and local conservation laws}\label{conslaw}

In the ODE case the concept of a first integral (or integral of motion) is one of the basic notions. A function $f(y^1,\dots,y^n)$ is called a {\it first integral} for the system \eqref{dynsys} if the value of this function does not depend on $t$ for any solution $\{y^1(t),\dots,y^n(t)\}$ of \eqref{dynsys}.  Since
$$
\frac{d}{dt} \Big( f(y^1(t),\dots,y^n(t))\Big) = X(f),
$$
where the vector field $X$ is defined by \eqref{ODEfield}, from the algebraic point of view a first integral is a solution of the first order PDE
$$
X f\Big(y^1,\dots,y^n\Big)=0.
$$

In the case of evolution equations of the form~\eqref{eveq} an integral of motion is not a function but a functional that does not depend on $t$ for any solution $u(x,t)$ of equation~\eqref{eveq}. 

More rigorously, a local conservation law for equation~\eqref{eveq} is a pair of functions
$\rho(u,u_x,...)$ and $\sigma(u,u_x,...)$ such that 
\be \label{conlaw}
\frac{\partial}{\partial t}\Big(\rho(u, u_x,\dots, u_p) \Big)=\frac{\partial}{\partial x}\Big(\sigma(u, u_x,\dots, u_q) \Big)
\ee
for any solution $u(x,t)$ of equation~\eqref{eveq}.
The functions $\rho$ and $\sigma$ are called a {\it density} and a {\it flux} of conservation law~\eqref{conlaw}. It is easy to see that $q=p+n-1$, where $n$ is the order of equation \eqref{eveq}.

For solutions of solitonic type, which are decreasing at $x \to \pm \infty$, we get 
$$
\frac{\partial}{\partial t} \int_{-\infty}^{+\infty} \rho\, dx = 0  
$$
for any polynomial conserved density with zero constant term. 
This justifies the name {\it conserved density} for the function $\rho$.
Analogously, if $u(x,t)$ is a function periodic in $x$  with period $L$, then the value of the functional 
$I_i=\int_0^L \rho\, dx$ does not depend on time and therefore it is a constant
of motion.

Suppose that functions $\rho$ and $\sigma$ satisfy \eqref{conlaw}. Then for any function $s(u,u_x,
\dots)$ the functions $\bar \rho=\rho+\frac{\partial s}{\partial x}$ and 
 $\bar \sigma=\sigma+\frac{\partial s}{\partial t}$ satisfy \eqref{conlaw} as well. We call the conserved densities $\rho$ and $\bar \rho$ {\it equivalent}. It is clear that  equivalent densities define the same functional.
 
\begin{example} \label{ex25} Any linear equation of the form $u_t=u_{2 n+1}$ has infinitely many conservation laws with densities $\rho_k=u_k^2.$ Developing the arguments of Remark \ref{rem12}, one can say that each conserved density is expected to be common to all equations of odd order from a commutative hierarchy. 
\end{example} 
\begin{exercise} Check that for any $n\ge 1$ the equation $u_t = u_{2 n}$ has only one conserved density $\rho = u.$

\end{exercise}

\begin{example} \label{Example1.7}
Functions
$$\rho_1=u,\qquad \rho_2=u^2,\qquad \rho_3=-u_x^2+2u^3 $$
are conserved densities of the Korteweg--de Vries equation~\eqref{kdv}.
Indeed, 
\begin{gather*}
\frac{\partial}{\partial t}\Big(\rho_1\Big)=\frac{\partial}{\partial x}\Big(u_2+3u^2\Big),\\[2mm]
\frac{\partial}{\partial t}\Big(\rho_2\Big)=\frac{\partial}{\partial x}\Big(2 u u_{xx}-u_x^2+4u^3\Big),\\[2mm]
\frac{\partial}{\partial t}\Big(\rho_3\Big)=\frac{\partial}{\partial x}\Big(9u^4+6u^2 u_{xx}+u_{xx}^2-12u u_x^2-2u_x u_3\Big).
\end{gather*}

Let us find all the conserved densities of the form $\rho(u)$ for equation~\eqref{kdv}. It is clear that the function $\sigma$ may depend on $u, u_x, u_{xx}$ only. Relation \eqref{conlaw} has the form 
\begin{equation}\label{med}
 \rho'(u) (u_3 + 6 u u_x) = \frac{\partial\sigma}{\partial u_{xx}}\, u_3+\frac{\partial\sigma}{\partial u_{x}}\, u_{xx}+\frac{\partial\sigma}{\partial u}\, u_x.
\end{equation}
Since $u$ is an arbitrary solution of the KdV equation, the latter expression should be an identity in variables $u,u_x,u_{xx}, u_3.$   Comparing the coefficients of $u_3,$ we find 
$\sigma=\rho'(u) u_{xx}+\sigma_1(u,u_x).$  Substituting it into \eqref{med} and equating the coefficients of $u_{xx}$, we obtain 
$$\sigma_1=-\frac{\rho''(u)}{2} u_x^2 +\sigma_2(u).$$
Taking it into account, we see that the coefficients of $u_x^3$ give rise to $\rho'''(u)=0$ i.e. $\rho=c_2 u^2 + c_1 u + c_0.$  So, up to the trivial term $c_0,$ the density is a linear combination of the above densities $\rho_1$ and  $\rho_2.$ 
 
\end{example}

\medskip

\section{Transformations}\label{Tra}

\subsection{Point and contact transformations} \label{pc}
Let $x_1,...,x_n$ be
independent variables and   $u$ is the dependent variable.
All symbols
\begin{equation}\label{jet}
x_1,...,x_n, \quad u, \quad \hbox{and}\quad u_{\alpha},
\end{equation}
where $\alpha=(\alpha_1,...,\alpha_n)$ and
$\ds u_{\alpha}=\frac{\partial^{\alpha_1+\cdots+\alpha_n}\,u}{\partial^{\alpha_{1}}
x_{1}\cdots \partial^{\alpha_{n}} x_{n} },$ are regarded as {\it
independent}. Let denote by ${\mathfrak F}$  the field of "all" functions depending on finite number of variables \eqref{jet}. 

The total $x_i$-derivatives are given by
$$
D_i=\sum_{\alpha} u_{(\alpha_1,...,\alpha_{i}+1, ...,\alpha_{n})}
\frac{\partial}{\partial\, u_{(\alpha_1,...,\alpha_i,
...,\alpha_{n})}}
$$
They are vector fields such that $[D_i, D_j]=0$.

 We consider transformations of the form
\begin{equation}\label{trans}
\bar x_i=\phi_i, \qquad \bar u=\psi, \qquad \phi_i, \psi \in
{\mathfrak F}
\end{equation}
The new derivatives are given by
$$
\bar u_{(\alpha_1,...,\alpha_{n})}=\bar D_1^{\alpha_1}\cdots \bar
D_n^{\alpha_n} \,(\bar u),
$$
where \begin{equation}\label{lincom} \bar D_i=\sum_{j=1}^{n}
p_{ij}\,D_j, \qquad p_{ij} \in {\mathfrak F}.
\end{equation}
To find $p_{ij},$ we apply \eqref{lincom} to $\bar x_k=\phi_k$ and
obtain
\begin{equation}\label{coef}
\sum_{j=1}^{n} p_{ij}\,D_j(\phi_k)=\delta_k^i.
\end{equation}
In other words, the matrix $P=\{ p_{ij}\}$ is inverse for the
generalized Jacobi matrix $J$ with the entries $J_{ij}=D_i(\phi_j).$

Transformation \eqref{trans} is called {\it point transformation} if the functions $\psi$ and $\phi_i$ depend 
only on $x_1,\dots,x_n,u$.   Such a transformation is (locally) invertible if $\psi$ and $\phi_i$ are functionally independent.

There are invertible transformations more general than point ones. 
\subsubsection{Contact transformation}
\begin{example}  Consider the Legendre transformation
\begin{equation}\label{leg1}
\bar x_i=\frac{\partial u}{\partial x_i}, \qquad \bar
u=u-\sum_{k=1}^{n} x_k \frac{\partial u}{\partial x_k}.
\end{equation}
Notice that the functions $\psi$ and $\phi_i$ depend on first detivatives of $u$. For such tranformations the functions $\frac{\partial \bar u}{\partial \bar x_{i}}$ usually depend on second detivatives of $u$ and so on. In this case the transformation extended to derivatives of degree not greater than $k$ is not invertible for any $k.$
However, for transformation \eqref{leg1},
according to \eqref{coef}, we have 
$$\sum_{j}
p_{ij}\,\frac{\partial^2 u}{\partial x_{j}\partial
x_{k}}=\delta_k^i.$$ 
One can verify that
\begin{equation}\label{leg2}
\frac{\partial \bar u}{\partial \bar x_{i}}=\sum_j p_{ij} D_j
(u-\sum_k x_k \frac{\partial u}{\partial x_{k}})=-x_i.
\end{equation}
Hence transformation \eqref{leg1} extended to the first deravatives by \eqref{leg2} is invertible. 
\end{example}

The Legendre transformation is an example of a contact
transformation.

\begin{theorem} {\rm   \cite{back}}. If a transformation extended to derivatives of degree not greater than $k$ 
is invertible for some $k,$ then this is either a point or a contact
transformation.
\end{theorem}

 Let us consider contact transformations for the case $n=1$ in more
details. They have the following form
\begin{equation}\label{conttran}
\bar x=\phi(x,u,u_x), \qquad \bar u=\psi(x,u,u_x), \qquad \bar
u_1=\chi(x,u,u_x),
\end{equation}
where 
\begin{equation}\label{contcond}
\psi_{u_{1}} (\phi_{u} u_1+\phi_x)=\phi_{u_{1}} (\psi_{u}
u_1+\psi_x).
\end{equation}
Under the latter condition the function 
$$
\chi=\frac{D(\psi)}{D(\phi)}=\frac{\psi_{u_{1}} u_2+\psi_{u}
u_1+\psi_x}{\phi_{u_{1}} u_2+\phi_{u} u_1+\phi_x}
$$
does not depend on $u_2$ and the transformation is invertible if the functions $\phi,\psi$ and $\chi$ are functionally independent.
 
\begin{remark} Formally, the point transformations may be regarded as a particular case of contact transformations: if the functions 
$\phi$ and $\psi$ do not depend on $u_x$, the contact condition \eqref{contcond} disappears.  
\end{remark} 

For the Legendre transformation we have  
\begin{equation}\label{legend}
\phi=u_1, \qquad \psi=u- x\
u_1.
\end{equation}
The contact condition \eqref{contcond} is fulfilled and $\chi=-x$.

 The formulas \eqref{conttran} and \eqref{contcond} become simpler if we introduce a generating
function $F$: 
$$
\psi(x,u,u_1)=F(x,u, \phi(x,u,u_1)).
$$
It is possible if $\phi_{u_{1}}\ne 0$. We have
$$
\bar u_{1}=F_{\phi}, \qquad  F_{x}+u_1\, F_{u}=0
$$
or $$ \bar u=F(x,u,\bar x), \qquad \bar u_1=\frac{\partial
F}{\partial \bar x}, \qquad u_1=-\frac{\partial F}{\partial
x}\left(\frac{\partial F}{\partial u}\right)^{-1}.$$

\begin{exercise} Verify that for any function $F(x,u,\bar x)$
the above formulas define a contact transformation.
\end{exercise}

Point and contact transformations \eqref{conttran} are important for classification of integrable evolution equations of the form
$$
u_t = F(x, u, u_x, \dots, u_n)
$$
since they preserve the form of the equation and transform symmetries to symmetries, i.e., map integrable equations to integrable 
ones.  Indeed, consider a contact transformation of the form 
$$
\bar t=t, \qquad \bar x=\phi(x,u,u_x), \qquad \bar u=\psi(x,u,u_x), \qquad \bar
u_1=\chi(x,u,u_x),
$$
where the contact condition \eqref{contcond} holds. Then the coefficients in the relations 
$$\bar D_t=\alpha D_t+\beta D,     \qquad  \bar D=\gamma D_t+\delta D$$ can be found from the conditions 
$$\bar D_t(t)=1, \qquad \bar D_t(\phi)=0, \qquad \bar D(t)=0, \qquad \bar D(\phi)=1. $$
We obtain 
$$
\bar D=\frac{1}{D(\phi)}\, D, \qquad \bar D_t=D_t - \frac{D_t(\phi)}{D(\phi)}\, D
$$
and, therefore, 
\begin{equation}\label{neweq}
\bar u_t = \Big(\psi^{+} -  \frac{D(\psi)}{D(\phi)}\, \phi^{+}\Big) (F).
\end{equation}
The variables $x,u,u_x,...$ in the rigth-hand side of the latter equation are supposed to be replaced by the new variables $\bar x, \bar u,,...$. 
Due to condition \eqref{contcond} the differential operator 
\begin{equation}\label{subop}
\psi^{+} -  \frac{D(\psi)}{D(\phi)}\, \phi^{+}
\end{equation}
has zero order and the evolution equation has the form 
$$
\bar u=\bar F(\bar x,\bar u, \bar u_1,  \dots, \bar u_n)
$$
with the same $n$.
\begin{exercise} Verify that the heat equation $u_t=u_2$ transforms to 
$$
\bar u_t=-\frac{1}{\bar u_2}
$$
under the Legendre transformation \eqref{legend}.
\end{exercise}

\begin{remark} If we restrict ourselves with the evolution equations, where the function $F$ does not depend on $x$ explicitly, then we still can use a subgroup of contact transformations with $\phi=x + s(u, u_x).$
\end{remark}

\subsection{Differential substitutions of Miura type}\label{miur}

The famous Miura transformation \cite{miura}
$$
\bar u=u_{1}-u^{2}
$$
relates the mKdV equation $$ u_{t}=u_{3}-6 u^{2} u_{1}$$ and the KdV
equation $$\bar u_{t}=\bar u_{3}+6 \bar u \bar u_{1}.$$ 
The Miura substitution is not locally invertible: given a solution $\bar u$ of the KdV equation, one has to solve the  Riccati equation to find a solution of the mKdV equation. 

\begin{definition} A relation
\begin{equation}\label{genmiur}
\bar x = \phi(x, u, u_2,\dots, u_{k}), \qquad  \bar u = \psi(x, u, u_1,\dots, u_{k})
\end{equation}
is called a {\it differential substitution} from the
equation
\begin{equation} \label{up}
u_{t}= F(x, u, u_1, \dots, u_{n})
\end{equation}
to the equation
\begin{equation} \label{down}
\bar u_{t}= G(x, \bar u, \bar u_1, \dots ,\bar u_{n})
\end{equation}
if for any solution $u(x,t)$ of equation \eqref{up} the function
$\bar u(t, \bar x)$ satisfies \eqref{down}. The order of differential operator \eqref{subop} is called the {\it order of the differential substitution}.  
\end{definition}
\begin{remark}
By definition, point and contact transformations are differential substitutions of order 0, while the order of the Miura substitution is equal to one.   
\end{remark}

The derivatives of $\bar u$ are expressed through the functions $\phi$ and $\psi$ just as in Section \ref{pc}. 
However, for the generic function $F$ the right-hand side of \eqref{neweq} cannot be expressed through the variables $\bar x, \bar u, \dots .$ This requirement imposes strong restrictions of $F$, $\phi$ and $\psi.$
\begin{exercise} Verify  that in the case of the Miura substitution for the MkDv equation, the right-hand side of \eqref{neweq} is expressed in terms of $\bar u,\dots, \bar u_3$ as $\bar u_{3}+6 \bar u \bar u_{1}$.
\end{exercise}
\begin{proposition}  Let $\bar D_t(\bar \rho)=\bar D(\bar \sigma)$ be a conservation law for equation \eqref{down}. Then $D_t(\rho)=D(\sigma),$ where
$$
\rho=\bar \rho \, D(\phi), \qquad \sigma = \bar \sigma + \bar \rho \,D_t(\phi),
$$
is the conservation law for equation \eqref{up}.

\end{proposition}

\subsubsection{Group differential substitutions}
In this section we consider differential substitutions associated with classical groups of symmetries \cite{SvSok26}. 
\begin{example} The Burgers equation 
\begin{equation}\label{bur}
\bar u_t=\bar u_{xx}+2 \bar u\, \bar u_x
\end{equation}
is related to the heat equation 
\begin{equation}\label{heat}
u_t = u_{xx}
\end{equation}
by the Cole-Hopf substitution 
\begin{equation}\label{CHopf}
\bar t=t, \qquad \bar x=x, \qquad  \bar u=\frac{u_x}{u}.
\end{equation}
This substitution admits the following group-theoretical interpretation.
The group $G$ of dilatations $u \to \tau u$ acts on the solutions of the heat equation. It is easily seen that any differential invariant of $G$ is a function of the variables
\begin{equation}\label{CHopfvar}
\bar t=t, \qquad \bar x=x, \qquad  \bar u=\frac{u_x}{u}, \qquad D\Big(\frac{u_x}{u} \Big),\quad  \dots , \quad D^i\Big(\frac{u_x}{u} \Big), \quad \dots .
\end{equation}
Since the group $G$ preserves
the equation, the (total) $t$-derivative $D_t$ sends invariants to invariants.
Consequently $\ds D_t\Big(\frac{u_x}{u}\Big)$ must be a function of the variables \eqref{CHopfvar}. Thus the
function $\ds \bar u = \frac{u_x}{u}$ must satisfy a certain evolution equation. A simple
calculation shows that this latter equation coincides with \eqref{bur}.
\end{example}
\begin{remark} The above reasoning shows that any evolution equation with the symmetry group $u \to \tau u$ admits the Cole-Hopf transformation \eqref{CHopf}.
\end{remark}
\begin{exercise} The heat equation \eqref{heat} admits the one-parameter group of shifts $x \to x+\tau$. The simplest invariants of the group are 
\begin{equation}\label{Blu}
\bar t=t, \qquad \bar x = u, \qquad \bar u = u_1.
\end{equation}
The derivation $\ds \bar D=\frac{1}{u_1}\,D$ generates a functional basis  $\bar u_i=\bar D^i(u_1)$ of all differential invariants. Therefore, 
the differential substitution \eqref{Blu} leads to an evolution equation for $\bar u$. Verify that 
$$
\bar u_t=\bar u^2 \,\bar u_{\bar x \bar x}.
$$
\end{exercise}

In the examples considered above the equation \eqref{down} is the restriction of equation \eqref{up} to the set of differential invariants of a symmetry group $G$. This equation \eqref{down} is called a {\it quotienting} of \eqref{up} by the group $G$. 

Consider evolution equations of the form 
\begin{equation}\label{u1}
u_t = F(u_1,u_2,\dots, u_n).
\end{equation}
 They admit the symmetry group $u\to u+\tau$. The invariants
$$
\bar t=t, \qquad \bar x=x, \qquad \bar u = u_1
$$
of this group define the differential substitution, which results the equation 
\begin{equation}\label{poten}
\bar u_t = D\Big(F(\bar u,\dots, \bar u_{n-1})\Big).
\end{equation}
\begin{definition}\label{pot1} Equation \eqref{poten} is said to be obtained  from the equation \eqref{u1} by the {\it potentiation}.
\end{definition}
\begin{remark} Since any one parametric group of point (contact) transformations is point (contact) equivalent to the group  $u\to u+\tau$, any quotienting by an one-parametric group is a composition of a point (contact) transformation and the potentiation. 
\end{remark}
\begin{remark} The transformation $u=\int \bar u\,dx$ from equation \eqref{poten} to equation \eqref{u1}, inverse to the potentiation, is related to the fact that equation \eqref{poten} has the conserved density $\rho=\bar u.$ If equation \eqref{eveq} posseses a conserved density of order not greater than 1, then it can be transformed to  $\rho=\bar u$ by a proper point or contact transformation. The resulting evolution equation has the form \eqref{poten} and we may apply the tranformation  $u=\int \bar u\,dx$ to obtain the corresponding equation of the form \eqref{u1}.
\end{remark}

Since one-parameter groups of contact transformations are in one-to-one correspondence \cite{Olv93} with the infinitesimal symmetries of the form 
\begin{equation}\label{symord1}
u_{\tau}=G(x, u, u_1),
\end{equation}
the existence of such a symmetry for equation \eqref{up} guarantees the presence of the quotienting equation of the form \eqref{down}, 
while the existence of a non-trivial conserved density of the form $\rho(\bar x,\bar u, \bar u_1)$ for an equation of the form \eqref{down} gives rise to the existence of the corresponding equation \eqref{up}. 

\begin{definition}\label{def15} Differential substitutions generated by infinitesimal symmetries or conservation laws of order not greater than 1 are called {\it quasi-local transformations}.
\end{definition}

\begin{question} When a quasi-local transformation preserves the higher symmetries?
\end{question}
The answer is evident: if \eqref{symord1} is a symmetry not only for equation \eqref{up} but also for the hierarchy of its symmetries, then the corresponding quotienting can be applied to the whole hierarchy.  In other words, if the group $G$ is a symmetry group for any equation from the hierarchy, then the corresponding substitution preserves the integrability in the sence of the symmetry approach. 

Similarly, if all equations from a hierarchy of symmetries for equation \eqref{down} have the same conserved density  $\rho(\bar x,\bar u, \bar u_1)$, then the corresponding equation \eqref{up} has infinitely many symmetries.

\medskip

\chapter{Symmetry approach to integrability}

The symmetry approach to the classification of integrable PDEs with two independent variables is based on the existence of the higher symmetries and/or local conservation laws (see Introduction).

\section{Description of some classification results}

\subsection{Hyperbolic equations}\label{hyphyp}

The first classification result obtained with the symmetry approach in 1979 
was formulated in Theorem \ref{ZS}. 
\begin{op} A complete classification of integrable hyperbolic equations of
the form
\begin{equation}\label{genhyp}
u_{xy}=\Psi(u, u_{x}, u_{y})
\end{equation}
is still an open problem. Some partial results were obtained in {\rm \cite{zib}}. 
\end{op}
\begin{example} The following equation \cite{BorZyk}
\begin{equation}\label{borzyk}
u_{xy}=S(u) \sqrt{\vphantom{u_y^2}1-u_x^2}\sqrt{1-u_y^2}, \qquad {\rm where} \qquad 
S''-2 S^3+c \, S=0,
\end{equation}
is integrable. 
\end{example}
\begin{example} The equation
 $$
u_{xy}=Q(u) \, b(u_x) \, \bar b(u_y), \qquad {\rm where} \qquad 
  Q''-2 Q Q'-4 Q^3=0,$$
and $b, \bar b$ are solutions of the cubic equations
$$
  (u_x-b)(b+2u_x)^2=1, \qquad
(u_y-\bar b)(\bar b+2u_y)^2=1,
 $$
is integrable. 
\end{example}
In \cite{meshsokHyp} the following problem was solved.  For hyperbolic equations \eqref{genhyp}
 the symmetry approach assumes the existence of both $x$-symmetries of the form
$$
u_{t}=A(u, u_{x}, u_{xx}, \dots, ),
$$
and  $y$-symmetries of the form
$$
u_{\tau}=B(u, u_{y}, u_{yy}, \dots, ).
$$
For example, the famous integrable sin-Gordon equation 
$$
u_{xy}=\sin{u}
$$
admits the symmetries
$$
u_{t}= u_{xxx}+\frac{1}{2} u_{x}^{3},\qquad u_{\tau}= u_{yyy}+\frac{1}{2}
u_{y}^{3}.
$$
It was assumed in \cite{meshsokHyp} that both $x$ and $y$-symmetries are third order {\it integrable} evolution equations (maybe different). 
The corresponding classification result is formulated in Appendix 1. 

Chapter \ref{Liutype} is devoted to a special class of integrable equations \eqref{genhyp} named {\it the Liouville type equations} \cite{zibsok}.

In \cite{zibshab2} integrable hyperbolic systems of the form
$$
u_x = p(u,\,v), \qquad v_y=q(u,\,v)
$$ 
were investigated.

\subsection{Evolution equations}

Some necessary   conditions for the existence of higher symmetries that  do not depend on symmetry orders were found in \cite{ibshab, sokshab} for evolution equations of the form \eqref{eveq} (see Section 2.2). These conditions lead to an overdetermined system of PDEs for the right-hand side of \eqref{eveq}. Solutions of this system are not always polynomial.

It was proved in \cite{soksvin1} that the same conditions hold if the equation \eqref{eveq} possesses infinitely many
local conservation laws. But the latter conditions are stronger than the conditions for symmetries. Moreover, there exist equations 
that have higher symmetries but have no higher conservation laws. 

\begin{definition}\label{def21} An equation is called {\it $S$-integrable} (in the terminology by F. Calogero) if it has infinitely many both 
higher symmetries and conservation laws.  An equation is called {\it $C$-integrable}  if it has infinitely many  
higher symmetries but finite number of higher conservation laws.
\end{definition}
The simplest example of $C$-integrabe equation is the Burgers equation \eqref{burgers}.
\begin{remark} Usually, the inverse scattering method can be applied to $S$-integrable equations  while $C$-integrable equations can be reduced to linear equations by differential substitutions. However, to eliminate  obvious exceptions, we need to refine Definition \ref{def21} (see Definition \ref{clar}). Otherwise, the linear equation $u_t=u_{xxx}$ will still be $S$-integrable. 

\end{remark}

\begin{proposition}\label{nocons1} {\rm (see Proposition \ref{nocon} and Theorem 29 in \cite{int2})} A scalar evolution equation \eqref{eveq} of even order $n =  2 k$ cannot 
possess  infinitely many  higher local conservation laws.
\end{proposition}

There are two types of classification results obtained in the frame of the symmetry approach:  a ``weak'' version, where equations 
with conservation laws are listed and a ``strong'' version related to symmetries.  The ``weak'' list contains $S$-integrable equations while 
the ``strong'' list consists of both $S$-integrable and $C$-integrable equations. 

\subsubsection{Second order equations}
All nonlinear integrable equations of the form
$$
u_t=F(x,\,t,\, u,\,u_1,\,u_2)
$$
were listed in \cite{svin4} and \cite{soksvin}.
The answer is:
$$
\begin{array}{l}
u_t=u_2+2 u u_x+h(x), \\[3mm]
u_t=u^2 u_2-\lambda x u_1+\lambda u, \\[3mm]
u_t=u^2 u_2+\lambda u^2, \\[3mm]
u_t=u^2 u_2-\lambda x^2 u_1+3 \lambda x u.  \\[3mm]
\end{array}
$$
This list is complete up to contact transformations.  According to Proposition \ref{nocons1} all these equations are $C$-integrable. 
They are related  \cite{SvSok26} to the heat equation $v_t=v_{xx}$ 
by group differential substitutions discussed in Section 1.3.2.

The first three equations possess local higher symmetries and form a list of integrable equations 
of the form 
\begin{equation}\label{svinsecond}
u_t=F(x, u,\,u_1,\,u_2)
\end{equation}
obtained in \cite{svin4}. 
\begin{remark} It turns out that any integrable equation  \eqref{svinsecond} has the form
$$
u_t=\frac{a_1 u_2+a_2}{a_3 u_2+a_4}, \qquad a_i=a_i(x,u,u_1). 
$$
Any such equation can be reduced to a quasi-linear form 
$$
u_t=a_1(x,u,u_1)\, u_2+a_2(x,u,u_1)
$$
by a proper contact transformation \eqref{conttran}.
\end{remark}

In the paper \cite{soksvin} equations with weekly non-local
symmetries were considered. The weekly non-locality of symmetries looks very natural if we assume then the right-hand side of the equation depends on $t$. As a result, the list from the paper  \cite{svin4} was extended by the fourth equation. 

\subsubsection{Third order equations}

A first result of the ``weak'' type for equations \eqref{eveq} is the following:
\begin{theorem} {\rm \cite{soksvin1}} A complete list up to quasi-local transformations {\rm (see Definition \ref{def15})}
of equations of the form 
\begin{equation} \label{vvs}
u_{t}=u_{xxx}+f(u, u_x, u_{xx})
\end{equation}
with an infinite hierarchy of conservation laws can be
written as:
$$
u_t=u_{xxx}+u \,u_{x}, 
$$
$$
u_t=u_{xxx}+u^2\, u_{x},
$$
$$
u_t=u_{xxx}-\frac{1}{2}u_{x}^3+(\alpha e^{2u}+\beta
e^{-2u})u_{x}, 
$$
\begin{equation}\label{CD2}
u_{t}=u_{xxx}-\frac{1}{2} P'' \,u_x + \frac{3}{8} 
\frac{(P-u_x^2)_{x}^2}{u_x\,(P-u_x^2)}, 
\end{equation}
\begin{equation}\label{KN}
u_t=u_{xxx}-\frac{3}{2}\,\frac{u^{2}_{xx}+P}{u_{x}},
\end{equation}
where $P'''''(u)=0$.  
\end{theorem}
For the ``strong'' version see \cite{soksvin2} and Appendix 3. 
A proof of the corresponding statement can be found in the survey \cite{meshsok}. 

For integrable third order equations 
\begin{equation}\label{genthird}
u_t=F(u,u_x,u_{xx}, u_{xxx}),
\end{equation}
more general than \eqref{vvs}, there are three possible types of $u_{xxx}$-dependence  \cite{MikShaSok91}:
\begin{itemize}
\item[1)]
$$
u_t=a \,u_{xxx}+b,
$$
\item[2)] 
$$
u_t=\frac{a}{(u_{xxx}+b)^2},
$$
and 
\item[3)]
$$
u_t=\frac{2 a\, u_{xxx}+b}{\sqrt{a \,u_{xxx}^2+b\, u_{xxx}+c}}+d,
$$
\end{itemize}
where the functions $a,b,c$ and $d$ depend on $u,u_x,u_{xx}.$ 
A complete classification of integrable equations of such type is not finished yet  \cite{hss,her}. 
\begin{conjecture}
All integrable third order equations \eqref{genthird} are related to the KdV equation or to the Krichever-Novikov equation \eqref{KN} by differential substitutions  of Cole-Hopf and Miura type {\rm \cite{cvsokyam}}.
\end{conjecture}
 
\subsubsection{Fifth order equations}
All equations of the form
$$
u_t=u_5+F(u, u_x, u_2, u_3, u_4),
$$
possessing higher conservation laws were found in \cite{DSS}. 
\begin{lemma} A possible dependence of the function $F$ on $u_4,u_3$ and $u_2$ for such equations is 
described by the following formula:
$$
u_t=u_5+A_1\,u_2u_4+A_2\,u_4+A_3\,u_3^2+(A_4\,u_2^2+A_5\,u_2+A_6)\,u_3+A_7\,u_2^4+A_8\,u_2^3+A_9\,u_2^2+A_{10}\,u_2+A_{11},
$$
where $A_i=A_i(u,\,u_x)$. 

\end{lemma}

The list of integrable cases contains both  well-known equations  \eqref{sk}, \eqref{kk}  and 
$$
u_t=u_{5}+5 (u_1-u^2) u_{3}+5 u_{2}^2-20 u u_1 u_{2}  -5 u_1^3+5 u^4 u_1
$$
(see \cite{SokSha80})
as well as several new equations. One of them 
(missed in \cite{FujWat83}) is  given by
$$
\begin{array}{rcl}
\ds u_t&=&u_{5}+5 (u_{2}-u_1^2+\lambda_1 e^{2u}-
\lambda_2^2 e^{-4u})\,u_3 -5 u_1 u_{2}^2+15 (\lambda_1 e^{2u}\, u_{3}+ 
 4 \lambda_2^2 e^{-4u})\, u_1
u_{2} \\[3mm]
& &\ds + \,u_1^5 -90 \lambda_2^2 e^{-4u}\, u_1^3+ 5(\lambda_1 e^{2u}-\lambda_2^2
e^{-4u})^2\, u_1.
\end{array}
$$
The ``strong'' version of this classification result  (see Appendix 3) was published in  \cite{meshsok}.

At first glance, the problem of the classification of integrable equations \begin{equation}
u_t=u_n+F(u, \, u_x, \, u_{xx}, \dots, u_{n-1}), \qquad u_i=\frac{\partial^i u}{\partial x^i}. \label{scalar}
\end{equation} 
with arbitrary $n$ seems to be far from a conclusive solution. This is not quite so. Each integrable equation together with all its symmetries form a so-called hierarchy of integrable equations. For the equations integrable by the inverse scattering method all the equations of the hierarchy possess the same $L$-operator. This fact lies in the basis of the commutativity of the equations in hierarchies (each equation of the hierarchy is a higher symmetry for all others). For another explanation of the commutativity see Remark \ref{rem12}. A general rigorous statement on ``almost'' commutativity of the symmetries for   equation \eqref{scalar} can be found in \cite{sok1}.

Under the assumption that the right-hand side of   equation \eqref{scalar} is polynomial and homogeneous, it was proved in the works
\cite{sw,ow} that the hierarchy of any such  equation contains an equation of second, third, or fifth order.

More references for the classification of scalar evolution equations can be found in the reviews \cite{sokshab, MikShaYam87, MikShaSok91, int2, meshsok}.
Here, I would like also to mention the papers \cite{CheLeeLiu79}--\cite{AbeGal83}.

\subsection{Systems of two equations}

In \cite{MikSha85,MikSha86} the necessary  conditions of integrability were generalized to the case of systems of evolution equations. However, component-wise computations in this case are very tedious. The only one serious classification problem has been solved \cite{MikSha85,MikSha86,MikShaYam87}: all systems of the form
\begin{equation}\label{sys2}
u_t=u_{2}+F(u,\,v,\,u_1,\,v_1), \, \qquad v_t=-v_{2}+G(u,\,v,\,u_1,\,v_1)
\end{equation}
possessing higher conservation laws were listed.  In other words, the authors 
have found all $S$-integrable systems \eqref{sys2}.

Besides the well-known NLS equation written as a system of two equations
\begin{equation}\label{NLS}
u_t=-u_{xx}+2 u^2 v,\,  \qquad v_t=v_{xx}-2 v^2 u,
\end{equation}
basic integrable models from a long list of such integrable models are: 
\begin{itemize}
\item a version of the Boussinesq equation
$$
u_t=u_{2}+(u+v)^{2},\, \qquad v_t=-v_{2}+(u+v)^{2};
$$
\item and the two-component form of the Landau-Lifshitz equation
$$
\begin{cases}
\displaystyle u_t=u_{2}-\frac{2 u_1^2}{u+v}-\frac{4\,(p(u,v)\, u_1+r(u)\, v_1)}
{(u+v)^2},
\\[4mm]
\displaystyle v_t=-v_{2}+\frac{2 v_1^2}{u+v}-\frac{4\,(p(u,v)\, v_1+r(-v)\, u_1)}
{(u+v)^2},
\end{cases}
$$
where $r(y)=c_4 y^4+c_3 y^3+c_2 y^2+c_1 y+c_0$ and
$$
p(u,v)=2 c_4 u^2 v^2+c_3 (u v^2-v u^2)-2 c_2 u v+c_1(u-v)+2 c_0.
$$
\end{itemize}
A complete list of integrable systems \eqref{sys2} up to transformations 
$$
u \to \Phi(u), \qquad v\to \Psi(v)
$$
should contain more than 100 systems. Such a list has never been 
published. Instead, in \cite{MikShaYam87} was presented a list complete up to ``almost invertible'' transformations \cite{MSY}.

\begin{remark}
All these equations have a fourth order symmetry of the form 
\begin{equation}
\label{kvazsym}
 \begin{cases}
u_{\tau}=u_{xxxx}+f(u, v, u_x, v_x, u_{xx}, v_{xx}, u_{xxx}, v_{xxx}), \\[1.5mm]
v_{\tau}=-v_{xxxx}+g(u, v, u_x, v_x, u_{xx}, v_{xx}, u_{xxx}, v_{xxx}).
\end{cases}
\end{equation}
Some of them have also a third order symmetry.
\end{remark}

There are at least three reasons why the paper \cite{SokWol99} on integrable systems of the form 
\begin{equation}
\label{kvazgen}
 \begin{cases}
u_t = u_{xx} + A_{1}(u,v)\, u_x + A_{2}(u,v)\, v_x + A_{0}(u,v), \\[1.5mm]
v_t = - v_{xx} + B_{1}(u,v)\, v_x + B_{2}(u,v)\, u_x + B_{0}(u,v) 
 \end{cases}
\end{equation}
was written. 

{\it Reason 1}. There are $C$-integrable cases that are not in the Mikhailov-Shabat-Yamilov classification. One of the examples is the system
\begin{equation} \label{borov} \begin{cases}
u_t =   u_{xx} - 2 u u_x  - 2 v u_x  - 2 u v_x + 2 u^2  v +  2 u v^2,   \\[1.5mm]

v_t   =  - v_{xx} + 2 v u_x + 2 u v_x  + 2 v v_x - 2 u^2  v - 2 u v^2.
\end{cases}
\end{equation}
The system \eqref{borov} was first discussed 
in \cite{BPR}. 
It can be reduced to
$$
U_t =   U_{xx}, \qquad
V_t   =  - V_{xx}
$$
by the following substitution of Cole-Hopf type:
$$
\label{borsub}
u=\frac{U_x}{(U+V)},
\qquad v=\frac{V_x}{(U+V)}.
$$

{\it Reason 2}. Systems \eqref{kvazgen} can be easily classified without any equivalence relations. The right-hand sides of such systems turn out to be polynomial.  

{\it Reason 3}. Results of any serious classification problem should be verified independently. Only after that there is an assurance 
that nothing has been missed.

For classification of systems \eqref{kvazgen} the simplest naive version of the symmetry test was applied.

\begin{lemma} If the system \eqref{kvazgen} has a fourth order symmetry \eqref{kvazsym}, 
then the system is of the following form:
\begin{equation}
\nonumber
 \begin{cases}
u_t = u_{xx} + (a_{12} u v + a_1 u + a_2 v + a_0)\, u_x + (p_2 v +
p_{11} u^2 + p_1 u + p_0)\, v_x + A_{0}(u,v),  \\[1.5mm]
v_t=- v_{xx} + (b_{12} u v + b_1 v + b_2 u + b_0)\, v_x + (q_2 u +
q_{11} v^2 + q_1 v + q_0) \,u_x + B_{0}(u,v), 
\end{cases}
\end{equation}
where $A_0$ and $B_0$ are polynomials of at most fifth degree.
\end{lemma}

The coefficients of the latter system satisfy an overdetermined
system of algebraic equations. The most essential equations are
\begin{equation}
\nonumber
\begin{array}{l}
p_2 (b_{12}-q_{11})=0, \qquad p_2 (a_{12}-p_{11})=0,
\qquad p_2 (a_{12}+2 b_{12})=0, \\[2mm] q_2 (b_{12}-q_{11})=0,
\qquad q_2 (a_{12}-p_{11})=0,
\qquad q_2 (b_{12}+2 a_{12})=0,
\end{array}
\end{equation}
\begin{equation}
\nonumber
\begin{array}{l}
a_{12} (a_{12}-b_{12}+q_{11}-p_{11})=0, \qquad \ b_{12}
(a_{12}-b_{12}+q_{11}-p_{11})=0, \\[2mm]
(a_{12}-p_{11})(p_{11}-q_{11})=0, \qquad \qquad
(b_{12}-q_{11})(p_{11}-q_{11})=0, \\[2mm]
(a_{12}-p_{11})(a_{12}-b_{12})=0, \qquad
\qquad (b_{12}-q_{11}) (a_{12}-b_{12})=0.
\end{array}
\end{equation}
As usual, such factorized equations lead to a tree of variants.

Solving the overdetermined system, we do not consider the so called
triangular systems like the following:
$$
u_t=u_{xx}+2 u v_x, \qquad v_t=-v_{xx}-2 v v_x.
$$
Here, the second equation is separated
and the first is linear with the
variable coefficients defined by a given solution of the second equation.

The classification statement is formulated in Appendix 4. 
\begin{op} Find all systems \eqref{sys2} that have infinitely many symmetries.
\end{op}

\section{Integrability conditions} 
\label{SectionIntegrabilityConditions}

For our aims the language of differential algebra \cite{kapl} is the most adequate one.  

\subsection{Evolutionary vector fields,  Fr\'echet and Euler derivatives}

Consider evolution equations of the form~\eqref{eveq}.  
Suppose that the right-hand side of~\eqref{eveq} as well as other functions in $u,u_x,u_{xx}, \dots$ belong to a differential field 
$ {\cal F}$. For our considerations one can assume that elements of $ {\cal F}$ are rational functions in a finite number of independent variables $$u_i=\frac{\partial^i u}{\partial x^i}.$$ When we are going to integrate a function with respect to one of it's arguments or to take  a radical of it, we have to extend the basic field $\cal F$. 
 
As usual in differential algebra, we have a principle derivation \eqref{DD},
 which generates all independent variables $u_i$ starting from $u_0=u$. This derivation is a formalization of the partial $x$-derivative, 
which acts on functions of the form $\ds g\Big(u(x), \frac{\partial u}{\partial x}, \dots\Big)$.  The vector field $D$ defined by  \eqref{DD} is called the {\it total} $x$-{\it derivative}.
\begin{remark}\label{rem23}
We very often use the fact that $D(f)=0,\, f\in {\cal F}$ implies $f={\rm const}$. 
\end{remark}
\begin{remark}
The variable $t$ in the local algebraic theory of evolution equations is considered as a parameter. 
\end{remark}

Amongst the main concepts related to dynamical system of ODEs \eqref{dynsys}, the finite-dimensional vector field  \eqref{ODEfield} plays a key role. 
 Analogously, the infinite-dimensional vector field   
\begin{equation}
\label{Dt} D_F= \sum_{i=0}^\infty D^{i}(F) \frac{\partial}{\partial u_i}  
\end{equation}
is associated with evolution equation \eqref{eveq}. 
This vector field commutes with $D$. We call vector fields of the form \eqref{Dt} {\sl evolutionary}.
The function $F$ is called the {\it generator} of that 
evolutionary vector field. Sometimes we call \eqref{Dt} the {\it total} $t$-{\it derivative}
with respect to \eqref{eveq} and denote it by $D_t$. 

The set of all evolutionary vector fields forms a Lie algebra over $\C$: $[D_G,\, D_H]=D_K,$ where
\begin{equation}\label{comev}
K=H_*(G)-G_*(H).
\end{equation}
Henceforth, we use the following  
\begin{definition}\label{dfrechet} For any function $a\in {\cal F}$ the Fr\'echet
derivative is defined as the linear differential operator  
$$
a_*=\sum_k\frac{\partial a}{\partial u_k}D^k \, . 
$$
The order of the function $a$ is defined as the order of the differential
operator $a_*$. We denote by $a_*^+$ the formally
conjugate operator
\[ a_*^+=\sum_k (-1)^k D^k \circ \frac{\partial a}{\partial u_k} \, . \]
\end{definition}

In the Introduction we defined a generalized symmetry of equation \eqref{eveq} as an evolution equation~\eqref{evsym}\footnote{For the sake of brevity we often name generator $G$ of symmetry simply by symmetry $G$.} that is compatible with~\eqref{eveq}. By definition, the compatibility means that $[D_F,\, D_G]=0.$ It can be written also in the form $G_*(F)=F_*(G)$ or 
\begin{equation}\label{DFsym}
D_t(G)-F_*(G)=0.
\end{equation}

Formula \eqref{comev} defines a Lie bracket on our differential field ${\cal F}$. An integrable hierarchy is nothing else but an infinite-dimensional commutative subalgebra of this Lie algebra. 

\begin{proposition}\label{defrec1} Suppose an operator\footnote{Usually ${\cal R}$ is a differential operator or a ratio of differential operators.} ${\cal R}$ satisfies the equation 
\begin{equation}\label{defR}
D_t({\cal R})=F_*\, {\cal R}-{\cal R}\,F_*.
\end{equation}
Then, for any symmetry \eqref{evsym} of the equation \eqref{eveq}, the equation $u_{\tau}={\cal R}(G)$ is a symmetry of  \eqref{eveq}.
\end{proposition}
Hereinafter,
$$
D_t\Big(\sum s_i D^i\Big) = \sum D_t(s_i)\, D^i.
$$
\begin{proof}
Let us rewrite \eqref{defR} as 
\begin{equation}\label{defR1}
[D_t - F_*, \, {\cal R}]=0.
\end{equation}
Now the statement follows from \eqref{DFsym}.  
\end{proof}
\begin{definition} An operator ${\cal R}: {\cal F}\to {\cal F}$ satisfying \eqref{defR} is called a {\it recursion operator} for the equation \eqref{eveq}.

\end{definition}

\subsubsection{Euler operator}
\begin{definition} The Euler operator or the variational derivative
of a function $\ a\in {\cal F}$ is defined as
\[
\frac{\delta a}{\delta u}=\sum_k (-1)^k D^k \left(\frac{\partial a}
{\partial u_k}\right)=a_*^+(1)\, .
\]
\end{definition}

If a function $a$ is a total derivative $a=D(b),\, b\in {\cal F}$ (we say that $a\in {\rm Im}\, D$), then
the variational derivative vanishes. Moreover, the vanishing of the variational
derivative is almost a criterion that the function belongs to $ {\rm Im}\, D$
\cite{GMSh}:

\begin{theorem}\label{tgms} For $a\in {\cal F}$ the variational
derivative vanishes
  \[\frac{\delta a}{\delta u}=0\, \]
  if and only if  $a\in  {\rm Im}\, D +  \C$.
\end{theorem}

\begin{lemma}\label{Lemvar} The following identities
\begin{eqnarray*}
&&(a b)_*=a b_*+b a_*, \qquad \qquad (D(a))_*=D \, a_*=D(a_*)+a_*\, D, \\[1.5mm]
&&(D_t(a))_*=D_t(a_*)+a_*\, F_*,  \qquad \qquad(a_*(b))_*=D_b(a_*)+a_*\,  b_*, \\[1.5mm]
&&\Big(\frac{\delta a}{\delta u}\Big)_*=\Big(\frac{\delta a}{\delta u}\Big)^+_*,
\qquad \qquad
\frac{\delta }{\delta u}(D_t (a))=D_t\Big(\frac{\delta a}{\delta u}\Big)+
F_*^+\Big(\frac{\delta a}{\delta u}\Big)
\end{eqnarray*}
hold for any $a,b,F\in {\cal F}$.
\end{lemma}

\subsection{Pseudo-differential series}

Consider a skew field of (non-commutative) formal series of the form
\begin{equation}\label{serA}
 S=s_{m}D^m+s_{m-1}D^{m-1}+\cdots + s_0+s_{-1}D^{-1}+
 s_{-2}D^{-2}+\cdots , 
\qquad s_i\in {\cal F}\, .
\end{equation}
The number $m\in \Z$ is called the {\it order} of $S$ and is denoted by ${\rm ord}\, S$.  If $s_i=0$ for $i<0$ that $S$ is called a {\it differential operator}.
\medskip

The product of two formal series is defined by the formula 
$$
 D^k\circ s D^m =s \, D^{m+k}+C_k^1 D(s)\,D^{k+m-1} + 
 C_k^2 D^2 (s)\,D^{k+m-2}+
\cdots \, ,
$$
where $k,m\in \mathbb Z$ and $C^j_n$ is the binomial coefficient
\[
C^j_n=\frac{n(n-1)(n-2)\cdots(n-j+1)}{j!},\qquad n\in \Z.
\]
For the series this formula is extended by associativity.

\begin{remark} \label{remcom} For any series $S$ and $T$ we have ${\rm ord (S\circ T -T\circ S)}\le {\rm ord}\,S + {\rm ord}\,T -1.$
\end{remark}

The formally conjugated formal series $S^+$ is defined as
\[
 S^+=(-1)^m D^m\circ\, s_{m}+(-1)^{m-1}D^{m-1}\circ\, s_{m-1}+\cdots +
 s_0-D^{-1}\circ\,s_{-1}+D^{-2}\circ\, s_{-2}+\cdots\,.
\]

\begin{example}
Let
$$
R=u D^2+u_1 D,\qquad S=-u_1 D^3,\qquad T=u D^{-1};
$$
then
\begin{gather*}
R^+=D^2\circ u-D\circ u_1=R,\\[2mm]
S^+=D^3\circ u_1=u_1D^3+3u_2D^2+3u_3D+u_4,\\[2mm]
T^+=-D^{-1}u=-u D^{-1}+u_1D^{-2}-u_2D^{-3}+\cdots \,.
\end{gather*}
\end{example}

For any series \eqref{serA}
one can uniquely find the inverse series
$$
T=t_{-m}D^{-m}+t_{-m-1}D^{-m-1}+\cdots\, ,\qquad t_k\in  {\cal F}
$$
such that $S\circ T=T\circ S=1$.
Indeed, multiplying $S$ and $T$ and equating the result to 1, we find that $s_m t_{-m}=1$, i.~e., $\ds t_{-m}=\frac{1}{s_m}$. Comparing the coefficients of $D^{-1},$ we get
$$
m s_m \, D(t_{-m})+s_m \,t_{-m-1}+s_{m-1}\,t_{-m}=0
$$
and therefore
$$
t_{-m-1}=-\frac{s_{m-1}}{s_m^2}-m D\Big(\frac{1}{s_m}\Big) \, ,\quad \mbox{etc.}
$$

Furthermore, we can find the $m$-th root of the series $S$, i.~e., a series
$$
R=r_1 D+r_0+r_{-1}D^{-1}+r_{-2}D^{-2}+\cdots
$$
such that $R^m=S$.
This root is unique up to any number factor $\varepsilon$ such that  $\varepsilon^m=1$.

\begin{example}\label{ex10}
Let $S=D^2+u$. Assuming
$$
R=r_1 D+r_0+r_{-1}D^{-1}+r_{-2}D^{-2}+\cdots ,
$$
we compute
$$
R^2=R\circ R= r_1^2 D^2+(r_1 D(r_1)+r_1 r_0 +r_0 r_1)\,D+r_1D(r_0)+r_0^2+r_1r_{-1}+r_{-1}r_1+\cdots\, ,
$$
and compare the result with $S$. From the coefficients of $D^2$ we find $r_1^2=1$ or
$r_1=\pm 1$. Let  $r_1=1$. Comparing coefficients of $D,$ we get $2r_0=0$, i.~e., $r_0=0$.
From $D^{0}$ we obtain $2r_{-1}=u$, terms of $D^{-1}$
$\ds r_{-2}=-\frac{u_1}{4}$, etc., i.~e.
$$
R=S^{1/2}=D+\frac{u}{2}D^{-1}-\frac{u_1}{4}D^{-2}+\cdots\, .
$$
\end{example}

\begin{definition}
The {\it residue} of a formal series  \eqref{serA} by definition is
the coefficient of $D^{-1}$:
$$
 \hbox{res}\, (S) \stackrel{def}{=} s_{-1}\, .
$$
The {\it logarithmic residue} of $S$ is defined as
$$
 \hbox{res} \log S \stackrel{def}{=} \frac{s_{m-1}}{s_m}\, .
$$
\end{definition}

We will use the following important  
\begin{theorem} \label{adler} {\rm \cite{adler}}
For any two formal series $S$, $T$ the
residue of the commutator belongs to ${\rm Im}\,D$:
$$
{\rm res} [S,\,T]=D(\sigma (S,\,T)),
$$
where
$$
\sigma (S,\,T)=\sum_{i\le {\rm ord}(T),\ j\le
{\rm ord}(S)}^{i+j+1>0}C^{i+j+1}_{j}\, \times  
\sum_{k=0}^{i+j}(-1)^k D^k(s_j)D^{i+j-k}(t_j)\, .
$$
\end{theorem}
 
 \begin{corollary}\label{cor21} For any series $S$ and $T$
 $$
 {\rm res}\,(S - T S T^{-1}) \in {\rm Im}\,D.
 $$
 \end{corollary}
 \begin{proof} It follows from the identity
 $$
 S - T S T^{-1} = [S T^{-1},\, T].
 $$
 \end{proof}

\subsection{Formal symmetries} 

\begin{definition} A pseudo-differential series
$$
\Lambda=l_{1} D+l_{0}+l_{-1}D^{-1}+\cdots \, ,
$$
where $l_{k}=l_k(u, \ldots, u_{s_k})\in {\cal F},$ is called a {\it formal symmetry}  (or {\it formal recursion
operator}\footnote{According to Proposition \ref{defrec1} any genuine operator that satisfies \eqref{Lambdaeq}  maps higher symmetries of the equation \eqref{eveq} to higher symmetries. }{\rm )} for equation  \eqref{eveq} if $R=\Lambda$ satisfies the equation
\begin{equation}\label{Lambdaeq}
D_t(R)=[F_{*},\,R], \qquad \mbox{where} \qquad F_{*}=\sum_{i=0}^{n}
\frac{\partial F}{\partial u_{i}}D^{i}.
\end{equation}
\end{definition}
\begin{proposition}\label{Lambdaspace} {\rm \cite{sokshab}}
Suppose a pseudo-differential series $R$ of order $k$ satisfies the equation  \eqref{Lambdaeq}. Then

1) $R^{\frac{1}{k}}$ is a formal symmetry;

2) If $R_1$ and $R_2$ satisfy  \eqref{Lambdaeq}, then   $R_1\circ R_2$ satisfies \eqref{Lambdaeq};

3) $R^{\frac{i}{k}}$ 
satisfies this equation for any $i\in \Z$;

4) Let $\Lambda$ be a formal symmetry.  Then $R$ can be written in the form
$$
R=\sum_{-\infty}^k a_i \Lambda^i, \qquad  k={\rm ord}\,R, \quad a_i\in \C; 
$$ 

5) In particular, any formal symmetry $\bar \Lambda$ has the form 
\begin{equation}\label{genLam}
\bar \Lambda = \sum_{-\infty}^1 c_i \Lambda^i, \qquad c_i\in \C.
\end{equation}
\end{proposition}

The coefficients of the formal symmetry can be found from equation
\eqref{Lambdaeq}.
\begin{example}
Let us consider equations of KdV type
\begin{equation}\label{kdvt}
u_{t}=u_{3}+f(u, u_1)\,
\end{equation}
and find a few coefficients $l_1,l_0,...$ of the
formal symmetry $\Lambda$.
We substitute
$$
F_{*}=D^3+\frac{\partial f}{\partial u_{1}} D+
\frac{\partial f}{\partial u},\qquad \quad
\Lambda =l_{1} D+l_{0}+l_{-1}D^{-1}+\cdots
$$
into \eqref{Lambdaeq} and collect the coefficients of $D^3, D^2, \dots$.

We obtain
$$
\begin{array}{l}
D^3:\quad 3 D(l_1)=0; \qquad D^2:\quad 3 D^2(l_1)+3 D(l_0)=0; \\[4mm]
\ds D: \quad D^3(l_1)+3 D^2(l_0)+3 D(l_{-1})+
\frac{\partial f}{\partial u_{1}}\, D(l_1)=
 D_t(l_1)+l_1\, D\left(\frac{\partial f}{\partial u_{1}}\right).
\end{array}
$$
From the first equation it follows (see Remark \ref{rem23}) that $l_1$ is a constant and we set $l_1=1$.
Now, from the second equation, it follows that $l_0$ is a constant and we choose
$l_0=0$
(any constant is a trivial solution of equation \eqref{Lambdaeq}).
It follows from the third equation that
\[ D(l_{-1})=D\Big(\frac{1}{3}\, \frac{\partial f}{\partial u_{1}}\Big)\, ,\]
and therefore
\[ l_{-1}=\frac{1}{3}\, \frac{\partial f}{\partial u_{1}}+c_{-1}\, ,
\qquad c_{-1}\in\C\, .\]
The constant of integration $c_{-1}$ can be set equal to zero without
loss of generality (see formula \eqref{genLam}). Therefore
$$
\Lambda=D+\frac{1}{3}\,
\frac{\partial f}{\partial u_{1}}D^{-1}+\cdots \, .
$$
Notice that a first obstacle for the existence of $\Lambda$ appears when we compare coefficients of $D^{-1}$. So  the formal symmetry does not exist for any arbitrary function $f(u,u_1)$ in \eqref{kdvt}.
\end{example}

\begin{remark}\label{rem24} In general, we define the coefficients of $\Lambda$ from \eqref{Lambdaeq} step by step solving equations of the form 
$D(l_k)=S_k$, where $S_k \in {\cal F}.$ This equation is resolvable only if $S_k \in {\rm Im}\, D$ {\rm (see Theorem \ref{tgms})}. So there 
are infinitely many obstacles for the existence of a formal symmetry. 

\end{remark}

\begin{theorem}\label{tlsym} {\rm  \cite{ibshab}} If equation \eqref{eveq}
possesses an infinite sequence of higher symmetries
$$
u_{\tau_{i}}=G_i(u, \dots ,u_{m_{i}}), \qquad m_i \rightarrow \infty,
$$
then the equation has a formal symmetry.
\end{theorem}
\begin{proof}

The main idea of the proof of Theorem \ref{tlsym} and the relation between the
structure of the formal symmetry and higher symmetries can be illustrated by
the following consideration \cite{sokshab}. Suppose that the equation \eqref{eveq}  has one symmetry
with a generator $G$. The function $G$ satisfies the equation \eqref{DFsym}.
Let us compute the Fr\'echet derivative from the left-hand side of this equation.
Using identities of Lemma \ref{Lemvar},  
we retrieve the equation
\[
D_t(G_*)+G_* F_* = D_G (F_*)+F_* G_*,\]
which can be rearranged in the form
\begin{equation}\label{GF}
D_t(G_*)-[F_*,G_*]=D_G(F_*)\,.
\end{equation}
If the symmetry $G$ has a very large order $m$, then the order of left-hand side of \eqref{GF} is much greater than the order of right-hand side. Therefore, the relations for several first coefficients of $G_*$ are exactly the same as for the first coefficients of the formal symmetry $\Lambda^m.$ The first coefficients of a series of order $m$ that satisfies \eqref{Lambdaeq} coincide with the coefficients of the series $G_*^{\frac{1}{m}}$, which belong to ${\cal F}$. 
\end{proof}

\subsection{Conservation laws}

The notion of first integrals, in contrast with infinitesimal symmetries, cannot be generalized to
the case of PDEs. It is replaced by the concept of local conservation laws.

\begin{definition} A function $\rho\in {\cal F}$ is called a {\it density} of a
local conservation law of equation  \eqref{eveq}  if there exists a function
$\sigma\in {\cal F}$ such that
\begin{equation}\label{rho}
D_t(\rho)=D(\sigma).
\end{equation}
\end{definition}
Equation \eqref{rho} is evidently satisfied if $\rho=D(h)$ for any $h\in {\cal F}$.
In this case $\sigma=D_t(h)$. Such ``conservation laws'' are called {\it trivial}.

\begin{definition}\label{def28} Two conserved densities $\rho_1, \rho_2$ are called
equivalent $\rho_1\sim \rho_2$ if the difference $\rho_1-\rho_2$ is a trivial
density (i.e. $\rho_1-\rho_2\in {\rm Im}\,D$).
\end{definition}

\begin{remark}\label{equi} Actually, a conserved density can be regarded as an equivalence class 
with respect to $\sim$. The Euler operator is well-defined on the equivalence classes. 
\end{remark}

\begin{lemma} Any density $\rho$ is equivalent to a density $\bar \rho(u,\dots,u_k)$ such that 
$$
\frac{\partial^2 \bar\rho}{\partial u_k^2} \ne 0.
$$
\end{lemma}
\begin{proof} If $\rho$ is linear in the highest derivative, then we may reduce the differential order of $\rho$ by subtracting an element from ${\rm Im}\,D.$
\end{proof}

\begin{definition}\label{orderrho} The number $k$ is called the {\rm order} of the conserved density $\rho.$  We denote $k$ by  ${\rm ord}\, (\rho).$
\end{definition}
\begin{remark}\label{rem26} It can be easily verified that the order of the differential operator $\ds \Big(\frac{\delta \rho}{\delta u}\Big)_* $ is equal to $\,2\, {\rm ord}\, (\rho)$.
\end{remark}

\begin{question} Suppose that equation \eqref{eveq} and a function $\rho \in {\cal F} $ are given. How to verify whether a solution $\sigma \in {\cal F}$ for \eqref{rho} exists or not ?
\end{question}
The first way is to apply the Euler operator to relation \eqref{rho}. By Theorem \ref{tgms} we obtain 
\begin{equation}\label{delro}
\frac{\delta}{\delta u} \Big(D_t(\rho)\Big)=0.
\end{equation}
The left-hand side of the latter identity is known. 

Unfortunately, we cannot find the function $\sigma$ by this method.  However, there is the following straightforward  inductive algorithm to do that. The left-hand side of \eqref{rho} is given and the problem is 
\begin{question}
How to solve an equation of the form $D(X)=S$ for given $S(u,\cdots, u_{m})?$  
\end{question}
It follows from \eqref{DD} that $S$ must be linear in the highest derivative:
$$S=A(u,u_1,\dots, u_{m-1})\, u_{m}+B(u,u_1,\dots, u_{m-1}).$$ If this is not true, then the equation has no solution. If this is true,  we can reduce the order of $S$ subtracting from both sides of the equation a function of the form $D(r(u,u_1,\dots, u_{m-1}))$ such that $S-D(r)$ has order less than $m$. For the function $r$ we can take any solution of the equation $\ds \frac{\partial r}{\partial u_{m-1}}=A.$ Thus if $X$ exists, it can be found in quadratures.

\subsection{Formal symplectic operators}

According to Lemma \ref{Lemvar},   
$\ds X=\frac{\delta \rho}{\delta u}$ satisfies the conjugate equation of \eqref{DFsym}:
\begin{equation}\label{varR}
D_t\,(X)+F_*^+\,(X)=0\, .
\end{equation}
\begin{definition}\label{cosym}
Any solution $X \in {\cal F}$ of equation \eqref{varR} is called {\it cosymmetry}.
\end{definition}
\begin{definition} A pseudo-differential series
$$
S=s_1 D + s_0+s_{-1} D^{-1} +\cdots\, ,\qquad s_1\ne 0
$$
is called a {\it formal symplectic operator}\footnote{Relation \eqref{Req} can be rewritten as $(D_t+F_{*}^{+})\circ S=S (D_t-F_{*}).$ This means that a genuine operator $S: {\cal F}\to {\cal F}$ maps symmetries to cosymmetries. If the equation \eqref{eveq} is a Hamiltonian one, then the
symplectic operator, which is inverse to the Hamiltonian operator, satisfies the equation
\eqref{Req} \cite{Dorf}.}   for equation \eqref{eveq}  if it
satisfies the equation
\begin{equation}\label{Req}
D_t(S)+S\, F_*+F_*^+\, S=0\, .
\end{equation}
\end{definition}
\begin{remark} If $S$ is a formal symplectic operator, then $S^+$  is a formal symplectic operator as well. Therefore, we may additionaly assume that $S^+=-S.$
\end{remark}
\begin{lemma}
 The ratio
$R=S_{1}^{-1} S_{2}$ of any two series $S_1$ and $S_2$ that satisfy \eqref{Req} satisfies equation
\eqref{Lambdaeq}.
\end{lemma}

\begin{theorem} \label{svsok} {\rm \cite{soksvin1}} If the equation $u_t=F$ possesses
an infinite sequence of local conservation laws
$$
D_t\Big(\rho_i(u, \dots, u_{m_{i}})\Big)=D(\sigma_i), \qquad \frac{\partial^2
\rho_i}{\partial u_{m_{i}}^2}\ne 0, \quad m_i \rightarrow \infty,
$$
then the equation has a formal symmetry $\Lambda$ and a formal symplectic operator $R$ such that $S^+=-S$ and
\begin{equation}\label{LS}
\Lambda^+=-S \Lambda S^{-1}.
\end{equation}
\end{theorem} 
\begin{proof} For the full proof see \cite{sokshab}. Here, we mention only that the relation between \eqref{Req} and \eqref{varR} can be derived as follows. Using identities from Lemma \ref{Lemvar}, let us
compute the Fr\'echet derivative from the left-hand side of equation \eqref{varR}, where  $\ds X=\frac{\delta \rho}{\delta u}$.
The result can be represented as
\begin{equation}\label{RQ}
D_t(T)+T\, F_*+F_*^+\, T=Q,
\end{equation}
where
\[ T=\Big(\frac{\delta \rho}{\delta u}\Big)_*\, ,
\qquad Q=-\sum_{k-1}^{2n}(-1)^kD^k\,
\Big(\frac{\delta \rho}{\delta u}(F_k)_*\Big)\]
and $F_k$ are the coefficients of the Fr\'echet derivative
$F_*$.  
\end{proof}
 
\begin{proposition}\label{nocon}
A scalar evolution equation of even order
$$
 u_t=F(u,u_1,...,u_{2n}) 
$$
cannot have a conserved density $\rho$ of order higher than $n$.
\end{proposition}
\begin{proof} According to Remark \ref{rem26}, we have ${\rm ord}\, T > 2 n$. Comparing the coefficients of the highest power of $D$ in \eqref{RQ}, we get $F_{2n}=0$. 
\end{proof}
\begin{remark}
This proposition is not true for systems of evolution equations.
\end{remark}

\subsection{Canonical densities and necessary integrability conditions}\label{sec226}

In this section we formulate the necessary conditions for the existence of higher
symmetries or conservation laws for equations of the form \eqref{eveq}. In accordance with Theorems \ref{tlsym} and \ref{svsok}, such equations possess a formal symmetry. In Remark \ref{rem24} we discussed obstructions to the existence of a formal symmetry, which are just integrability conditions. Below we show that these conditions can be written in the form of conservation laws.

For equations \eqref{eveq} possessing a formal symmetry $\Lambda$ we define a
sequence of {\sl canonical conserved densities}.

\begin{definition}\label{drho} The functions
\begin{equation}\label{rro}
\rho_i={\rm res}\,(\Lambda^i), \qquad i=-1,1,2,\dots, \, \quad {\rm and}
\qquad \rho_0={\rm res}\, \log (\Lambda)
\end{equation}
are called {\sl canonical densities} for the equation \eqref{eveq} .
\end{definition}

\begin{remark}\label{rem210} Despite the fact that the formal symmetry $\Lambda$ is not unique, the canonical densities are well defined. Namely, it follows from  \eqref{genLam} that the canonical density $\bar \rho_i$ corresponding to any formal symmetry $\bar \Lambda$ is a linear combination of the densities $\rho_j,\,j\le i$ defined by $\Lambda$.
\end{remark}

\begin{theorem}\label{prho} If the equation \eqref{eveq} has a formal
symmetry $\Lambda$, then the canonical densities \eqref{rro}
 define local conservation laws
\begin{equation}\label{canlaws}
D_{t}(\rho_i)=D(\sigma_{i}), \qquad \sigma_{i} \in {\cal F}, \qquad
i=-1,0,1,2,\dots,
\end{equation}
for the equation \eqref{eveq}.
\end{theorem}

{\bf Proof.} If the formal symmetry $\Lambda$ satisfies the equation \eqref{Lambdaeq}, then it follows from Proposition \ref{Lambdaspace}
that  $\Lambda ^k,\ k=-1,1,2,3...$ also satisfy  \eqref{Lambdaeq}. Using Adler's theorem \ref{adler},
we get
\[ D_t(\rho_k)=D_t({\rm res}\,\Lambda^k)={\rm res}\, ([F_*,\Lambda^k])\in {\rm Im} \,D ,
\quad k=-1,1,2,3...
\]
Moreover,
\[ D_t(\rho_0)={\rm res}\,(D_t(\Lambda)\,\Lambda^{-1})={\rm res}\, ([F_*,\Lambda]\,\Lambda^{-1})=
{\rm res}\,([F_*\Lambda^{-1},\,\Lambda]) \in {\rm Im} \,D .\]

\begin{theorem}\label{contriv}
Under the assumptions of Theorem \ref{svsok} there exists a formal symmetry $\Lambda$ such that all even canonical densities $\rho_{2j}$
are trivial.
\end{theorem}
\begin{proof} It follows from \eqref{LS} that 
$$
(\Lambda^{2j})^{+} = S\, \Lambda^{2j}\, S^{-1}.
$$
Since ${\rm res}\, (\Lambda^{2j}) = -{\rm res}\, \Big((\Lambda^{2j})^{+}\Big),$ the theorem follows from  Corollary \ref{cor21}.
\end{proof}

\begin{example} The differential operator $\Lambda = D$ is a formal symmetry for any linear equation of the form $u_t=u_n.$ Therefore,
all canonical densities are equal to zero.
\end{example}
\begin{example} \label{kdvl} The KdV equation \eqref{kdv} 
has a recursion operator 
\begin{equation}\label{recop}
\hat{\Lambda}=D^2+4 u+2 u_1 D^{-1}\, ,
\end{equation}
which satisfies equation \eqref{Lambdaeq}.
A formal symmetry for the KdV equation can be obtained 
as $\Lambda =\hat{\Lambda}^{1/2}$. The infinite commutative hierarchy of  symmetries
for the KdV equation is generated by the recursion operator:
$$
G_{2k+1}=\hat{\Lambda}^k(u_1)\, .
$$
The first five canonical densities for the KdV equation (cf. Example \ref{kdvl}) are
$$
\rho_{-1}=1,\qquad\rho_0=0,\qquad \rho_1=2 u, \qquad \rho_2=2 u_1, \qquad
\rho_3=2 u_2+ u^2.
$$
We see that even canonical densities are trivial. 
\end{example}

\begin{example}
The Burgers equation \eqref{burgers} has the recursion operator
\[ \Lambda=D+u+u_1 D^{-1}\, .\]
Functions $G_n=\Lambda^n(u_1)$ are generators of symmetries for the
Burgers equation. The canonical densities for the Burgers equation are
$$
\rho_{-1}=1,\qquad\rho_0=u,\qquad \rho_1=u_{1}, \qquad
\rho_2=u_2+ 2 u u_{1}, \dots\, .
$$
Although $\rho_{0}$ is not trivial, all other canonical densities are trivial. This is in accordance with Proposition \ref{nocon}.
\end{example}

Now we can refine Definition \ref{def21} such that linear equations become $C$-integrable. 
\begin{definition}\label{clar}
Equation \eqref{eveq} is called {\it $S$-integrable}   if it has a formal symmetry that provides infinitely many linearly independent non-trivial canonical densities.  An equation is called {\it $C$-integrable}  if it has a formal symmetry such that only finite number of canonical densities are non-trivial and linearly independent.
\end{definition}
The notions of $S$ and $C$-integrability are well-defined due to Remark \ref{rem210}.

Suppose that equation  \eqref{eveq} possesses a formal
symmetry $\Lambda$. 
\begin{question}
How are the functions $\rho_i, \sigma_i$ in \eqref{canlaws} related to the right-hand side of equation \eqref{eveq}?
\end{question}

The coefficients of the formal symmetry $\Lambda$ can be
found directly from the
linear equation \eqref{Lambdaeq}. The first $n-1$ coefficients $l_1,l_0,...,l_{3-n}$
of $\Lambda$ coincide with the first $n-1$ coefficients of the formal series
$(F_*)^{1/n}.$ Indeed, $u_{\tau}=F$ is a symmetry of
equation \eqref{eveq} and we can use the main idea of the proof of Theorem \ref{tlsym}. Carefully calculating the number of correct coefficients, we arrive at the ansatz
\[
\Lambda=(F_*)^{1/n}+\tilde l_{2-n} D^{2-n}+\tilde l_{1-n}D^{1-n}+\cdots \,.
\]
Having the first $n-1$ coefficients of $\Lambda$, we can find $n-1$ canonical
densities $\rho_{-1},\rho_0,...,\rho_{n-3}$
explicitly in terms of the coefficients
$$F_{i}=\frac{\partial F}{\partial u_{i}}$$ of the Fr\'echet derivative
$$F_*=F_n D^n+F_{n-1}D^{n-1}+\cdots +F_0.$$

Equating coefficients of $D$
in equation \eqref{Lambdaeq}, we find  that the first unknown coefficient $\tilde l_{2-n}$ of $\Lambda$
can be found if and only if the first canonical density
\begin{equation}\label{rho-1}
\ds \rho_{-1}=F_n^{-\frac{1}{n}}
\end{equation}
is a density of a local conservation law for equation \eqref{eveq}, i.e., there
exists such a function $\sigma_{-1}\in {\cal F}$ that $D_t(\rho_{-1})=D(\sigma_{-1})$.

If $D_t(\rho_{-1})\not\in {\rm Im}\,D$,  then the formal
symmetry does not exist and consequently the equation \eqref{eveq}  cannot
have an infinite sequence of higher symmetries or conservation laws.

If $\sigma_{-1}\in {\cal F}$ exists, then the coefficient $\tilde l_{2-n}$ can be expressed explicitly in terms of the coefficients $F_n,...,F_0$ and $\sigma_{-1}$.  Similarly,
the next coefficient $\tilde l_{1-n}$ can be found (as an element of ${\cal F}$) if and only if
 $D_t(\rho_0)=D(\sigma_0)$ for some $\sigma_0\in {\cal F}$. In this case $\tilde l_{1-n}$ can be explicitly expressed
in terms of $F_n,...,F_0, \sigma_{-1}, \sigma_0$,  etc. 

\begin{example}
Consider evolution equations of second order
$$
u_{t}=F(u, u_{1}, u_{2}).
$$
Computations described above show that the  three first canonical densities can be written in the form
$$
\begin{array}{c}
\displaystyle \rho_{-1}= F_2^{-1/2}, \,\, \qquad 
\displaystyle \rho_{0}= F_2^{-1/2}\,\sigma_{-1} - F_2^{-1} F_1, \\[3mm]
\displaystyle \rho_{1}=\rho_{-1}\,F_0 - \frac{\rho_{0}^{2}}{4 \rho_{-1}} +
\frac{\rho_{0}\sigma_{-1}}{2}-\frac{\rho_{-1} \sigma_{0} }{2}.
\end{array}
$$
\end{example}
\begin{remark}
It follows from Theorems \ref{tlsym}, \ref{svsok}, \ref{prho} and \ref{contriv} that if we are going  to find equations \eqref{eveq} with higher symmetries, we have to use conditions \eqref{canlaws} only, while for equations with higher conservation laws we may additionally assume that $\rho_{2 j}=D(\theta_{j})+c_j,$ where $\theta_{j}\in {\cal F}$ and $c_j \in\C$. 
Thus the necessary conditions, which we employ for conservation laws are stronger than the ones for symmetries.  
\end{remark}

Using the ideas of \cite{CheLeeLiu79, mesh}, one can derive a recursive formula for the whole  infinite chain of the canonical conserved densities $\rho_i$.  
For equations of the form \eqref{vvs} such a formula was obtained in \cite{meshsok}. It has the following form:
\begin{align}
\rho_{n+2}&=\frac{1}{3}\bigg[\sigma_n-\delta_{n,0}f_{0} -f_{1}\rho_{n}-
f_{2}\Big(D(\rho_{n}) + 2\rho_{n+1}+\sum_{s=0}^{n} \rho_{s}\,\rho_{n-s}\Big)\bigg]
-\sum_{s=0}^{n+1} \rho_{s}\,\rho_{n+1-s}
 \nonumber\\[2mm]
&-\frac{1}{3}\sum_{0\le s+k\le n}\rho_{s}\,\rho_{k}\,\rho_{n-s-k}
-D\biggl[\rho_{n+1}+\frac{1}{2}\sum_{s=0}^{n}\rho_{s}\,\rho_{n-s}+
\frac{1}{3} D(\rho_{n}) \biggr], \qquad n\ge0, \label{rekkur_sc}
\end{align}
where the first two elements of the sequence $\rho_i$ read as
\begin{equation} \label{first2}
 \rho_0=-\frac{1}{3}f_{2},\qquad \rho_1=\frac{1}{9}f_{2}^2-\frac{1}{3}f_{1}+\frac{1}{3} D(f_{2}).
\end{equation}
Here, $\delta_{i,j}$ is the Kronecker delta, $\ds f_{i}=\frac{\p f}{\p u_i}$, where $i=0,1,2.$ The density $\rho_{-1}=1$ disappears in the above recursive formula. 

Given an equation of the form \eqref{vvs},  one can verify on a computer as many integrability conditions \eqref{canlaws}, \eqref{rekkur_sc} as the computer allows. 

Moreover, using several first integrability conditions (see, for instance, the next section), one can solve the following classification problem: find all equations of the prescribed type, which have a formal symmetry.  A full classification result includes:
\begin{itemize}
\item[1)] A complete\footnote{Very often complete up to a class of admissible transformations.} list of integrable equations that satify the necessary integrability conditions;
\item[2)] A confirmation of integrability for each equation from the list;
\item[3)] A contructive description of transformations that bring a given integrable equation to one from the list;
\item[4)] The number of necessary conditions, which should be verified for a given equation to establish its integrability. 
\end{itemize}
A proof of a classification result contains Items 3) and 4). For Item 2) one can find a Lax representation or a transformation that links the equation with an equation known to be integrable. The existence of an auto-B\"acklund tranformation with an arbitrary parameter is also a proper justification of integrability.

\subsection{Classification of integrable KdV-type equations}
To demonstrate how the necessary conditions are efficient, we solve in this section a simple classification problem \cite{ibshab0,fokas}.

 Consider evolution equations of the form \eqref{kdvt}.
 From \eqref{first2} it follows that for such equations $\rho_0=\sigma_0=0$ and
\begin{equation} \label{con1}
D_t\left(\frac{\partial f}{\partial u_{1}}\right)=D(\sigma_1),
\end{equation}
where $\sigma_1$ is a function depending on $u$, $u_1$, \ldots, $u_3$.

\begin{example} For the mKdV equation 
\begin{equation} \label{mkdv} 
u_t=u_3+6\,u^2 u_1
\end{equation}   
the conservation law \eqref{con1} reads as follows:
$$
D_t (u^2)=D (2 u u_2-u_1^2+3 u^4). 
$$
\end{example}

Applying the Euler operator to both sides of \eqref{con1} and using \eqref{delro}, we obtain
\begin{equation} \label{conn1}
0=\frac{\delta}{\delta u} D_t\left(\frac{\partial f}{\partial u_{1}} \right) =
3 u_4 \left(u_2 \, \frac{\partial^4 f}{\partial u_{1}^4}+u_1 \, \frac{\partial^4 f}{\partial u_{1}^3 \partial u}\right)+O(3),
\end{equation}
where $O(3)$ denotes terms of order not greather than 3.
It should  be identity in the variables $u,u_1,\dots,u_4$. Equating the coefficient of $u_4$ to zero and employing the fact that $f$ is independent of $u_2$, we get
$$
f(u, u_1)=\mu u_1^3+A(u) u_1^2+B(u) u_1+C(u)
$$
with some constant $\mu$. It can be readily verified that for such a function $f$  the condition \eqref{conn1} is equivalent to the following system of ODEs:
$$
\mu A'=0, \quad \qquad B'''+8 \mu B'=0, \quad \qquad (B'C)'=0, \qquad \quad A B'+6 \mu C'=0.
$$

The next necessary integrability condition \eqref{rekkur_sc} reads 
$$
D_t\left(\frac{\partial f}{\partial u}\right)=D(\sigma_2),
$$
which implies
$$
\frac{\delta}{\delta u} D_t\left(\frac{\partial f}{\partial u} \right)=0.
$$
The latter condition leads to the following additional equations:
$$
A'=0, \qquad \quad AC''=0, \qquad \quad (C'''+2\mu C')'=0, \qquad \quad (CC'')'=0.
$$

In the case $\mu\ne 0$ solving ODEs obtained above, we determine the functions $A$, $B$ and $C$. As a result, up to a scaling $u\to {\rm const} \,u,$ we arrive at the equations
\begin{equation}
u_t=u_{xxx}-\frac{1}{2} u_{x}^3+(c_1 e^{2u}+c_2 e^{-2u}+c_3)\,u_{x}\, \label{e1}
\end{equation}
and
\begin{equation}
u_t=u_{xxx}+c_1 u_x^3+c_2 u_x^2+c_3 u_x+c_4, \label{e2}
\end{equation}
where $c_i$ are arbitrary constants.

If $\mu=0$, then solving the above system of ODEs for the functions $A,B,C$, we find that the equation has the form
$$
u_t=u_{xxx}+c_0 u_x^2 +(c_1 u^2+c_2 u +c_3) u_x+c_4u+c_5,
$$
where
$$
c_0 c_1=0,\qquad c_0 c_2=0,\qquad c_4 c_1=0,\qquad c_4 c_2=0,\qquad c_1 c_5=0.
$$
By the third integrability condition (see formula \eqref{rekkur_sc}  for $\rho_3$), we find the additional relations
$$
c_0 c_4=0,\qquad c_2 c_5=0.
$$
In the case $c_0\ne0$ we arrive at a particular case of equation \eqref{e2}. If $c_0=0$, two cases are possible: a) $c_4=c_5=0$ and b) $c_1=c_2=0.$ In the case a) we get 
\begin{align}
u_t&=u_{xxx}+(c_1 u^2+c_2 u +c_3) \, u_x. \label{e3}
\end{align}
 In the case b) we have a linear equation, which has a formal symmetry $\Lambda=D.$ Integrability 
 of equations \eqref{e1}--\eqref{e3} has to be determined by different means. Equation \eqref{e3} can be reduced to the KdV or to the mKdV equation by scalings and a shift of $u$. Equation \eqref{e2} is related to \eqref{e3} by the potentiation (see Definition \ref{pot1}).   Equation \eqref{e1} was found in \cite{caldeg}. It is interconnected with the KdV equation by a differential substitution of Miura type \cite{cvsokyam}.

\subsection{Equations of Harry-Dym type}
The Harry-Dym equation 
\begin{equation}\label{HDeq}
u_t=u^3\, u_{xxx}
\end{equation}
is one of well-known integrable evolution equations. Equations of the form 
\begin{equation}\label{HDtype}
u_t=f(u)\, u_3+Q(u,u_1,u_2), \qquad f'(u)\ne 0
\end{equation}
are called {\it equations of Harry-Dym type.} 

The simplest integrablity condition for the equation \eqref{eveq} (see formula \eqref{rho-1} has the form 
$
D_t(\rho_{-1})=D(\sigma),
$
where
$$
\rho_{-1}=\Big(\frac{\partial F}{\partial u_{n}} \Big)^{-\frac{1}{n}}
$$
In this section we show that this condition  allows one to reduce the function $f$ for any integrable equation of the form \eqref{HDtype} to 1 by the quasi-local transformations (see Definition \ref{def15}).

\begin{theorem} {\rm \cite{sokshab}}. Any integrable equation of the
form
can be reduced by a potentiation and point transformations to the form
\begin{equation}\label{Geq}
v_t=v_3+G(v_1, v_2).
\end{equation}
\end{theorem}

\begin{proof} First of all we make the point transformation $\tilde u=f(u)^{-1/3}$ to bring the equation \eqref{HDtype} to the form
$$
\tilde u_t=\frac{\tilde u_{3}}{\tilde u^3}+\tilde Q(\tilde u, \tilde u_{1},\tilde u_{2}).
$$
For such an equation we have $\rho_{-1}=\tilde u$, and therefore for any integrable equation of this form
the function $\tilde u$ is a conserved density:
$$
\tilde u_t=D\left(\frac{\tilde u_{2}}{\tilde u^3}+\Psi(\tilde u, \tilde u_{1})\right).
$$
 
The second step is the potentiation $D \hat u=\tilde u$. As the result, we obtain
$$
\hat u_t=\frac{\hat u_3}{\hat u_1^3}+\Psi(\hat u_1, \hat u_2).
$$

The last step is the point transformation
\begin{equation}\label{trgad}
\hat t=t, \qquad \hat x = v, \qquad \hat u=x.
\end{equation}
For this transformation (see Section \ref{Tra}) we have
 $$
\hat u_1=\frac{1}{v_1}, \qquad \hat u_2= -\frac{v_2}{v_1^3}, \qquad
\hat u_3=-\frac{v_3}{v_1^4}+\frac{3 v_2^2}{v_1^5}, \qquad \hat u_t= -\frac{v_t}{v_1}.
$$
Using this formulas, one can check that any equation of the form
$$
\hat u_t=\frac{\hat u_3}{\hat u_1^3}+\Psi(\hat u_1, \hat u_2)
$$
transforms to an equation of the form \eqref{Geq}.
\end{proof}

\begin{example} For the Harry-Dim equation \eqref{HDeq} we take $\ds \tilde u=\frac{1}{u}$ to get the equation
$$
\tilde u_t=D_x\left(\frac{\tilde u_{2}}{\tilde u^3}-\frac{3\, \tilde u_1^2}{2\, \tilde u^4}\right).
$$
Applying the transformation \eqref{trgad}, we obtain the Swartz-KdV equation
$$
v_t=v_{3}-\frac{3\, v_2^2}{2\, v_1},
$$
which is a particular case of the Krichever-Novikov equation \eqref{KN}.
\end{example}

\subsection{Integrability conditions for non-evolution equations}

In this section we generalize \cite{HSS} the main concepts of the symmetry approach such as
the formal recursion operator and the canonical conserved densities to the
case of non-evolutionary equations of the form
\be\label{qtt}
q_{tt}=F(q,q_1,q_2,\dots , q_n,\, q_{t},q_{t1}, q_{t2},\dots,q_{tm}), \qquad n > m.
\ee
Such type equations were excluded from consideration in \cite{sokshab, MikShaSok91},
where only evolution equations have been investigated. 
\begin{remark} Equation 
\eqref{qtt} can be rewritten as a system of two evolution equations
\[q_t=p, \qquad p_t=F(q,q_1,q_2,\dots,q_n, p,p_1,\dots,p_m).   \]
However,  the matrix coefficient of leading derivatives in the linearization operator  
is a Jordan block whereas in the paper {\rm \cite[Section 3.2.1]{MikShaSok91}}  this coefficient was supposed to be diagonalizable.
\end{remark}

For equation \eqref{qtt} all mixed derivatives of $q$ containing at
least two time differentiations can be expressed in terms of
\begin{equation} \label{dyntvar}
q, \quad q_{x},\quad q_{xx},\quad \dots, q_{i}, \, \dots,  \qquad q_{t}, \quad q_{1t}=q_{xt},\quad
q_{2t}=q_{xxt},\,\,\dots, \quad q_{i\,t}, \dots
\end{equation}
in virtue of \eqref{qtt}.
The derivatives \eqref{dyntvar} are regarded as {\it independent} variables.

Equation of the form
\be\label{rs}
q_\tau=G(q, q_1, q_2, \dots , q_{r},\,\, q_{t}, q_{t1}, q_{t2}, \dots, q_{ts})
\ee
compatible with \eqref{qtt} is called infinitesimal (local) symmetry of equation 
\eqref{qtt}. Compatibility implies that the function $G$ satisfies (cf. with \eqref{DFsym}) the
equation ${\cal L}(G)=0$, where
$$
{\cal L}=D_{t}^{2}-\sum_{i=0}^{n} \frac{\partial F}{\partial q_{i}}\,D_{x}^{i}-
\left(\sum_{i=0}^{m} \frac{\partial F}{\partial q_{ti}}\,D_{x}^{i}\right)
D_{t} = D_{t}^{2}-(M+ND_t)
$$
is the linearization operator for equation \eqref{qtt}.

In order to rewrite consistency conditions of \eqref{qtt} and \eqref{rs}
in terms of conservation laws of \eqref{qtt} one can use a formal
Lax representation of the problem. The linearization of equations \eqref{qtt},
\eqref{rs} gives rise to the compatibility problem for linear equations
\[ \phi_{tt}=(M+ND_t)\,\phi,\qquad  \phi_\tau=(A+BD_t)\,\phi     \]
or, equivalently,
\begin{equation}\label{couplin} \Phi_t=F_*\Phi,\qquad  \Phi_\tau=G_*\Phi,\qquad \Phi=\left(\begin{array}{c}\phi\\
                            \phi_t\end{array}\right),
\quad F_*=\left(\begin{array}{cc}   0 & 1\\
                      M & N\end{array}\right)
\end{equation}
where
\be\label{G*}  G_*=
   \left(\begin{array}{cc}  A & B\\
           \hat A & \hat B\end{array}\right),\qquad \hat A = A_t+B M ,
           \qquad \hat B = B_t+B N+A.
\ee
The cross differentiation of equations \eqref{couplin} yields
$$ D_t(G_*)=[F_*,G_*]+D_\tau(F_*)  $$
where $F_*,\, G_*$ are matrix differential operators. The crucial
step in the symmetry approach (see proof of Theorem \ref{tlsym})
is to consider instead of above equation one as follows
\be\label{RF}
  D_t(R)=[F_*, \,R],
\ee
where $R$ is a pseudo-differential series with matrix coefficients. We call $R$
{\it matrix formal symmetry}.

Denoting as before $R_{11}=A,\, R_{12}=B,$ we can rewrite \eqref{RF} as
follows
\be\label{NM1}
A_{tt}-N A_t+[A,\,M]+(2B_t+[B,\,N])\,M+B M_t=0,
\ee
\be\label{NM2}
B_{tt}+2A_t+[B,\,M]+[A,\,N]+([B,\,N]+2B_t)\,N+B N_t-N B_t=0.
\ee

Identities \eqref{NM1}, \eqref{NM2} mean that the scalar pseudo-differential series
${\cal R}=A+B\,D_{t}$ is related to the linearization ${\cal L}$
of equation \eqref{qtt} by
\begin{equation} \label{defrec}
{\cal L}\, (A+B\,D_{t})=(\bar A+B\,D_{t})\, {\cal L},
\end{equation}
where $\bar A=A+2 B_{t}+[B,N]$.
A pseudo-differential series ${\cal R}=A+B\,D_{t}$, whose
components
 $$
 A= \sum_{-\infty}^{n} a_{i} D^{i}, \qquad
 B= \sum_{-\infty}^{m} b_{i} D^{i}
 $$
satisfy \eqref{NM1} and \eqref{NM2}, is called
{\it scalar formal symmetry} for
equation \eqref{qtt}. If $A$ and $B$ are differential operators
(or ratios of differential operators), condition \eqref{defrec}
implies the fact that the recursion operator ${\cal R}$ maps symmetries of
equation \eqref{qtt} to symmetries.  

\begin{remark} Let ${\cal R}_{1}=A_{1}+B_{1}\,D_{t}$ and ${\cal R}_{2}=A_{2}+B_{2}\,D_{t}$
be two scalar formal recursion symmetries. Then the product
${\cal R}_{3}={\cal R}_{1} {\cal R}_{2}$, in which $D_{t}^{2}$ is replaced by
$(M+ND_t)$ is also a scalar formal symmetry.
\end{remark}

An operator $S=P+Q D_{t}$ is said to be symplectic (see \eqref{Req} ) if
$$
{\cal L}_{*} S+\bar S {\cal L}=0, \qquad {\rm where} \qquad \bar S=\bar P+\bar Q D_{t}.
$$
If $S$
can be applied to symmetries, then it maps symmetries to cosymmetries.
In the symmetry approach  $P$ and $Q$ are supposed to be
formal pseudo-differential symbols.

The operator equations for the components $P$ and $S$ of formal symplectic operator $S=P+Q D_{t}$
have the following form
$$
P_{tt}+N^{+} P_{t}+2 Q_{t} M +Q M_{t} = M^{+} P-P M-(Q N + N^{+}Q)\,M-N_{t}^{+} P,
$$
$$
Q_{tt}+2 P_{t}+2 Q_{t} N + N^{+} Q_{t} = M^{+} Q-Q M-(Q N + N^{+} Q)\,N-(P N + N^{+} P) -
(N_{t}^{+} Q + Q N_{t}).
$$
The linearization $P+Q D_{t}$ of the variational derivative of any conserved
density for equation \eqref{qtt} satisfies these equations up to a "small" rest (cf. \cite{soksvin1}).

\begin{theorem} \begin{itemize}
\item[i)] If equation \eqref{qtt} possesses an infinite sequence of
higher symmetries of the form
$$
q_{\tau_{i}}=G_{i}(q,q_1,q_2,\dots , q_{r_{i}},\, q_{t},q_{t1}, q_{t2},
\dots,q_{ts_{i}}),
$$
then there exists a scalar formal symmetry of the form
\begin{equation} \label{fro}
R=(a_{0}+a_{-1} D^{-1}+\dots)+(b_{-1} D^{-1}+b_{-2}D^{-2}+\dots)\,D_{t},
\end{equation}
where $a_{i}, b_{i}$ are some functions of variables \eqref{dyntvar}.
\item[ii)]
If equation \eqref{qtt} possesses an infinite sequence of local
higher conservation laws $D_{t} \rho_{i}=D \sigma_{i}$, where
$\rho_{i}, \sigma_{i}$ are functions of variables \eqref{dyntvar}, then
there exists a formal symmetry \eqref{fro} and a formal symplectic
operator of the form
$$
S=(p_{0}+p_{-1} D^{-1}+\dots)+(q_{-1} D^{-1}+q_{-2}D^{-2}+\dots)\,D_{t}.
$$
\end{itemize}
\end{theorem}

Let $R$ be a formal symmetry of the form \eqref{fro}. It is easy
to find an operator
$$
R^{-1}=(\alpha_{-1} D^{-1}+\alpha_{-2} D^{-2}\dots)+
(\beta_{-2} D^{-2}+\beta_{-3}D^{-3}+\dots)\,D_{t}
$$
such that $R\,R^{-1}=R^{-1}\,R=1.$ Recall that we eliminate the term $D_{t}^{2}$  in the product of scalar formal symmetries in
virtue of the relation ${\cal L}=0$.
Operator $R^{-1}$ is uniquely defined, the coefficient $\beta_{-2}$ is equal to
$(b_{-1})^{-1}$.

In Section \ref{sec226} we have used the residues of the powers of a formal symmetry to define the canonical conserved densities. Here, we apply a different construction.

\begin{theorem} Let $R$ be a formal symmetry of the form \eqref{fro}. Then there exists a unique representation of total derivative
operators $D$ and $D_{t}$ in the form
\begin{equation} \label{decomp}
D_{x}=\sum_{-\infty}^{2} \rho_{i} R^{i}, \qquad
D_{t}=\sum_{-\infty}^{3} \sigma_{i} R^{i}.
\end{equation}
Functions $\rho_{i}$ and $\sigma_{i}$ are densities and fluxes of some
(maybe trivial) conservation laws
\begin{equation}\label{cond}
(\rho_{i})_{t}=(\sigma_{i})_{x}
\end{equation}
for equation \eqref{qtt}.
\end{theorem}

\begin{example}  
The following formulas define the simplest integrability conditions \eqref{cond}
for equations of the form 
\be \label{qF} q_{tt}=q_{xxx}+F(q,q_x,q_t,q_{xx},q_{xt}):    \ee  
$$
\rho_{1}=u_{2}, \qquad 
 \rho_{2}= v_{2}+\frac{2}{3} \sigma_{1}, \qquad 
 \rho_{3}=6 \sigma_{2}-u_{2}\sigma_{1}+9 u_{1}-3 u_{2} v_{2}-\frac{1}{3}
u_{2}^{3},
$$
$$
\rho_{4}=6 \sigma_{3}-9 u_{2} \sigma_{2}+3 \sigma_{1}^{2}+27 u_{1} u_{2}-u_{2}^{4}+81 v_{1}-
9 u_{2}^{2}v_{2}-27 v_{2}^{2},
$$
$$
\rho_{5}=-2 \sigma_{4}-18 \sigma_{1} \sigma_{2}+27 (\sigma_{1})_{t}+
3 \sigma_{1}^{2} u_{2}+3 \sigma_{3} u_{2}+9 \sigma_{1} u_{2} v_{2}+
\sigma_{1} u_{2}^{3}-27 \sigma _{1} u_{1},
$$
where 
$$
u_1=\frac{\partial F }{\partial q_t}, \qquad v_1=\frac{\partial F }{\partial q_x},  \qquad u_2=\frac{\partial F }{\partial q_{xt}},
 \qquad v_2=\frac{\partial F }{\partial q_{xx}}.
$$
The conditions mean that $\rho_{i}$ are densities of {\it local} conservation
laws for equation \eqref{qF}. In other words, for any $\rho_{i}$ there
exists a function $\sigma_{i}$ depending on variables \eqref{dyntvar}.
\end{example}

\begin{op} Find all integrable Boussinesq type equations of the form
$$
q_{tt}=q_{xxxx}+F(q, q_x, q_{xx}).
$$
\end{op}
 
 \section{Recursion and Hamiltonian quasi--local operators}
 
Recursion and Hamiltonian operators establish additional relations between higher symmetries and conserved densities. 

A recursion operator is an operator ${\cal R}$ that satisfies \eqref{defR} and therefore maps symmetries to symmetries. The simplest symmetry for any equation \eqref{eveq}
is $u_{\tau}=u_x$. The
usual way to get generators of all the other symmetries is to act on $u_x$ by a recursion operator. 
\begin{lemma}\label{lem25}
Let ${\cal R}$ be a recursion operator. Then the operator  ${\cal R}^{+}$ maps cosymmetries to cosymmetries.
\end{lemma}
\begin{proof} It follows from \eqref{defR1} that 
$[D_t+F_{*}^{+},\, {\cal R}^{+}]=0$. Using Definition \ref{cosym}, we arrive at the statement of the lemma.
\end{proof}

The main problem is that for almost all integrable models the recursion operator is non-local (see, for instance, 
\eqref{recop}) and we can apply it only to very special functions from $\cal F$ to get a function from $\cal F$.

 \subsection{Quasi-local recursion operators}
 
 In this section we consider non-formal recursion operators, which generate hierarchies of symmetries for integrable evolution equations (see Proposition \ref{defrec1}).

Most of known recursion operators have the following special non-local structure:
\begin{equation}\label{anz}
{\cal R}=R+\sum_{i=1}^k G_i\, D^{-1}\circ g_i,\qquad g_i, G_i\in {\cal F}, 
\end{equation}
where $R$ is a differential operator.  Without loss of generality we assume that the functions $G_1,\dots, G_k$ (as well as $g_1,\dots, g_k$) are linearly independent over $\C$. Operators of the form \eqref{anz} are called {\it quasi-local} or {\it weakly nonlocal}.

\begin{remark} It is easy to see that any expression of the form $f D^{-1}\circ g$ can be written as 
$$f D^{-1}\circ g = (a D-b)^{-1}, \quad {\rm where} \qquad a = \frac{1}{f g}, \quad b=\frac{D(f)}{g f^2}.$$ 
Thus, the quasi-local ansatz is a non-commutative analog of the  partial fraction decomposition of a rational function. Any ratio $A B^{-1}$ of differential operators can be written in a quasi-local form \eqref{anz}. However, the functions $g_i$ and $G_i$ belong to a differential extension of the basic differential field ${\cal F}$ such that $B$ admits a factorization into first order multipliers. The assumption $g_i, G_i\in {\cal F}$ is important for further considerations. 
\end{remark}

\begin{lemma}\label{lem26} Let 
\begin{equation}\label{anz1}
\bar{\cal R}=\bar R+\sum_{i=1}^{\bar k} \bar G_i \, D^{-1} \circ \bar g_i 
\end{equation}
be a quasi--local operator. Then the product ${\cal R}\circ \bar {\cal R}$ is quasi-local if $\,\,g_i \bar G_j=D(a_{ij}), \,\, a_{ij}\in {\cal F}$ for any $i,j.$
\end{lemma}
\begin{proof} It suffices to prove that $G_i D^{-1} \circ g_i \, \bar G_j D^{-1} \circ \bar g_j $ is quasi--local. It follows from identities
$$
G_i D^{-1} \circ g_i \, \bar G_j D^{-1} \circ \bar g_j =G_i D^{-1} (D \circ a_{ij} - a_{ij} D) D^{-1} \circ \bar g_j = G_i a_{ij} D^{-1} \circ \bar g_j 
- G_i D^{-1} \circ a_{ij} \bar g_j.
$$
\end{proof}
\begin{definition} \label{def214}An operator ${\cal R}$ of the form \eqref{anz} is called a {\it quasi-local recursion operator} for eqiation \eqref{eveq} if 
\begin{itemize}
\item[1)]  ${\cal R}$, considered as a pseudo-differential series, satisfies \eqref{defR};
\item[2)]  the functions $G_i$ are generators of some symmetries for \eqref{eveq};
\item[3)]  the functions $g_i$ are variational derivatives of conserved
densities.
\end{itemize}
\end{definition}
\begin{example}
It is easy to see that the recursion operator \eqref{recop} for the KdV equation is quasi-local with $\ds k = 1,\, G_1 = \frac{u_x}{2}$ and $\ds g_1 =\frac{\delta u}{\delta u}= 1$.
\end{example}
\begin{remark}
Maybe it is reasonable {\rm \cite{sylv}} to add  the hereditary property {\rm \cite{Fuch}} of the operator ${\cal R}$ to  1--3.
\end{remark}

\begin{remark} Substituting \eqref{anz} into  \eqref{defR} and equating the non-local terms, we obtain that
$$
\sum_{i=1}^k (D_t-F_*)(G_i)\, D^{-1}\circ g_i +\sum_{i=1}^k G_i\, D^{-1}\circ (D_t+F_*^+)(g_i)=0.
$$
This ``almost'' implies that the functions $G_i$ are  symmetries and the functions $g_i$ are cosymmetries.  

\end{remark}

\begin{question} Why applying a quasi-local recursion operator to local symmetries, we obtain local symmetries?
\end{question}
\begin{question} Why the product of two  quasi-local recursion operators is a quasi-local recursion operator?
\end{question}

Suppose that equation \eqref{eveq} is a member of a commutative hierarchy  (see Remark \ref{rem12} and Example \ref{ex25}). 
This means that
\begin{itemize}
\item[i)] Any equation of the hierarchy is a higher symmetry for all the other equations;
\item[ii)] Each conserved density of any equation from the hierarchy is a conserved density for  all the other equations.
\end{itemize}
In this case the expression ${\cal R}(g)$, where $g$ is any symmetry of equation \eqref{eveq}, belongs to ${\cal F}$. It follows from Item ii) and  
\begin{lemma} {\rm \cite{Olv93}} The product of the right-hand side $F$ of equation \eqref{eveq} 
and the variational derivative of any conserved density for equation \eqref{eveq}
belongs to ${\rm Im}\, D$.
\end{lemma}
A different choice of integration
constants gives rise to an additional linear combination of the
symmetries $G_1,\dots, G_k$.  Probably a quasi-local ansatz for finding a recursion operator was used for the first time in \cite{sokkn}.

Let ${\cal R}$ and $\bar{\cal R}$ given by \eqref{anz} and \eqref{anz1} be quasi-local recursion operators. Then the non-local terms in the product ${\cal R}\circ \bar{\cal R}$ has the form
$$
\sum_{i=1}^{\bar k} {\cal R} (\bar G_i) \, D^{-1} \circ \bar g_i + \sum_{i=1}^{k}  G_i \, D^{-1} \circ \bar{\cal R}^{+}(g_i). 
$$
To obtain this formula one can use Lemma \ref{lem26} and the fact that for any differential operator $S$ and $p,q \in {\cal F}$ the function 
$p\, S(q)+q\, S^{+}(p)$ belongs to ${\rm Im}\,D.$

We see that the operator ${\cal R}\circ \bar{\cal R}$ satisfies the requirements 1) (see Proposition \ref{Lambdaspace}) and 2) of Definition \ref{def214}. It follows from Lemma \ref{lem25} that the functions $g_i$ for the product are cosymmetries (but not necessarily variational derivatives of some conserved densities). However, for a particular equation one can try to prove that any cosymmetry is a variational derivative of a conserved density. Briefly outline
how to prove it for the KdV equation. 
\begin{itemize}
\item[1.] Using that for the KdV equation the operator $\ds \frac{\partial}{\partial u}$ maps any cosymmetry to a cosymmetry reducing its order by 2, prove that any cosymmetry has an even order;
\item[2.] Derive from \eqref{varR} that any cosymmetry $S$  has the form $S=c\,u_{2 k}+O(2k-1);$ 
\item[3.] Employ the fact that for any $k$ the KdV equation  has a conserved density of the form $\rho_k=u_k^2+O(k-1)$ to reduce the order of cosymmetry: $\ds S_1=S+(-1)^{k+1} \frac{c}{2} \frac{\delta \rho_k}{\delta u}$ has a lower order;
\item[4.] Use the induction over $k$.
\begin{op}
Prove the same statement for the Krichever--Novikov equation.
\end{op}
\end{itemize} 
\begin{remark} The function $u_1$ is a cosymmetry for the linear equation $u_t=u_{xxx}$, which is not a variational derivative of a conserved density.  
\end{remark}

Thus  the set of all quasi-local recursion operators for the KdV equation form a commutative (see Proposition \ref{Lambdaspace}) associative algebra $A_{rec}$ over $\C.$ It is possible to prove that this algebra is generated by the operator \eqref{recop}. In other words, $A_{rec}$ 
is isomorphic to the algebra of all polynomials in one variable.

It turns out \cite{demsok} that it is not true for integrable models such as the Krichever--Novikov and the Landau--Lifshitz equations.  

\subsection{Recursion operators for Krichever-Novikov equation} 
 
Consider  the Krichever-Novikov equation \eqref{KN}.  This equation is the most interesting among integrable equations of the form  \eqref{vvs} because of two reasons. 

{\it Reason 1}. All integrable equations \eqref{vvs} can be obtained by some transformations and limit procedures from \eqref{KN}. 

{\it Reason 2}. All algebraic structures related to equation \eqref{KN} are the most generic and sophisticated.  
 
Denote by $G_1$ the right-hand side of \eqref{KN}.
 The fifth order symmetry of \eqref{KN} is given by
$$
G_2=u_5-5\frac{u_4 u_2}{u_1}-\frac{5}{2}
\frac{u_3^2}{u_1}+\frac{25}{2} \frac{u_3
u_2^2}{u_1^2}-\frac{45}{8} \frac{u_2^4}{u_1^3}-\frac{5}{3} P
\frac{u_3}{u_1^2}+\frac{25}{6} P \frac{u_2^2}{u_1^3}-\frac{5}{3}
P' \frac{u_2}{u_1}-\frac{5}{18} \frac{P^2}{u_1^3}+\frac{5}{9} u_1
P''. \label{kn5}
$$
The three simplest conserved densities of \eqref{KN} are
$$
\begin{array}{l}
\ds \rho_1=-\frac{1}{2} \frac{u_2^2}{u_1^2}-\frac{1}{3}
\frac{P}{u_1^2},\qquad \qquad 
\rho_2=\frac{1}{2} \frac{u_3^2}{u_1^2}-\frac{3}{8} \frac{u_2^4}{u_1^4}+\frac{5}{6} P \frac{u_2^2}{u_1^4}+
\frac{1}{18}\frac{P^2}{u_1^4}-\frac{5}{9} P'', \\[9mm]
\ds \rho_3=\frac{u_4^2}{u_1^2}+3 \frac{u_3^3}{u_1^3}-\frac{19}{2} \frac{u_3^2 u_2^2}{u_1^4}+
\frac{7}{3}P \frac{u_3^2}{u_1^4}+\frac{35}{9} P' \frac{u_2^3}{u_1^4}+\frac{45}{8} \frac{u_2^6}{u_1^6}-
\frac{259}{36} \frac{u_2^4 P}{u_1^6}+\frac{35}{18} P^2 \frac{u_2^2}{u_1^6}\\[5mm]
\ds\qquad -\frac{14}{9} P'' \frac{u_2^2}{u_1^2}+\frac{1}{27} \frac{P^3}{u_1^6}-
\frac{14}{27} \frac{P'' P}{u_1^2}-\frac{7}{27} \frac{P'^2}{u_1^2}-\frac{14}{9} P^{(IV)} u_1^2.
\end{array}
$$
In the paper \cite{sokkn} the fourth order quasi--local recursion
operator 
$$
{\cal R}_1=D^4+a_1 D^3+a_2 D^2+a_3 D+a_4+G_1 D^{-1} \circ \frac{\delta
\rho_1}{\delta u}+u_x D^{-1} \circ \frac{\delta \rho_2}{\delta u},
\label{KNR1}
$$
was found. Here, the coefficients $a_i$ are given by
$$
\begin{array}{l}
\ds a_1=-4 \frac{u_2}{u_1},\qquad a_2=6 \frac{u_2^2}{u_1^2}-2 \frac{u_3}{u_1}-\frac{4}{3} \frac{P}{u_1^2}, \\[5mm]
\ds a_3=-2 \frac{u_4}{u_1}+8 \frac{u_3 u_2}{u_1^2}-6 \frac{u_2^3}{u_1^3}+4 P \frac{u_2}{u_1^3}-
\frac{2}{3}\frac{P'}{u_1}, \\[5mm]
\ds a_4=\frac{u_5}{u_1}-2 \frac{u_3^2}{u_1^2}+8 \frac{u_3
u_2^2}{u_1^3}-4 \frac{u_4 u_2}{u_1^2}- 3
\frac{u_2^4}{u_1^4}+\frac{4}{9} \frac{P^2}{u_1^4}+\frac{4}{3}P
\frac{u_2^2}{u_1^4}+\frac{10}{9}P''- \frac{8}{3}
P'\frac{u_2}{u_1^2}.
\end{array}
$$
The following statement can be verified directly.

\begin{theorem} There exists a quasi--local recursion
operator for \eqref{KN} of the form
\begin{equation}\label{KNR2}
\begin{array}{l}
\ds {\cal R}_2=D^6+b_1 D^5+b_2 D^4+b_3 D^3+b_4 D^2+b_5 D+b_6-\\[3mm]
\qquad \,\, \ds \frac{1}{2} u_x D^{-1}\circ  \frac{\delta \rho_3}{\delta u}
  +G_1 D^{-1}\circ \frac{\delta \rho_2}{\delta u}+G_2 D^{-1}\circ
\frac{\delta \rho_1}{\delta u},
\end{array}
\end{equation}
where
$$
\begin{array}{l}
\ds b_1=-6\frac{u_2}{u_1},\qquad \qquad b_2=-9 \frac{u_3}{u_1}-2 \frac{P}{u_1^2}+21 \frac{u_2^2}{u_1^2},\\[6mm]
\ds b_3=-11 \frac{u_4}{u_1}+60 \frac{u_3 u_2}{u_1^2}+14 P \frac{u_2}{u_1^3}-57 \frac{u_2^3}{u_1^3}-3\frac{P'}{u_1},\\[6mm]
\ds b_4=-4 \frac{u_5}{u_1}+38 \frac{u_4 u_2}{u_1^2}+22 \frac{u_3^2}{u_1^2}+99 \frac{u_2^4}{u_1^4}-155 \frac{u_3 u_2^2}{u_1^3}+\frac{34}{3} P \frac{u_3}{u_1^3}-44 P \frac{u_2^2}{u_1^4} \\[3mm]
\ds \qquad +\frac{4}{3} \frac{P^2}{u_1^4}+12 P' \frac{u_2}{u_1^2}-P'', 
\end{array}
$$
$$
\begin{array}{c}
\ds b_5=-2 \frac{u_6}{u_1}+29\frac{u_4 u_3}{u_1^2}+80 P\frac{u_2^3}{u_1^5} +
\frac{23}{3} P' \frac{u_3}{u_1^2} -104 \frac{u_2 u_3^2}{u_1^3} -70\frac{u_4 u_2^2}{u_1^3} +241 \frac{u_2^3 u_3}{u_1^4} +14 \frac{u_5 u_2}{u_1^2}\\[4mm]
\ds \qquad +\frac{20}{3}P \frac{u_4}{u_1^3}-\frac{170}{3} P\frac{u_2u_3}{u_1^4} +
\frac{4}{3} \frac{P' P}{u_1^3}-22P' \frac{u_2^2}{u_1^3} +2P'' \frac{u_2}{u_1} -
\frac{16}{3}P^2 \frac{u_2}{u_1^5} -108\frac{u_2^5}{u_1^5}, \\[8mm]
\ds b_6=\frac{u_7}{u_1}-6\frac{u_2 u_6}{u_1^2} +\frac{8}{9} P^2 \frac{u_2^2}{u_1^6}-
195 \frac{u_3^2 u_2^2}{u_1^4}+6 P\frac{u_3^2}{u_1^4}+\frac{142}{3}P \frac{u_2^4}{u_1^6}+\frac{28}{9}P' P\frac{u_2}{u_1^4} +101 \frac{u_4 u_3 u_2}{u_1^3} \\[4mm]
\ds \qquad +\frac{34}{3} P \frac{u_4 u_2}{u_1^4}-72 \frac{u_2^6}{u_1^6}-\frac{28}{9} P''' u_2+\frac{38}{3} P'' \frac{u_2^2}{u_1^2}-\frac{19}{3} P'\frac{u_4}{u_1^2}-\frac{122}{3} P' \frac{u_2^3}{u_1^4}-10 \frac{u_4^2}{u_1^2}+22 \frac{u_3^3}{u_1^3}\\[4mm]
\ds\qquad -\frac{178}{3} P \frac{u_3 u_2^2}{u_1^5}+\frac{14}{9} P^{(IV)} u_1^2+\frac{113}{3} P' \frac{u_3 u_2}{u_1^3}-\frac{2}{3} P \frac{u_5}{u_1^3}
-\frac{17}{3} P'' \frac{u_3}{u_1}-\frac{4}{3} P^2\frac{u_3}{u_1^5} -89 \frac{u_4 u_2^3}{u_1^4}\\[6mm]
\ds\qquad+236 \frac{u_3 u_2^4}{u_1^5} -13 \frac{u_5 u_3}{u_1^2}+25
\frac{u_5 u_2^2}{u_1^3}-\frac{7}{9}
\frac{P'^2}{u_1^2}-\frac{8}{27} \frac{P^3}{u_1^ 6}-\frac{4}{9}
\frac{P'' P}{u_1^2}.
\end{array}
$$
The operators ${\cal R}_1$ and ${\cal R}_2$ are
related by the elliptic curve
$$
{\cal R}_2^2={\cal R}_1^3-g_2 {\cal R}_1-g_3, \label{rel}
$$
where
$$
\begin{array}{l}
\displaystyle g_2 = \frac{16}{27}\Big((P'')^2-2 P''' P'+2 P^{(IV)} P\Big), \\[4mm]
\displaystyle g_3 = \frac{128}{243}\Big(-\frac{1}{3}(P'')^3-\frac{3}{2}(P')^2P^{(IV)}+P'P''P'''
+2P^{(IV)}P''P-P(P''')^2\Big).
\end{array}
$$\end{theorem}

\begin{remark} It is easy to verify that $g_2$ and $g_3$ are
constants for any polynomial $P(u)$ such that $\hbox{deg}\,P\le 4$.
Under M\"obius transformations of the form
$$
u\rightarrow \frac{\alpha u+\beta}{\gamma u+\delta}
$$
in equation \eqref{KN} the polynomial $P(u)$ changes according
to the same rule as in the differential $\displaystyle \omega =
\frac{d u}{\sqrt{P(u)}}.$ The expressions  $g_2$ and $g_3$ are
invariants with respect to the M\"obius group action.
\end{remark}

\begin{remark} The ratio ${\cal R}_3={\cal R}_2 {\cal
R}_1^{-1}$ satisfies equation \eqref{Lambdaeq}. It belongs to the
skew field of differential operator fractions {\rm \cite{ore}}.
However, this operator is not quasi-local and it is unclear how to
apply it even to the simplest symmetry generator $u_x$.
\end{remark}

\subsection{Hamiltonian operators and bi-Hamiltonian structure for the Krichever--Novikov equation} 

Most of known integrable equations \eqref{eveq} can be written in a
Hamiltonian form
$$
u_t={\cal H}\left(\frac{\delta \rho}{\delta u}  \right),
$$
where $\rho$ is a conserved density and ${\cal H}$ is a Hamiltonian
operator.  The analog of the operator identity \eqref{defR} for Hamiltonian operators is given by
\begin{equation}\label{Heq}
(D_t-F_*) \, {\cal H}= {\cal H} (D_{t}+F_{*}^+),
\end{equation}
which means that $\cal H$ maps cosymmetries to symmetries.

The Poisson bracket  corresponding to the Hamiltonian operator ${\cal H}$ is defined by 
\begin{equation}\label{PDEbrop}
\{f,\,g\}= \frac{\delta f}{\delta u}  \, {\cal H} \Big(\frac{\delta g}{\delta u} \Big).
\end{equation}
\begin{remark}\label{rem216} Actually, the Poisson bracket has to be defined between functionals $\int f\,dx$ and   $\int g\,dx,$ i.e., between equivalence classes of $f$ and $g$ {\rm (see Remark \ref{equi})}.  
\end{remark}
Since the Poisson bracket is defined on the vector space of equivalence classes ${\cal F}/{\rm Im} D$, the Leibniz rule has no sense for 
\eqref{PDEbrop}. The skew-symmetricity and the Jacobi identity for \eqref{PDEbrop} are required. In terms of elements of ${\cal F}$ this means that 
\be \label{HHcond1}
\{f,g\}+\{g,f\} \in {\rm Im}\,D , 
\ee
\be \label{HHcond2}
\{\{f,g\},h\}+\{\{g,h\},f\}+\{\{h,f\},g\} \in {\rm Im}\,D.
\ee

Therefore, besides \eqref{Heq} the Hamiltonian operator for equation \eqref{eveq} should satisfy some 
 identities (see for example, \cite{Olv93}) equivalent to \eqref{HHcond1}, \eqref{HHcond2}.  
\begin{lemma} If operators ${\cal H}_1$ and ${\cal H}_2$ satisfy \eqref{Heq}, then ${\cal R}={\cal H}_2
{\cal H}_1^{-1}$ satisfies \eqref{Lambdaeq}.
\end{lemma}

As a rule, the Hamiltonian operators are local (i.e. differential)
or quasi--local operators. In the latter case
$$
{\cal H}=H+\sum_{i=1}^m G_i D^{-1} \bar G_i,
$$
where $H$ is a differential operator and $G_i, \bar G_i$ are fixed generators of 
symmetries \cite{mokfer, malnov}.

The KdV equation possesses two local Hamiltonian operators  
$$
{\cal H}_1=D, \qquad {\cal H}_2=D^3+4 u D+2 u_x.
$$
The first example 
$$
{\cal H}_0=u_x D^{-1} u_x
$$
of  a quasi--local Hamiltonian operator for the Krichever-Novikov equation \eqref{KN} was found in \cite{sokkn} (see also \cite{demsok}).

The recursion operators ${\cal R}_i$ for \eqref{KN} presented above appear to be the ratios
$$
{\cal R}_1={\cal H}_1 {\cal H}_0^{-1}, \qquad  {\cal R}_2={\cal H}_2 {\cal H}_0^{-1}
$$
of the following quasi--local Hamiltonian operators
$$
\begin{array}{l}
\ds {\cal H}_1= \frac{1}{2}(u_x^2 D^3+D^3\circ u_x^2)+(2 u_{xxx} u_x-\frac{9}{2} u_{xx}^2-\frac{2}{3} P)D+D\circ(2 u_{xxx} u_x-\frac{9}{2} u_{xx}^2-\frac{2}{3} P) \\[4mm]
\ds \qquad +G_1 D^{-1}\circ G_1+  u_x D^{-1}\circ G_2 +G_2 D^{-1}\circ u_x, \\[7mm]
\ds {\cal H}_2= \frac{1}{2}(u_x^2 D^5+D^5\circ u_x^2)+(3u_{xxx} u_x-\frac{19}{2} u_{xx}^2-P)D^3+D^3\circ(3 u_{xxx} u_x-\frac{19}{2} u_{xx}^2-P)\\[4mm]
\ds \qquad +h D+D\circ h+G_1 D^{-1}\circ G_2+G_2 D^{-1}\circ G_1+u_x D^{-1}\circ G_3 +G_3 D^{-1}\circ u_x,
\label{KNH2}
\end{array}
$$
where
$$
h=u_{xxxxx} u_x-9 u_{xxxx} u_{xx}+\frac{19}{2} u_{xxx}^2-\frac{2}{3} \frac{u_{xxx}}{u_x} (5 P-39 u_{xx}^2)+\frac{u_{xx}^2}{u_x^2} (5 P-9 u_{xx}^2)+\frac{2}{3} \frac{P^2}{u_x^2}+u_x^2 P'',
$$
and $G_3={\cal R}_1(G_1)={\cal R}_2 (u_x)$ is the generator of the seventh order
symmetry for \eqref{KN}:
$$
\label{KNG7}
\begin{array}{l}
\dsize G_3=u_7-7 \frac{u_2 u_6}{u_1}-\frac{7}{6} \frac{u_5}{u_1^2}(2 P+12 u_3 u_1-27 u_2^2)-\frac{21}{2} \frac{u_4^2}{u_1}+\frac{21}{2} \frac{u_4}{u_1^3} u_2(2 P-11 u_2^2)\\[4mm]
\dsize\qquad -\frac{7}{3} \frac{u_4}{u_1^2}(2 P' u_1-51 u_2 u_3)  +\frac{49}{2} \frac{u_3^3}{u_1^2}+\frac{7}{12} \frac{u_3^2}{u_1^3} (22 P-417 u_2^2)+\frac{2499}{8} \frac{u_2^4}{u_1^4} u_3\\[4mm]
\dsize\qquad  +\frac{91}{3} P'\frac{u_2}{u_1^2} u_3-\frac{595}{6} P \frac{u_2^2}{u_1^4} u_3-\frac{35}{18} \frac{u_3}{u_1^4} (2 P'' u_1^4-P^2)-\frac{1575}{16} \frac{u_2^6}{u_1^5}+\frac{1813}{24} \frac{u_2^4}{u_1^5} P\\[4mm]
\dsize\qquad -\frac{203}{6}\frac{u_2^3}{u_1^3} P' +\frac{49}{36} \frac{u_2^2}{u_1^5} (6 P'' u_1^4-5 P^2)-\frac{7}{9} \frac{u_2}{u_1^3} (2 P''' u_1^4-5 P P')+\frac{7}{54} \frac{P^3}{u_1^5}\\[4mm]
\dsize\qquad -\frac{7}{9} P'' \frac{P}{u_1}+\frac{7}{9} P'''' u_1^3-\frac{7}{18} \frac{P'^2}{u_1}.
\end{array}
$$
The operators ${\cal H}_i,\, i=1,2,$  were found in \cite{demsok}.
It was verified that they satisfy \eqref{Heq}. It is easy to verify that the Poisson brackets \eqref{PDEbrop} corresponding to ${\cal H}_i$ satisfy identity \eqref{HHcond1}. 

\begin{op} Prove that the Poisson brackets \eqref{PDEbrop} corresponding to ${\cal H}_i$ are compatible and satisfy the Jacobi identity \eqref{HHcond2}.
\end{op}

\begin{remark} Very recently in {\rm \cite{sylv}} it was proved that ${\cal H}_1$ satisfies the Jacobi identity and compatible with  ${\cal H}_0$ .
\end{remark}

\chapter{Integrable hyperbolic equations of Liouville type}\label{Liutype}

    The (open) Toda
  lattices
\begin{equation}\label{toda}
  (u_i)_{xy}=\sum_j A_i^j \exp{(u_j)},
\end{equation}
where $A_i^j$ is the Cartan matrix of a simple Lie algebra \cite{LezSav},  provide examples
of Liouville type systems.

For the Lie algebra of type $A_1$ the system coincides with  the famous Liouville equation
\begin{equation}\label{liou}
u_{xy}=\exp{u}.
\end{equation}
The Liouville equation possesses the following remarkable properties:
\begin{itemize}
\item  {\bf 1.} It has a local
  formula for the general solution
  $$
  u(x,y)=\log{\left( \frac{2 f'(x) g'(y)}{(f(x)+g(y))^2}\right)};
  $$
\item  {\bf 2.} It admits a group of classical symmetries
  $$
  x \rightarrow \phi(x),  \qquad  y \rightarrow \psi(y), \qquad u
  \rightarrow u-\log{\phi'(x)}-\log{\psi'(y)}
  $$
  depending on two arbitrary functions of one variable;
\item {\bf 3.} It possesses the generalized first integrals
 $$
 w=u_{xx}-\frac{1}{2} u_x^2, \qquad \bar w=u_{yy}-\frac{1}{2}
 u_y^2;
 $$
\item {\bf 4.} It has a non-commutative hierarchy of higher infinitesimal symmetries of the form
  $$
  u_{\tau}=(D_x+u_x) \, P(w,w_x,\dots, w_n)
  + (D_y+u_y) \, Q(\bar w,\bar w_y,\dots, \bar w_m),
  $$
where $P$ and $Q$ are arbitrary functions, $n$ and $m$ are arbitrary integers;
\item {\bf 5.} It  has a terminated sequence of the Laplace invariants.
\end{itemize}
These are typical features of the so called  {\it equations of Liouville type} (or,  which is the same, Darboux integrable equations) \cite{darbu1}--\cite{gosse}. 

\section{Generalized integrals}

Consider 
hyperbolic equations of the form
\begin{equation}
\label{hyper} u_{xy}=F(x,y,u,u_x,u_y).
\end{equation}
The corresponding total $x$ and $y$-derivatives are given by the recursive formulas $$
D=\frac{\partial}{\partial x}+\sum_{i=0}^{\infty} u_{i+1}
\frac{\partial}{\partial u_i}+ \sum_{i=1}^{\infty} \bar D^{i-1}(F)
\frac{\partial}
{\partial \bar u_i} $$ and $$ \bar D=\frac{\partial}{
\partial y}+\sum_{i=0}^{\infty} \bar u_{i+1} \frac{\partial}{\partial
\bar u_i}+ \sum_{i=1}^{\infty} D^{i-1}(F) \frac{\partial}{
\partial u_i}, $$ where
\begin{equation}\label{hypdyn}
u_0=\bar u_0=u, \qquad u_1=u_x,  \qquad \bar u_1=u_y, \qquad u_2=u_{xx},  \qquad \bar
u_2=u_{yy}, \dots \,.
\end{equation}
Although at first glance it seems that $D$ is defined in terms of $\bar D$ and vice versa, the vector fields are actually well defined by these formulas. One can verify that $[D,\, \bar D]=0$ in virtue of \eqref{hyper}.

\begin{example} For the Liouville equation \eqref{liou}
we have $$ \bar D=\frac{\partial}{\partial y}+\sum_{i=0}^{\infty} \bar
u_{i+1} \frac{\partial}{\partial \bar u_i} +
 \exp(u)\left( \frac{\partial}{
\partial u_1}+ u_1 \frac{\partial}{
\partial u_2}+ (u_2+ u_1^2) \frac{\partial}{
\partial u_3}+\cdots\right). $$
It is easy to verify that $\ds \bar D\Big (u_2-\frac{1}{2} u_1^2\Big)=0.$
\end{example}

\begin{definition} A function $\,W(x,y,u,u_1,\dots,u_p)\,$ is
called $y$-{\it integral} for equation \eqref{hyper} if $\bar
D(W)=0$. The number $p$ is called the {\it order} of $W$. Analogously,  a function $\bar W(x,y,u,\bar u_1,\dots,\bar u_{\bar
p})$ such that $D (\bar W)=0$ is called $x$-{\it integral}.
\end{definition}

 It is obvious that for any $y$-integral $w$ and any function $S$
 the expression
 \begin{equation}\label{genint}
 W=S\left( x, w, D (w), \cdots, D ^k(w)\right)
 \end{equation}
 is a $y$-integral as well.

 \begin{proposition}\label{prop31} {\rm \cite{zib}}\quad i) Any $y$-integral is a function of the variables $x,y,u,u_1,u_2,\ldots,u_k,\ldots$.
 \newline ii)  Any $y$-integral $W$ has the form
 \eqref{genint}, where $w$ is an integral of minimal possible
 order. The minimal integral is defined uniquely up to any transform $w
 \rightarrow \phi (x,w)$.  \newline iii) If the order $n$ of minimal integral is greater than 1, then there exists a minimal integral $w$ such that  
$\ds \frac{\partial^2 w}{\partial u_n^2}=0$.\footnote{The similar statements are true for $x$-integrals.}
\end{proposition}

\begin{proof} Let a function $W=W(x,y,u,u_1,\ldots,u_n,\bar u_1,
\ldots,\bar u_m)$ be a $y$-integral of equation \eqref{hyper}. Then
$$
  \left(\frac\partial{\partial y}+\bar u_1\frac\partial{\partial u}+
  \bar u_2\frac\partial{\partial\bar u_1}+\cdots+
  \bar u_{m+1}\frac\partial{\partial\bar u_m}\right) W
   +\left(F\frac\partial{\partial u_1}+D(F)\frac\partial{\partial u_2}+\cdots+
  D^{n-1}(F)\frac\partial{\partial u_n}\right)W=0.
$$
Since $F,D(F),\ldots,D^{n-1}(F)$ depend on $x,y,u,\bar u_1,u_1,u_2,\ldots,u_n$, we get 
$$
  \frac{\partial W}{\partial\bar u_m}=0, \qquad
  \frac{\partial W}{\partial\bar u_{m-1}}=0, \quad \ldots \quad 
  \frac{\partial W}{\partial\bar u_1}=0.
$$
Then any $y$-integral is a function that depends on the variables
$x,y,u,u_1,u_2,\ldots,u_k,\ldots$ only.   

Denote  by $w$ a 
$y$-integral of the smallest order $n$. Let 
$W=W(x,y,u,u_1,\ldots,u_m), \,m\ge n$ be another $y$-integral. It can be rewritten in the form
 $$
  W=W(x, y, u, u_1, \ldots, u_{n-1}, w, w_1, \ldots, w_{m-n}).
$$
For any values of  $x, w, w_1,\ldots, w_{m-n}$ this function is a $y$-integral of order less than $n$. Therefore, it is a constant and W does not depend on $y, u, u_1, \ldots,u_{n-1}$. 
 
To prove iii) let us differentiate the relation 
$$
  \bar D w=\left(\frac\partial{\partial y}+\bar u_1\frac\partial{\partial u}+
  F\frac\partial{\partial u_1}+D(F)\frac\partial{\partial u_2}+\cdots+
  D^{n-1}(F)\frac\partial{\partial u_n}\right) w=0
$$
by $u_n$ twice to obtain
$$
  (\bar D+F_{u_1})\frac{\partial w}{\partial u_n}=0, \qquad
  (\bar D+2 F_{u_1})\frac{\partial^2 w}{\partial u_n^2}=0.
$$
This implies that 
$$
  \frac{\partial^2 w}{\partial u_n^2} 
   \, \Big(\frac{\partial w}{\partial u_n}\Big)^{-2}
$$
is a $y$-integral and therefore it is a function of $x$ and $w$. Hence
 $$
  \frac{\partial w}{\partial u_n} H(x,w) = g(x,y,u_1,\ldots,u_{n-1})
$$
for some functions $H$ and $g$ and 
$$
  W=\int H\,d w 
$$
is an integral linear in  $u_n$. 
\end{proof}

\begin{definition} \label{Darb} An equation of the form \eqref{hyper} is called  {\it Darboux integrable} if it has both $x$ and $y$ integrals.
\end{definition}

Some of linear hyperbolic equations are Darboux integrable. 

\begin{example}\label{li1} The first order integrals for the wave equation
$$
  u_{xy}=0
$$
are given by
$$
 w=u_1, \qquad \bar w=\bar u_1.
$$
\end{example}
\begin{example}\label{li2} The Euler-Poisson equation 
$$
  u_{xy}=\frac{u_y-u_x}{x-y}
$$
possesses minimal integrals of second order
$$
 w=\frac{u_2}{x-y}, \qquad \bar w=\frac{\bar u_2}{x-y}.
$$
\end{example}

\section{Laplace invariants for hyperbolic operators}\label{laplin}

Consider a linear
hyperbolic operator
\begin{equation}\label{hypop}
L_0= \frac{\partial^2}{\partial x\partial y}+
  a_0(x,y) \frac\partial{\partial x}
  +b_0(x,y) \frac\partial{\partial y}+c_0(x,y).
\end{equation}
It is easy to verify that $$L_0= \left(\frac {\partial}{\partial x}
+b_0\right)\left(\frac {\partial}{\partial y} +a_0\right)-h_1=
\left(\frac {\partial}{\partial y} +a_0\right)\left(\frac
{\partial}{\partial x} +b_0\right)-k_0,
$$
where
\begin{equation}\label{maininv}
  h_1=\frac{\partial a_0}{\partial x} + b_0 \,a_0-c_0, \qquad
  k_0=\frac{\partial b_0}{\partial y} +a_0 \, b_0-c_0.
\end{equation}
\begin{definition} The functions \eqref{maininv} are called {\it main Laplace invariants} of the operator $L_0.$
\end{definition}
\begin{lemma} Operators $L_0$ and $\bar L$ of the form \eqref{hypop} are related by a gauge transform  
$$\bar L =  
\alpha (x,y)\,L_0 \,\alpha (x,y)^{-1}$$ for some function $\alpha$ iff their main Laplace invariants coincide.
\end{lemma}

As demonstrated below, the functions \eqref{maininv} can be considered as subsequent terms of a sequence. That is why the first of these functions is denoted by $h_1$ (not by $h_0$).
\begin{definition}
The functions $h_1$ and $k_0$ are called  the {\it main Laplace invariants} of the operator \eqref{hypop}.
\end{definition}

 The equation $L_0(V)=0$ is equivalent to the system
 $$
 \left(\frac {\partial}{\partial y} +a_0\right) V = V_1, \qquad
 \left(\frac {\partial}{\partial x} +b_0\right) V_1=h_1 V.
 $$
If  $h_1 \ne 0,$ we can find $V$ from the second equation and substitute into the first one. Therefore, $V_1$ satisfies a hyperbolic equation $L_1(V_1)=0,$ where 
$$
 L_1=\frac{\partial^2}{\partial x\partial y}+
   a_1(x,y) \frac\partial{\partial x}
   +b_1(x,y) \frac\partial{\partial y}+c_1(x,y).
 $$
  The coefficients and the Laplace invariants of the new hyperbolic operator $L_1$ 
 are given by 
$$
 \begin{array}{c}
 a_1=a_0-(\log h_1)_y,  \qquad b_1=b_0, \qquad  c_1=a_1 b_0+(b_0)_y-h_1, \\[3mm]
 h_2 =(a_1)_x - (b_0)_y + h_1, \qquad k_1 = h_1.
 \end{array}
$$
 \begin{remark}
 The invariant $h_2$ can be rewritten in terms of the main invariants of the operator \eqref{hypop} only:  $$\, h_2 = 2 h_1 - k_0 - (\log h_1)_{xy}.$$
 \end{remark}

 If $h_2 \ne 0,$ we may continue. As a result, we get a
 chain of hyperbolic operators
 $$
   L_i = \frac{\partial^2}{\partial x\partial y}+
   a_i \frac\partial{\partial x}
   +b_i \frac\partial{\partial y}+c_i, 
 $$
where
 $  i\in \N$ and
 $$
 a_i=a_{i-1}-(\log h_i)_y, \qquad b_i=b_0, \qquad c_i=a_i
 b_0+(b_0)_y-h_i,  $$
 \begin{equation}\label{hinv}
 k_i = h_i, \qquad h_{i+1}=2 h_i-h_{i-1} - (\log h_i)_{xy}.
\end{equation}

The initial linear hyperbolic equation $L_0(V)=0$ can be also  rewritten as
$$
 \left(\frac {\partial}{\partial x} +b_0\right) V=V_{-1}, \qquad
 \left(\frac {\partial}{\partial y} +a_0\right) V_{-1}=k_0 V.
 $$
 If $k_0\ne 0,$
 then $V_{-1}$ satisfies the following equation
$$
   \left(\frac{\partial^2}{\partial x\partial y}+
   a_{-1} \frac\partial{\partial x}
   +b_{-1} \frac\partial{\partial y}+c_{-1}\right) \, V_{-1}=0,
 $$
 etc. We get the chain of operators
$$
   L_{-i} = \frac{\partial^2}{\partial x\partial y}+
   a_{-i} \frac\partial{\partial x}
   +b_{-i} \frac\partial{\partial y}+c_{-i},
   \qquad i\in \N,
 $$
 where
 $$
 a_{-i}=a_0, \qquad b_{-i}=b_{-(i-1)}-(\log k_{-(i-1)})_x, \qquad 
 c_{-i}=a_0 b_{-i}+(a_0)_x-k_{-(i-1)},$$
 \begin{equation}
 \label{kinv}
h_{1-i}=k_{-(i-1)}, \qquad k_{-i}=h_{-i}=2 h_{-(i-1)}-h_{-(i-2)}-(\log h_{-(i-1)})_{xy}. 
 \end{equation}
 \begin{definition}
 The functions $h_i, \,\, i\in \Z $ defined by \eqref{hinv}, \eqref{kinv} are called {\it the Laplace invariants of the operator \eqref{hypop}}.
 \end{definition}
 
The sequence of the Laplace invariants is uniquely defined by
the following recursive formula  
\begin{equation}\label{HHH}
h_i=2h_{i-1}-h_{i-2}-(\log h_{i-1})_{xy}, \qquad i\in \Z
\end{equation}
and by the initial data
\begin{equation}\label{INI}
  h_1=\frac{\partial a_0}{\partial x} +a_0 b_0-c_0, \qquad
  h_{0}=\frac{\partial b_0}{\partial y} +a_0 b_0-c_0.
\end{equation}
Remarkably, \eqref{HHH} is nothing else but the integrable infinite $A$-Toda lattice.

If one of the Laplace invariants vanishes, then the next invariant is not defined and the sequence of invariants terminates. If it terminates in both directions:  $h_r=h_{-s}=0$ for $r\ge 1, s\ge 0,$ then the sequence becomes finite and \eqref{HHH} turns out to be the open Toda lattice \cite{LezSav}.  In this case the equation  $L_0(V)=0$ can be solved explicitly (see, for example, \cite{zibsok}).

\section{Nonlinear hyperbolic equations of Liouville type}

The linearization operator for \eqref{hyper} is given by the formula
\begin{equation}\label{linop}
L= D \bar D - \frac{\partial
F}{\partial u_1}\, D - \frac{\partial F}{\partial \bar u_1}\, \bar D -\frac{\partial F}{\partial u}. 
\end{equation}
This is a linear hyperbolic operator of the form \eqref{hypop}, where the partial derivatives are replaced by the total derivative operators $D$ and $\bar D$ and the coefficients become functions of $x, y$ and finite number of variables \eqref{hypdyn}.

In accordance with \eqref{INI},  we define
the main Laplace
invariants for equation \eqref{hyper}
as $$ H_1\stackrel{def}{=}-D \left(\frac{\partial
F}{\partial u_1}\right) +\frac{\partial F}{\partial u_1}
\frac{\partial F}{\partial \bar u_1}+\frac{\partial F}{\partial
u}, \qquad 
H_0\stackrel{def}{=}-\bar D \left(\frac{\partial F}{\partial \bar
u_1}\right) +\frac{\partial F}{\partial u_1} \frac{\partial
F}{\partial \bar u_1}+\frac{\partial F}{\partial u}.
$$
The invariants $H_i$ for $i>1$ and for $i<0$ are determined from
$$
D  \bar D(\log H_i)=-H_{i+1}-H_{i-1}+2 H_i, \qquad i\in \Z.
$$
For the Liouville equation we have $H_0=H_1=\exp{u}$. It is easy
to verify that
$
H_2=H_{-1}=0.
$

\begin{definition}\label{LioInt} We call \eqref{hyper} an {\it equation of Liouville
type} if there exists $r\ge 1$ and $s\ge 0$ such that
$$
H_r=H_{-s}\equiv 0.
$$
\end{definition}
Notice that this is a much more constraining property than the existence of the generalized integrals used by Darboux since Definition \ref{LioInt} involves only the right-hand side of the equation.

\begin{remark} For linear hyperbolic equations  the Laplace invariants coincide with the Laplace invariants of the corresponding linear operator {\rm (see Section \ref{laplin})}.  In particular equations form Examples \ref{li1}, \ref{li2} are Liouville integrable. 
\end{remark}

\begin{example} For the equation
\begin{equation}\label{pryamoy}
u_{xy}=\frac{1}{u} \sqrt{1-u_{\mathstrut x}^2}
\sqrt{1-u_{\mathstrut y}^2}
\end{equation}
we have $H_2=H_{-1}=0.$
\end{example} 
\begin{example} For the equation
\begin{equation}\label{kosliu}
u_{xy}= u\,u_y
\end{equation}
the invariants $H_3$ and $H_{0}$ are equal to zero. 
\end{example}
\begin{example} For the equation
\begin{equation}\label{dliny}
u_{xy}=-\frac{2 k}{x+y} \sqrt{u_{\mathstrut y}}
\sqrt{u_{\mathstrut y}}, \qquad k\in \N
\end{equation}
we have $H_{k+1}=H_{-k}=0.$
\end{example}

\begin{theorem} \label{th61} {\rm \cite{zibsok1, andkam1}} The equation \eqref{hyper} has
non-trivial both $x$ and $y$-integrals iff it is an equation of
Liouville type.
\end{theorem}

\begin{proposition}\label{prop62} For any Liouville type equation \eqref{hyper} there exist
functions $\psi(x,y,u,u_1,\dots,u_{p})$  and  $\,\, \bar\psi(x,y,u,\bar
u_1,\dots,\bar u_{\bar p})$ such that
$$
\frac{\partial F}{\partial u_1}=\bar D \, \log
\psi(x,y,u,u_1,\dots,u_{p}), \qquad \frac{\partial F}{\partial \bar
u_1}= D \, \log \bar \psi(x,y,u,\bar u_1,\dots,\bar u_{\bar p}).
$$
\end{proposition}
\begin{remark} \label{zib} Another statement of this kind is the following:
$$
\frac{\partial F}{\partial u_1}
\frac{\partial F}{\partial \bar u_1}+\frac{\partial F}{\partial
u} = \bar D \, 
\phi(x,y,u,u_1,\dots,u_{s}) =  D \,   \bar \phi(x,y,u,\bar u_1,\dots,\bar u_{\bar s})
$$ 
for some functions $\phi$ and $\bar \phi$.
\end{remark}

\begin{theorem}\label{th62} For any Liouville type equation \eqref{hyper} the
evolution equation
\begin{equation}\label{hypsym}
u_{\tau}={\cal M}\,\left[ Q\left( x, w, D (w), \cdots, D ^k(w)\right) \right], \qquad k\ge 0,
\end{equation}
where
\begin{equation}\label{MM}
{\cal M}=\bar\psi \displaystyle \frac{1}{H_1}D \circ\frac{1}{H_2}\cdots
D \circ \frac{1}{H_{r-1}}D \circ \frac{\psi H_1\cdots H_{r-1}}{\bar\psi},
\end{equation}
$w$ is the minimal $y$-integral and $Q$ is an arbitrary function, is an infinitesimal symmetry.
\end{theorem}
\begin{remark}
For a generic function $Q$ the evolution equation \eqref{hypsym} is not integrable. 
\end{remark}

For a generalization of Theorem \ref{th62} to the case of the Darboux integrable multi-component systems see \cite{SokSt}.

\begin{theorem}\label{th63} For any Liouville type equation all coefficients
of the differential operator \begin{equation}\label{LLLop} {\cal L}=\frac{\bar\psi}\psi
H_{0}H_{-1}\cdots H_{1-s}D\circ \frac 1{H_{1-s}}\circ D \cdots
\frac 1{H_0}\circ
D\circ \frac 1{H_1}\cdots D \circ \frac 1{H_{r-1}}D\circ \frac{\psi H_1\cdots
H_{r-1}}{\bar\psi}
\end{equation}
are $y$-integrals.
\end{theorem}

\begin{example}\label{ex63} For the Liouville equation
\eqref{liou} we have
  $$
  {\cal L}= \exp(u) D \circ \exp(-u) D \circ \exp(-u) D \circ \exp(u)=D ^3+2 w D +w_x,
$$ where $\,\, w=u_{xx}-\frac{1}{2} u_x^2.
  $
  The operator ${\cal M}$ is given by
  $${\cal M}=\exp(-u) D \circ \exp(u)=D+u_x.
  $$
  If $Q(x,w,\dots)=w,$ then the corresponding symmetry is the integrable evolution equation
$$
  u_{\tau}=u_{xxx}-\frac{1}{2} u_x^3.
$$
\end{example}

An attack to the problem of  a complete classification of the Darboux integrable equations has been made in \cite{zibsok}. The proof of the classification statement consisted of more than 150 pages and has been not published. However, O. Kaptsov pointed out to an omission in the classification.  In Appendix 2  Darboux integrable equations known to the author are collected.

\section{Integrable multi-component systems of Liouville type}

Consider multi-component systems of the form
\begin{equation}\label{hhh}
\vec {\bf u}_{xy}=\vec {\bf F}(x,y,\vec {\bf u},\vec {\bf u}_{x},\vec {\bf
u}_{y})\, \qquad \vec{\bf u}=(u^1,\dots,u^N). 
\end{equation}
 For the system \eqref{hhh} the coefficients of the linearization operator \eqref{linop} 
are $N\times N$-matrices.

Evidently, most part of the definitions, constructions and statements about
the Liouville type equations presented above can be generalized to the case
of systems \eqref{hhh}. However a serious problem arises in
the definition of the Laplace invariants.

The linearization operator \eqref{linop} becomes an operator of the form
\begin{equation}\label{LMat}
L = D \bar D+a D+b \bar D+c
\end{equation}
with matrix coefficients. A straightforward
generalization of all definitions to the matrix case looks as
follows. The main Laplace invariants are defined by the formulas
$$
H_1= D(a)+b\,a-c, \qquad H_0=\bar D(b)+a\, b-c.
$$
Now they are $N\times N$-matrices.
The matrices $H_i$ for $i>1$ are recurrently determined from the following
system of equations
\begin{equation}\label{matlap1}
\bar D H_i-H_i\, a_{i-1}+a_{i}\,H_i=0,
\end{equation}
\begin{equation}\label{matlap2}
H_{i+1}=2 H_i+D(a_{i}-a_{i-1})+[b, \, a_{i}-a_{i-1}]-H_{i-1},
\end{equation}
where $a_0=a$. Obviously, in the scalar case these formulas coincide with
the corresponding formulas from \eqref{hinv}.

Suppose that the matrices $H_i$ for $i\le k$ and $a_i$ for $i\le k-1$ are
already given. Then we derive $a_{k}$ from \eqref{matlap1} and after that
find $H_{k+1}$ from \eqref{matlap2}. However if $\det H_k = 0$, then
$a_{k}$ does not exist at all or it is defined not uniquely but
up a matrix $\alpha$ such that $\alpha \, H_k=0$.
In the last case, the existence and properties of next Laplace invariants could
essentially depend on the choice of $\alpha$.

The degeneration  $\det H_k = 0$ for some $k$ is typical for
the open Toda lattices.

\begin{example} Consider the $A_2$-Toda lattice:
$$
u_{xy}=-2 \exp u+ \exp v, \qquad v_{xy}=\exp u-2 \exp v.
$$
The linearization operator has the form
$$
\label{defns}
D \bar D + \begin{pmatrix} 2 \exp u & -\exp v \cr -\exp u & 2 \exp v\cr\end{pmatrix}.
$$
In this case
$$
 h= \begin{pmatrix}-2 \exp u & \exp v \cr \exp u & -2 \exp v \end{pmatrix}
$$
is a non-degenerate matrix. Using \eqref{matlap1}, we get
$$
a_1=\frac 13 \, \begin{pmatrix}-4 u_y+v_y& -2 u_y+2 v_y \cr 2 u_y-2 v_y
& u_y-4 v_y\cr\end{pmatrix},
$$
and
$$
h_1=
\begin{pmatrix}\exp u-2 \exp v &2 \exp u -\exp v \cr -\exp u+2 \exp v &-
2 \exp u+\exp v\cr\end{pmatrix}.
$$
We see that $\det h_1=0$. 
\end{example}

It is not difficult to prove (\cite{star}) the following general
statement (compare with Theorem \ref{th61}):
\begin{theorem}
Suppose that \eqref{hhh} has non-degenerate $y$ and $x$-integrals
$$
W(x,y, \vec u, \vec u_x,\dots, \vec u_p), \qquad
\bar W(x,y, \vec u,  \vec{u}_y,\dots, \vec{u}_{\bar p});
$$
then  $\det H_r=\det H_{-s} = 0$ for some
$r \le \bar p$ and $s \le p -1.$
\end{theorem}
According to the theorem, for systems of Toda type some of Laplace invariants must be
degene\-rate and we are confronted with the question of how to well define
a chain of Laplace invariants.

Let us set
$$
Z_k=H_k \, H_{k-1} \cdots H_1.
$$
It follows from \eqref{matlap1} that
\begin{equation}\label{ZZ1}
Z_k \, (\bar D +a)=(\bar D +a_k) \, Z_k.
\end{equation}
If the matrices $H_i$ (and hence $Z_k$) for $i\le k$
and $a_i$ for $i\le k-1$ are already given, then we define $a_{k}$ from
\eqref{ZZ1} and after that find $H_{k+1}$
\begin{equation}\label{ZZ2}
H_{k+1}= (D+b) \, a_k-(\bar D +a_k)\, b+H_k.
\end{equation}
The latter relations are equivalent to \eqref{matlap2}. The matrix $a_k$
is defined up to arbitrary matrix $\alpha$ such that $\alpha Z_k = 0.$

\begin{theorem} {\rm \cite{zibsok,ZibSt}}
Suppose that $H_i$ for $i\le k$ are already known and for $i<k$ the
following conditions
\begin{equation}\label{main1}
(\bar D+a) \, \Big( {\rm ker} Z_i\Big) \subset {\rm ker} Z_i
\end{equation}
and
\begin{equation}\label{main2}
 (D-b^t) \, \Big({\rm ker}\, Z_i^t\Big) \subset {\rm ker}\, Z_i^t
\end{equation}
are fulfilled. Then $a_k$ exists iff condition \eqref{main1} is
fulfilled for $i=k$. Further, $Z_{k+1}$ does not depend on the arbitrary
matrix $\alpha$, which appears in the general formula for $a_k$,
iff condition \eqref{main2} holds for $i=k$.
\end{theorem}
\begin{proof}
It follows from \eqref{ZZ1} that \eqref{main1} with $i=k$ is a necessary
condition for the existence of the $a_k$. The sufficiency follows from the
Kronecker-Capelli theorem. The formula \eqref{matlap2} implies the condition
\eqref{main2}.
\end{proof}

\begin{remark} In the case $a=b=0$ in \eqref{LMat}, conditions \eqref{main1}, \eqref{main2}
are fulfilled iff {\rm (see {\rm \cite{ZibSt}})} the vector spaces ${\rm ker}\, Z_k$ and ${\rm ker}\, Z_k^t$
admit bases,  consisting of vectors from $\ker
\bar{D}$ and $\ker D$, respectively.
\end{remark}

\begin{example} For the open $A_3$-Toda lattice
$$
\begin{cases}
(u_1)_{xy}=2 \exp{u_1}-\exp{u_2}, \\[3mm]
(u_2)_{xy}=-\exp{u_1}+ 2 \exp{u_2}-\exp{u_3}, \\[3mm]
(u_3)_{xy}=-\exp{u_2}+2 \exp{u_3},
\end{cases}
$$
all matrices $Z_k$ are uniquely defined and ${\rm rank}\, Z_k=4-k$.
In particular, $Z_4=0.$ The vector ${\bf e}_1=(1,1,1)^t$ forms a basis
of ${\rm ker}\, Z_2$. A basis of ${\rm ker}\, Z_3$ can be chosen as follows:
${\bf e}_1$ and ${\bf e}_2=(1,0,-1)^t$. For bases of ${\rm ker}\, Z_2^t$ and ${\rm ker} Z_3^t$
can be taken ${\bf f}_1=(3,4,3)^t$ and ${\bf f}_1$, ${\bf f}_2=(1,0,-1)^t$. Thus for
this Toda lattice, the vector spaces ${\rm ker} Z_k$ and ${\rm ker}\, Z_k^t$
admit constant bases and therefore conditions \eqref{main1}, \eqref{main2}
hold.
\end{example}

\begin{example} For the $C_3$-Toda lattice
$$
\begin{cases}
(u_1)_{xy}=2 \exp{u_1}-\exp{u_2}, \\[3mm]
(u_2)_{xy}=-\exp{u_1}+ 2 \exp{u_2}-\exp{u_3}, \\[3mm]
(u_3)_{xy}=-2 \exp{u_2}+2 \exp{u_3},
\end{cases}
$$
the matrices $Z_k$ are uniquely defined and ${\rm rank}\, Z_1=3,$ \,
${\rm rank}\, Z_2=2,$\, ${\rm rank}\, Z_3=2,$\, ${\rm rank}\, Z_4=1,$ \,
${\rm rank}\, Z_5=1,$ and $Z_6=0.$ All vector spaces ${\rm ker}\, Z_k$
and ${\rm ker} Z_k^t$ admit constant bases.
\end{example}
\begin{example} In the case of $D_3$-Toda lattice
$$
\begin{cases}
(u_1)_{xy}=2 \exp{u_1}-\exp{u_2}-\exp{u_3}, \\[3mm]
(u_2)_{xy}=-\exp{u_1}+ 2 \exp{u_2}, \\[3mm]
(u_3)_{xy}=- \exp{u_1}+2 \exp{u_3},
\end{cases}
$$
we have ${\rm rank}\, Z_k=4-k.$ All vector spaces ${\rm ker} Z_k$
and ${\rm ker} Z_k^t$ admit constant bases.
\end{example}

\begin{conjecture} {\em \cite{zibsok}} For any open Toda lattice \eqref{toda}, the indices
$i$ such that the rank of $Z_i$ is decreased, coincide with the exponents
of the corresponding simple Lie algebra and the index $h$ such that $Z_h=0$
is equal to the Coxeter number.
\end{conjecture}
For the classical simple Lie algebras the conjecture was proved in \cite{gurzib}. Later A.M. Guryeva verified it for the exceptional Lie algebras. In the following example we present formulas for $Z_k$ from \cite{gurzib} in the $A_n$-case.
\begin{example} \label{An} We rewrite the $A_n$-Toda lattice with the Cartan matrix
$$A=\left( \begin{array}{ccccccccc}
\;\;\;2  & -1      & \;\;\;0 & \;\;\;0 & \ldots & \;\;\;0 & \;\;\;0 & \;\;\;0 \\
-1       & \;\;\;2 & -1      & \;\;\;0 & \ldots   & \;\;\;0 & \;\;\;0 & \;\;\;0 \\
\;\;\;0  & -1      & \;\;\;2 & -1      & \ldots   & \;\;\;0 & \;\;\;0 & \;\;\;0 \\
\;\;\;.  & \;\;\;. & \;\;\;. & \;\;\;. & \ldots & \;\;\; .& \;\;\;. & \;\;\;. \\
\;\;\;0  & \;\;\;0 & \;\;\;0 & \;\;\;0 & \ldots & -1      & \;\;\;0 & \;\;\;0  \\
\;\;\;0  & \;\;\;0 & \;\;\;0 & \;\;\;0 & \ldots & \;\;\;2 & -1      & \;\;\;0  \\
\;\;\;0  & \;\;\;0 & \;\;\;0 & \;\;\;0 & \ldots & -1      & \;\;\;2 & -1    \\
\;\;\;0  & \;\;\;0 & \;\;\;0 & \;\;\;0 & \ldots & \;\;\;0 & -1 &
\;\;\;2
\end{array} \right ) $$
in the form
$$
 D \bar D\,{\bf u}=A\,U\,{\bf c},
$$
where $$\; {\bf u}=(u_{1},u_{2},u_{3}, \ldots,u_{n-1} ,u_{n})^{t}, \qquad 
 {\bf c}=(1,1,1, \ldots, 1,1)^{t},  $$
$$U={\rm diag}\Big(\exp(u_{1}),\, \exp(u_{2}), \ldots ,\,\exp(u_{n})\Big). $$ 
 
 The linearization operator is given by 
$$
L=D\overline{D}-A\,U.  
$$
For this operator we have
$$  Z_{k}=AJ^{1-k}S_{k}\left(J^{t}\right)^{1-k},
\quad \quad k=1,2, \ldots ,n. 
$$
Here,
$$ S_{k}={\rm diag} \left\{0,0, \ldots , 0, \, \exp \Big(\sum_{i=1}^{k}u^{i}\Big), \, \exp \Big(\sum_{i=2}^{k+1}u^{i}\Big), \ldots, \,
\exp \Big(\sum_{i=n-k}^{n-1}u^{i}\Big),\,  \exp \Big(\sum_{i=n-k+1}^{n}u^{i}\Big)\right\}
$$
and
$$J=\left( \begin{array}{cccccc}
1&1&1&\ldots&1&1 \\
0&1&1&\ldots&1&1\\
 \;\;\;.  & \;\;\;. & \;\;\;. & \;\;\;. & \ldots & \;\;\; .  \\
  0  &  0 & 0 & \ldots & 1 &1 \\
 0  &  0 & 0 & \ldots & 0 &1 \\
 \end{array} \right ). $$
It is clear that $S_{n+1}=0$.

Notice that ${\rm rank}\, Z_{k}=n-k+1\,\,$, $\, k=1,2,\ldots.$ \newline Thus
the numbers $k$ such that ${\rm rank}\, Z_{k+1}<{\rm rank}\, Z_{k}$
coincide with the exponents $1,2, \ldots, n \,$ for the Lie algebra of $\,
A_{n}, \,$ and the number $\, k=n+1 \,$ such that $\,Z_{k}=0 \,$ is
equal to the Coxeter number.
\end{example}
\begin{remark} It would be interesting to understand an algebraic meaning of the matrices, which appear in 
Example \ref{An} and generalize them to any simple Lie algebra. 
\end{remark}

For $i > 0$ the invariants $H_{-i}$ and the matrices $Z_{-i}=H_{-i} H_{-(i-1)} \dots H_0$
  are given by the formulas
$$
D(Z_{1-i}) - Z_{1-i} b + b_{-i} Z_{1-i} = 0,
$$
$$
 H_{-i} = 2 H_{1-i} + \bar{D} (b_{-i} - b_{1-i}) + [a,b_{-i} - b_{1-i}] - H_{2-i}.
$$
The conditions \eqref{main1} and \eqref{main2} have to be replaced by 
$$
( D+b) \, \Big( {\rm ker} Z_{-i}\Big) \subset {\rm ker} Z_{-i}, \qquad 
 (\bar{D}-a^t) \, \Big({\rm ker}\, Z_{-i}^t\Big) \subset {\rm ker}\, Z_{-i}^t .$$


The termination of the sequence $Z_i$ can be taken as a definition of
Liouville type systems \eqref{hhh}.
\begin{definition}\label{Lio}
Suppose that for a system of hyperbolic equations of the form \eqref{hhh}
all conditions \eqref{main1} and \eqref{main2} are fulfilled and
there exist $r\geq 1$ and $s \geq 0$ such that $Z_r=Z_{-s}\equiv 0$; then
\eqref{hhh} is called a {\it system of Liouville type}.
\end{definition}

In the scalar case Theorem \ref{th61} shows that the equation \eqref{hyper} is Darboux integrable (see Definition \ref{Darb}) iff it is an equation of Liouville type. For the multi-component systems this is not true. In \cite{ZibSt} an example of a system, which is Darboux integrable but is not a system of Liouville type  has been constructed.

\begin{op} Prove that any system of Liouville type is Darboux integrable.
\end{op}

Using Definition \ref{Lio}, let us classify all Liouville type systems of the form
\begin{equation}\label{togagen2}
\begin{cases}
(u_1)_{xy}=2 \exp{u_1}+k_1 \exp{u_2}, \\[3mm]
(u_2)_{xy}=k_2 \exp{u_1}+ 2 \exp{u_2}
\end{cases}
\end{equation}
with non-degenerate ($k_1 k_2 \ne 4$) and non-diagonal Cartan matrix.

It is easy to verify that for any $k_1,\, k_2$ we have
${\rm rank}\, Z_1=2,$ \,${\rm rank}\, Z_2=1.$ Further, $Z_3=0$ iff
$k_1=k_2=-1$ (the open $A_2$-Toda lattice).

Let $Z_3 \ne 0$, then condition \eqref{main1} with $i=3$ is fulfilled
iff $k_1=-1$ or $k_2=-1.$ Without loss of generality, we set $k_1=-1$.
Then $Z_4=0$ iff $k_2=-2$ ($C_2$-Toda lattice).

If $k_2 \ne -2$, then ${\rm rank}\, Z_5=1$ and condition \eqref{main1}
with $i=5$ holds iff $k_2=-3$. In this case $Z_6=0$ ($G_2$-Toda lattice).

Thus we proved that all Liouville type systems \eqref{togagen2} are exhausted
by the Toda lattices corresponding to the simple Lie algebras of rank 2.

\begin{exercise} Prove a similar statement for the case of rank 3. 
\end{exercise}

\section{Differential substitutions and Liouville type equations }

In this section we consider differential substitutions (see Section \ref{miur}) of the form 
$$
\hat u = P(x, u, u_1,\dots, u_{k}),
$$
which connect two evolution equations 
$$
u_{t}= f(x, u, u_1, \dots, u_{n})
$$
and
$$
\hat u_{t}= g(x, \hat u, \hat u_1, \dots ,\hat u_{n}).
$$
 
  Suppose we have a Liouville type equation \eqref{hyper} with
 minimal $y$-integral
 $$\hat u=w(x, u, u_1, \dots, u_{p}).$$
 Let $$u_{\tau}=f(x, u, u_1, \dots, u_{n})$$ be a symmetry for
 \eqref{hyper}.

 Since the total derivatives $D_y$ and $D_{\tau}$ commute, we have
 $$
 D_y D_{\tau}(\hat u)=D_{\tau} D_y (\hat u)=0
 $$
 i.e. $D_{\tau}(\hat u)$ is a $y$-integral as well. According to Proposition \ref{prop31}, we have
 $$
 \hat u_{\tau}=Q(x, \hat u, \hat u_1, \dots, \hat u_n)
 $$
 for some function $Q$. Thus, the minimal $y$-integral of any hyperbolic equation \eqref{hyper} defines a differential substitution from $x$-symmetries of \eqref{hyper} to some evolution equations. Similarly, the minimal $x$-integral defines a differential substitution from
 $y$-symmetries.

\begin{example} The Liouville equation has a symmetry
\begin{equation}\label{LiuSym}
u_{\tau}=u_{xxx}-\frac{1}{2} u_x^3.
\end{equation}
The minimal $y$-integral $$\hat u=u_{xx}-\frac{1}{2} u_x^2$$
defines the differential substitution from \eqref{LiuSym} to the KdV-eqution
$$
\hat u_{\tau}=\hat u_{xxx}+3\, \hat u \hat u_x.
$$
\end{example}
\begin{example} Consider the Liouville type equation \eqref{kosliu}. 
Its minimal $y$ and $x$-integrals are given by
$$
w=u_x-\frac12 u^2, \qquad \bar w=\frac{u_{yyy}}{u_y}-\frac32
\frac{u_{yy}^2}{u_y^2}.
$$
The first of these substitutions is just the Miura transformation. 

General formulas from Theorems \ref{th62} and \ref{th63} give us
$$
{\cal M}=D^2+u D+u_x, \qquad \bar {\cal M}=u_y, \qquad 
{\cal L}=D^3+2 w D+w_x.
$$
The simplest $x$ and $y$-symmetries are given by 
$$
u_{\tau}={\cal M} w = u_{xxx}-\frac{3}{2} u^2 u_x
$$
and
$$
u_{\tau}=\bar {\cal M} \bar w=u_{yyy}-\frac{3}{2}
\frac{u_{yy}^2}{u_y}
$$
are well-known integrable evolution equations. It is easy to
verify that $$w_{\tau}=w_{xxx}+3\, w w_x$$ and $\bar w_{\tau}=\bar
w_{yyy}+3\, \bar w \bar w_y.$

 Moreover, for any $x$-symmetry of equation \eqref{kosliu} given by 
\begin{equation}\label{utau}
 u_{\tau}=\Big( D^2+u D+u_x\Big) \, G(x, w, w_x, \dots, w_n)
 \end{equation}
 the minimal $y$-integral $$w=u_x-\frac12 u^2$$ defines the
 differential substitution from \eqref{utau} to an equation of the form   
 $w_{\tau}=Q(x, w, w_x, \dots)$.

 Let us find  the function $Q$ explicitly. We have
$$
\begin{array}{c}
w_{\tau}=\Big( D-u\Big) \, u_{\tau}= \Big( D-u\Big) \,
\Big( D^2+u D+u_x\Big) \, G(x, w, w_x, \dots, w_n)= \\[3mm]
\qquad \Big( D^3+2 w D+w_x\Big) \, G(x, w, w_x, \dots, w_n)=
{\cal L} \, G(x, w, w_x, \dots, w_n).
\end{array}
$$
\end{example}

Arguing as above, we prove the following
\begin{proposition}\label{prop63}
Let $w$ be a minimal $y$-integral of an equation of Liouville type and
let
\begin{equation}\label{msym}
u_{\tau}={\cal M} \, G(x,w,w_x,\dots),
\end{equation}
be a symmetry of this equation. Then equation \eqref{msym} is related to
an evolution equation
\begin{equation} \label{weq}
\hat u_{\tau}= w_* \, {\cal M} \,\Big( G(x, \hat u, \hat u_x,\dots)\Big)
\end{equation}
by the differential substitution
\begin{equation}\label{miu}
\hat u=w(x,u,u_x,\dots, u_k).
\end{equation}
Here, $w_*$ is the Fr\'echet derivative of the minimal $y$-integral $w$.
The possible freedom $w\rightarrow f(x,w)$ in the choice of the minimal
integral corresponds to a point transformation in equation \eqref{weq}.
\end{proposition}

\begin{remark} In particular, this proposition states that the coefficients of the differential operator $ w_* \, {\cal M} $
are $y$-integrals.
\end{remark}
 
\begin{remark}\label{rem63} For all known examples the operator $w_* \, {\cal M}$
coincides  with the operator
${\cal L}$  given by \eqref{LLLop}. If this true, it follows from \eqref{MM}, \eqref{LLLop}
 that the formula \footnote{If $H_0=0$, then $w_* = \frac{\bar{\psi}}{\psi} \, D (\frac{1}{\bar\psi}).$}

$$
w_*=\frac{\bar\psi}\psi H_{0}H_{-1}\cdots H_{1-s}D \frac 1{H_{1-s}}D
\cdots \frac 1{H_0}D \frac 1{\bar \psi }
$$
must be valid. This formula was proved in {\rm \cite{zibsok}} under the
assumption that the order $k$ of minimal integral is equal to 
${\rm ord}\, {\cal L}-{\rm ord}\, {\cal M}.$
\end{remark}

 \subsection{Differential substitutions of first order }\label{sec351}
 
We associate the hyperbolic equation
\begin{equation}\label{hypmiur1}
u_{xy}=-\frac{P_{u}}{P_{u_{x}}}\, u_y
\end{equation}
with any differential substitution of the form
$$
v = P(x,u,u_x).
$$
It is clear that $P$ is the minimal $y$-integral for  
\eqref{hypmiur1}. It is easy to verify that  $H_0=0$ for this equation.

\begin{observation} For all known differential
substitutions of the form \eqref{hypmiur1} is an equation of the Liouville
type {\rm (}i.e. $H_r=0$ for some $r>0${\rm )}.
\end{observation}

 \begin{theorem} {\rm \cite[Lemma 4.1]{zib}, \cite{star1}} Equation \eqref{hypmiur1} is a
 Liouville type equation iff {\rm (}up to a transformation of the form
 $u\rightarrow f(x,u)$\,{\rm )}  the function $P(x,u,u_x)$ is given by
 \begin{equation}\label{miur3}
 u_x=\alpha(x,P)\, u^2+\beta(x,P)\, u+\gamma(x,P)
 \end{equation}
 for some functions $\alpha, \beta$ and $\gamma$. Equation
 \eqref{hypmiur1} with $P$ defined by \eqref{miur3} has the  $x$-integral
 $$
 \bar W=\frac{u_{yyy}}{u_y}-\frac32 \frac{u_{yy}^2}{u_y^2}.
 $$
 \end{theorem}

\begin{example} For the Miura substitution we have
 $\ds u_x=P+\frac{1}{2} u^2.$
\end{example}

 \begin{example} For well-known differential substitution \cite{cvsokyam} $$v=u_x+\exp
 {(u)}+\exp {(-u)}$$ the corresponding hyperbolic equation is
 given by
 \begin{equation}\label{miur2}
 u_{xy}=\Big( \exp {(-u)}-\exp {(u)}\Big) \, u_y.
 \end{equation}
The function $P$ for \eqref{miur2}, after the
 transformation
  $u\rightarrow \ln u,$ satisfies the following relation of the form \eqref{miur3}: $u_x=-u^2+P u-1.$
\end{example}
\begin{example}
 For the Cole-Hopf substitution \eqref{colehopf}
 we have $u_x=P u$.
\end{example}

\subsection{Integrable operators}

Consider the set ${\cal F}$ of all functions depending on
$$
u_0=u, \quad u_1=u_x, \quad u_2=u_{xx}, \quad \dots.
$$
For any $f(u,u_1,\dots,u_n)$ we denote
$$
f_*\stackrel{def}{=} \sum_{i=0}^{n} \frac{\pa f}{\pa u_i}\, D^i.
$$
The Lie bracket
\begin{equation} \label{Lie}
\Big[f, \,\, g\Big] \stackrel{def}{=} g_* \, (f)-f_* \, (g)
\end{equation}
equips ${\cal F}$ with a structure of Lie algebra (see \eqref{comev}). Bracket \eqref{Lie}
corresponds to the commutator of the flows for evolution equations
$u_{t_{1}}=f$ and $u_{t_{2}}=g$. Therefore, given a Liouville type equation,
the set of all symmetries of the form
\begin{equation} \label{symmm}
u_{\tau}=S(x,u,u_1, u_2, \dots)
\end{equation}
is a Lie subalgebra in ${\cal F}$ (as usual, we identify symmetries and their
right-hand sides).

The operator  ${\cal L}=D^3+2 u D+u_1$, corresponding to the Liouville
equation (see Example \ref{ex63}), possesses the following remarkable property.
Its image is a Lie subalgebra in ${\cal F}$. Namely, it is easy to verify
that for any $f,g \in {\cal F}$
$$
\Big[{\cal L}(f), \, {\cal L}(g)\Big]={\cal L} \Big(D(f)\, g-D(g)\, f
+g_* {\cal L}(f)-f_* {\cal L}(g) \Big).
$$
Since ${\cal L}$ is injective, the last formula defines a new Lie bracket
$$
\Big[f, \, g\Big]_1 \stackrel{def}{=} D(f)\, g-D(g)\, f
+g_* {\cal L}(f)-f_* {\cal L}(g)
$$
on ${\cal F}.$

A differential operator ${\cal L}$ is called {\it integrable} \cite{sylv} if ${\rm Im}\, {\cal L}$ is a Lie subalgebra in ${\cal F}$ .
\begin{remark} It can be shown that for any Hamiltonian operator ${\cal H}:
{\cal F} \rightarrow {\cal F}$ {\rm (see (\cite{Olv93})\,)} its image is a
subalgebra in ${\cal F}$. Therefore, integrable operators can be regarded as
a non-skewsymmetric generalization of the Hamiltonian operators.
\end{remark}

For all known Liouville type equations, the operator
${\cal L}$ given by \eqref{LLLop} is integrable. We show that, in
connection with the above results, it looks rather natural.

Let us denote by ${\cal F}_w$ a Lie algebra of functions depending
on variables $x, w, w_1, \dots$ with respect to bracket \eqref{Lie}.
It follows from Theorem \ref{th63} that all coefficients of
${\cal L}$ are functions of $x, w, w_1, \dots$, where $w$ is the minimal
$Y$-integral. Thus we have a differential operator
${\cal L}: {\cal G}_w \rightarrow {\cal G}_w$. Let us show that under
some assumptions  ${\rm Im}\,{\cal L}$ is a Lie subalgebra in ${\cal G}_w$.

Indeed, consider the set ${\cal S}$ of all symmetries of the form
\eqref{msym} of a given Liouville type equation. For all known examples,
${\cal S}$ coincides with the set of all symmetries of the form \eqref{symmm}.
In this case, ${\cal S}\in {\cal F}$ is a Lie subalgebra with respect
to bracket \eqref{Lie}. It was shown in Proposition \ref{prop63}
that the image of the symmetry \eqref{msym} under substitution
\eqref{miu} is given by $w_{\tau}=Q$, where
\begin{equation} \label{symmQ}
Q= w_* \, {\cal M} \,\Big( G(x,w,w_1,\dots)\Big).
\end{equation}
Hence the set of all functions of the form \eqref{symmQ} is expected to be
a subalgebra with respect to \eqref{Lie}. To conclude the speculation,
we recall that, according to Remark \ref{rem63}, $w_* \, {\cal M}={\cal L}$ for all
known equations of Liouville type.

\begin{example} Consider equation \cite{zibsok2}
\begin{equation}\label{Eq1}
  u_{xy}=\frac 1{6u+y}B^2(B-1)\bar B (\bar B -1)^2+\frac 1{6u+x}\bar B ^2(\bar B
-1)B(B-1)^2,
\end{equation}
where $B=B(u_x)$ and $\bar B =\bar B (u_y)$ are solutions of cubic equations
\begin{equation}\label{AB}
  \frac 13B^3-\frac 12B^2=u_x, \qquad \frac 13\bar B ^3-\frac 12\bar B ^2=u_y.
\end{equation}
For this equation we have  $H_3=H_{-2}=0$.
The minimal $y$ and $x$-integrals $w$ and $\bar w$ are given by
$$
\begin{array}{c}
\displaystyle
w=D\left\{\ln\left[u_2-\frac{B^4(B-1)^2}{6u+y}-
\frac{B^2(B-1)^4}{6u+x}\right]-\ln B(B-1)\right\}-\\[2ex]
\displaystyle -\left[\left(\frac 1{6u+y}+\frac 1{6u+x}\right)B-\frac
1{6u+x}\right]B(B-1),
\end{array}
$$
\medskip
$$
\begin{array}{c}
\displaystyle \bar w=\bar D\left\{\ln\left[\bar u_2-\frac{\bar B^4(\bar
B-1)^2} {6u+x}-\frac{\bar B^2(\bar B-1)^4}{6u+y}\right]-\ln\bar B(\bar
B-1)\right\}-\\[2ex]
\displaystyle -\left[\left(\frac 1{6u+y}+\frac
1{6u+x}\right)B-\frac 1{6u+x}\right]B(B-1).
\end{array}
$$
The functions $\psi$ and $\bar \psi$ (see Proposition \ref{prop62}) are defined by
$$
 \psi=u_2-\displaystyle \frac{B^4(B-1)^2}{6u+y}-\frac{B^2(B-1)^4}{6u+x},
  \qquad 
\bar \psi=\bar u_2-\displaystyle
  \frac{\bar B ^4(\bar B -1)^2}{6u+x}-\frac{\bar B ^2(\bar B -1)^4}{6u+y}.
$$
The Laplace invariants for equation \eqref{Eq1} have the form
$$
\begin{array}{c}
\displaystyle
H_1=\frac{\bar B (\bar B -1)}{B(B-1)}\left[\frac{\bar B -1}{(6u+y)(B-1)^2}-\frac
{\bar B} {(6u+x)B^2}\right]\, \psi,\\[6mm]
\displaystyle
H_{0}=\frac{B(B-1)}{\bar B (\bar B -1)}\left[\frac{B-1}{(6u+x)(\bar B -1)^2}-\frac
{B}{(6u+y)\bar B ^2}\right]\, \bar \psi,\\[6mm]
\displaystyle
H_2=\frac{2(6u+y)(6u+x)}{\bar B (\bar B -1)\,[(6u+x)B^2(\bar B -1)-(6u+y)\bar B
(B-1)^2]^2}\, \psi \,
\bar \psi,
\\[6mm]
\displaystyle
H_{-1}=\frac{2(6u+y)(6u+x)}{B(B-1)\, [(6u+y)\bar B ^2(B-1)-(6u+x)B(\bar B
-1)^2]^2}\,
\psi \, \bar \psi.
\end{array}
$$
The corresponding operator ${\cal L}$   can be factorized:
\begin{equation} \label{Eq1LL}
{\cal L}=D(D+ w)(D+ w)(D+2 w)(D+3 w). 
\end{equation}
One can verify that the operator  $\,{\cal L}$ is integrable. 
\end{example}

Let ${\cal L}$ be a differential operator with coefficients from ${\cal F}_w$.
It is not very difficult to check whether its image is a Lie subalgebra. This
condition turns out to be rather rigid.

For example, let us consider operators "similar" to \eqref{Eq1LL}. Namely, the
coefficients of \eqref{Eq1LL} are polynomials in  $w, w_1, w_2, \dots$.
Operator \eqref{Eq1LL} is homoge\-neous if we assign a weight 1 to $D$ and
$i+1$ to $w_i$. We present here \cite{zibsok} a complete list of integrable operators ${\cal L}$ of
orders 2-6 with the above homogeneity property. It turns out that all of them
can be decomposed into factors of first order.
\begin{itemize}
\item Operator of second order:
$$
{\cal L}_1^{(2)}=D\,(D+ w);
$$
\item Operator of third order:
$$
{\cal L}_1^{(3)}=D\,(D+ w)\,(D+ w);
$$
\item Operators of fourth order:
$$
\begin{array}{l}
{\cal L}_1^{(4)}=D\,(D+ w)\,(D+ w)\,(D+w),\\[2mm]
{\cal L}_2^{(4)}=D\,(D+ w)\,(D+ w)\,(D+2 w);
\end{array}
$$
\item Operators of fifth order:
$$
\begin{array}{l}
{\cal L}_1^{(5)}=D\,(D+ w)\,(D+ w)\,(D+w)\,(D+w),\\[2mm]
{\cal L}_2^{(5)}=D\,(D+ w)\,(D+ w)\,(D+2 w)\,(D+3 w);
\end{array}
$$
\item Operators of sixth order:
$$
\begin{array}{l}
{\cal L}_1^{(6)}=D\,(D+ w)\,(D+ w)\,(D+w)\,(D+w)\,(D+w),\\[2mm]
{\cal L}_2^{(6)}=D\,(D+ w)\,(D+ w)\,(D+w)\,(D+w)\,(D+2 w),\\[2mm]
{\cal L}_3^{(6)}=D\,(D+ w)\,(D+ w)\,(D+2 w)\,(D+3 w)\,(D+3 w),\\[2mm]
{\cal L}_4^{(6)}=D\,(D+ w)\,(D+ w)\,(D+2 w)\,(D+3 w)\,(D+4 w).
\end{array}
$$
\end{itemize}
Note that the operator ${\cal L}_2^{(5)}$ coincides with \eqref{Eq1LL}.

In the paper \cite{SanWan} these examples were prolonged by infinite sequences of integrable operators of arbitrary order. 
In \cite{kis} some examples of integrable operators with matrix coefficients were constructed.

\chapter{Integrable non-abelian equations}

\section{ODEs on free associative algebra}

 We consider ODE systems of the form
\begin{equation}\label{geneq}
\frac{d x_{\alpha}}{d t}=F_{\alpha}({\bf x}), \qquad {\bf x}=(x_1,...,x_N),
\end{equation}
where $x_i(t)$ are $m\times m$ matrices, $F_{\alpha}$ are (non-commutative) polynomials with constant scalar coefficients. As usual, a symmetry is defined as an equation
\begin{equation}\label{gensymmat}
\frac{d x_{\alpha}}{d \tau}=G_{\alpha}({\bf x}), 
\end{equation}
compatible with \eqref{geneq}. 

\subsubsection{Manakov top}

Let $N=2$, $\,x_1=u,\, x_2=v$. The following  system 
\begin{equation}\label{man}
u_{t}=u^2 \, v-v \, u^2, \qquad v_{t}=0
\end{equation}
has infinitely many symmetries for arbitrary size of matrices $u$ and $v$.

Many important {\it multi--component} integrable systems
  can be obtained as
reductions of \eqref{man}. 

For instance, if $u$ is $m\times m$ matrix such that $u^t=-u$, and $v$ is a constant diagonal matrix, then 
\eqref{man} is equivalent to the $m$-dimensional Euler top. The
integrability of this model by the inverse scattering method was established by S.V. Manakov
in \cite{man}. 

Consider the cyclic reduction
$$
u=\left( \begin{array}{cccccc}
0&u_1&0&0&\cdot&0\\
0&0&u_2&0&\cdot&0\\
\cdot&\cdot&\cdot&\cdot&\cdot&\cdot\\
0&0&0&0&\cdot&u_{m-1}\\
u_{m}&0&0&0&\cdot&0
\end{array}\right)\, , \qquad 
v=\left( \begin{array}{cccccc}
0&0&0&\cdot&0&J_m\\
J_1&0&0&\cdot&0&0\\
0&J_2&0&\cdot&0&0\\
\cdot&\cdot&\cdot&\cdot&\cdot&\cdot\\
0&0&0&\cdot&J_{m-1}&0
\end{array}\right)\, ,
$$
where $u_k$ and $J_k$ are matrices of lower size.
Then \eqref{man} is equivalent to   the non-abelian Volterra chain
$$
\frac{d}{dt}u_k=u_k  u_{k+1}  J_{k+1}-J_{k-1}  u_{k-1}  u_k, \qquad k\in \Z_k.
$$
If we assume $m=3,\, J_1=J_2=J_3= \,{\rm Id}$ and
$u_3=-u_1 -u_2$ then the latter system yields
$$
u_t=u^2+u  v+v  u\, ,\qquad v_t=-v^2-u  v-v  u\, .
$$

It turns out that for any $i,j\in \N$ system \eqref{man} has first integrals of the form ${\rm tr}\,P_{i,j},$ where $P_{i,j}$ is a non-commutative homogeneous polynomial of degree $i$ in $v$ and degree $j$ in $u$. For example, 
$$
P_{2,2}=2 v^2 u^2 + v u v u, \qquad P_{3,2}=v^3 u^2+v^2 u v u, \qquad
P_{2,3}=v^2  u^3+v  u  v u^2. 
$$
It follows from the property ${\rm tr}\,(u v - v u)=0$ that the polynomials $P_{ij}$ are defined up to any cyclic permutations in factors of their monomials. 

There exists the following integrable generalization \cite{odrubsok} of system \eqref{man} to the case of $N$ matrices for arbitrary $N:$  
$$
\frac{d x_\alpha}{d t}=\sum_{\beta\ne \alpha} \frac{x_\alpha x_\beta^2-x_\beta^2 x_\alpha}{(\lambda_\alpha-\lambda_\beta) c_\beta}+
\sum_{\beta\ne \alpha} \frac{x_\beta x_\alpha^2-x_\alpha^2 x_\beta}{(\lambda_\alpha-\lambda_\beta) c_\alpha}.
$$
Here,
$$
\sum_1^N x_\alpha=C,
$$
where $C$ is a constant matrix.

\subsubsection{Non--abelian systems}

To check that \eqref{gensymmat} is a symmetry of \eqref{geneq} for arbitrary size of matrices $x_i,$ one only need the  associativity of the matrix product. Therefore, the compatibility of \eqref{geneq} and \eqref{gensymmat} is valid even if $x_1,\dots,x_N$ are generators of the free associative algebra ${\cal A}.$ We call systems on ${\cal A}$  {\it non-abelian systems}. 

However, in the non-abelian case the following questions arise
\begin{itemize}
\item What do formulas like  \eqref{geneq} and \eqref{gensymmat} mean?
\item How to define the functional ${\rm}\, tr$ for first integrals?
\end{itemize}

\begin{definition}
A linear map $d: {\cal A}\to {\cal A}$ is called derivation if it satisfies the 
Leibniz rule: $d (x y)=x d(y)+d(x) y.$
\end{definition}
If we fix $d(x_i)=F_i({\bf x}),$ where $x_i$ are generators of ${\cal A}$, then $d(z)$ is uniquely defined for any element $z\in {\cal A}$ by the Leibniz rule. It is clear that the polynomials $F_i$ can be taken arbitrarily. 

Instead of the dynamical system \eqref{geneq} we consider the derivation $D_t: {\cal A}\to {\cal A}$ such that $D_t(x_i)=F_i.$ The compatibility of \eqref{geneq} and \eqref{gensymmat} means that the corresponding derivations $D_t$ and $D_{\tau}$ commute: $D_t D_{\tau}-D_{\tau} D_t=0.$

To introduce the concept of a first integral, we need an analog of the
trace, which is not defined for the algebra ${\cal A}$. As a
matter of fact, in our calculations we use only two properties of
the trace, namely its linearity and the possibility to perform cyclic
permutations in monomials. Let us define an equivalence relation
for elements of ${\cal A}$ in a standard way.

\begin{definition}\label{def311} We say that two elements $f_1$ and $f_2$ of ${\cal A}$ are equivalent, and we denote it by $f_1 \sim f_2$,   iff $f_1$ can be obtained from $f_2$ by
cyclic permutations of factors in its monomials. We denote by ${\rm tr}\,f$ the equivalence class of an element $f$. 
\end{definition}

\begin{definition}\label{def33} An element ${\rm tr}\,h,$ where $h\in {\cal A},$ is said to be a first integral of the system \eqref{geneq}
if $D_t(h) \sim 0$.
\end{definition}

There is an obvious similarity between Definition \ref{def33} and the definition for conserved
densities in the theory of evolutionary PDEs (see Remark \ref{equi}). In both cases,
first integrals and conserved densities are defined as elements of
equivalence classes, the
difference is in the choice of the equivalence relation.

\begin{remark} \label{quotA} Any commutator is equivalent to zero and vise versa if $f_1 \sim f_2,$ then there exist elements $a_i,\, b_i$ such that 
$f_1-f_2=\sum_i [a_i, b_i]$. Therefore, we may regard ${\rm tr}\,f$ as an element of the quotient vector space ${\cal T}={\cal A}/[{\cal A},\,{\cal A}].$  It is clear that the derivation $D_t$ is well defined on ${\cal T}$.
\end{remark}

From the view-point of the symmetry approach, system \eqref{geneq} is integrable if it possesses infinitely many linearly independent symmetries. Note that if $x_i$ are matrices of a fixed size, then the existence of an infinite set of symmetries is impossible.

Another criterion of integrability is the existence of an infinite set of first integrals. 

\begin{example} \cite{miksokcmp} The following non-abelian system 
$$
u_{t}=v^2, \qquad v_{t}=u^2
$$
possesses infinitely many symmetries and first integrals.
F. Calogero has observed that  in the matrix case the functions $x_i=\lambda_i^{1/2}$, where the $\lambda_i$ are eigenvalues of the matrix  $u-v,$ satisfy  the following integrable system:
$$
x_i^{''}=-x_i^5+\sum_{j\ne i} \Big[(x_i-x_j)^{-3}+(x_i+x_j)^{-3}\Big].
$$ 
\end{example}

\subsection{Quadratic non-abelian systems}

In the case of arbitrary $N$ a homogeneous quadratic non-abelian system has the form 
\begin{equation}\label{quadN}
(u^i)_t = \sum_{j,k} C^i_{jk}\, u^j \,u^k, 
\end{equation}
where $u^i, \, i=1,\dots, N$ are non-commutative variables. We denote the set of coefficients $C^i_{jk},\, \, i,j,k=1,\dots,N$
by ${\bf C}$.

\begin{remark}\label{rem32} One can associate a $N$-dimensional algebra with
structural constants $ C^i_{jk}$ with any system \eqref{quadN}.
\end{remark}

The class of systems \eqref{quadN} is invariant with respect to the group ${\rm GL}_N$
of linear transformations
$$
\hat u^i = \sum_{j} s^i_{j}\, u^j.
$$
 
\begin{definition} A function $F(\bf C)$ is called a ${\rm GL}_N$
semi-invariant of degree $k$ if
$$
F(\hat {\bf C})=({\rm det}\,S)^k \, F({\bf C}),
$$ where $S$ is the matrix with entries $s^i_j.$ Semi-invariants of
degree $0$ are called {\it invariants}.
\end{definition}

\begin{definition} A row-vector ${\bf V}(\bf C)$ is called a vector ${\rm GL}_N$
semi-invariant of degree $k$ if
$$
{\bf V}(\hat {\bf C})=({\rm det}\,S)^k \, {\bf V}({\bf C})\,S.
$$ 
Vector semi-invariants of
degree $0$ are called {\it vector invariants}.
\end{definition}

\begin{lemma}\label{lem31} i) The product $I\, {\bf V}$ of a semi-invariant $I$ and a vector semi-invariant {\bf V} of degrees $k_1$ and $k_2$ is a vector semi-invariant of degree $k_1+k_2$.

ii) Let ${\bf V}_i,\, i=1,\dots,N$ be vector semi-invarians of degrees $k_i.$ Then the determinant of the matrix constituted of the vectors ${\bf V}_i$ is a (scalar) semi-invariant of degree $1+\sum k_i.$

\end{lemma}

\begin{op} Construct series of simple algebras with the structural constants $C^i_{jk}$ such that equation \eqref{quadN} has a
symmetry of the form 
$$
(u_i)_{\tau}= \sum_{j,k,m} B^i_{jkm}\, u^j \,u^k \, u^{m}.
$$
\end{op} 

\subsection{Two-component non-abelian systems}

Consider non-abelian systems 
\begin{equation}\label{uv}
u_t=P(u,v)\, ,\qquad v_t=Q(u,v)\, ,\qquad P,Q\in {\cal A}
\end{equation}
on the free associative algebra ${\cal A}$ with generators $u$ and $v$.   

Define the involution
$\star$ on ${\cal A}$ by the formulas
\begin{equation}\label{star1}
u^\star=u\, ,\quad v^\star=v\, ,\quad (a\, b)^\star=
b^\star \, a^\star \, ,\quad a,b\in {\cal A}.
\end{equation}

Two systems related to each other by a linear
transformation of the form
\begin{equation}\label{invert}
\hat{u}=\alpha u+\beta v\, ,\qquad \hat{v}=\gamma u+\delta v\, ,\qquad
\alpha \delta-\beta\gamma \ne 0
\end{equation}
and by the involution \eqref{star1} are called {\it equivalent}.

\begin{definition} A system, which is equivalent to a system of the form
$$
u_t=P(u,v)\, ,\qquad v_t=Q(v)\,
$$
is called {\it triangular}. If $Q(v)=0,$ then the system is called {\it 
strongly triangular}.
\end{definition}

\begin{remark} From the algebraic point of view the triangularity is
equivalent to the fact that the corresponding algebra (see Remark \ref{rem32}) has one-dimensional double-side ideal.
\end{remark}

\begin{op}
Describe all non-equivalent non-triangular systems \eqref{uv} which possess infinitely many symmetries and/or first integrals.  
\end{op}
The simplest class of such systems are quadratic systems of the form
\begin{equation}
\begin{cases}
u_t   =  \alpha_1 u \, u  + \alpha_2 u \, v + \alpha_3  v \, u +
\alpha_4 v \, v,\\[2mm]
v_t   =  \beta_1 v  \, v  + \beta_2 v  \, u + \beta_3  u \, v + \beta_4
u \, u. \label{equ}
\end{cases}
\end{equation}

 \subsubsection{Equivalence}  
 
The group ${\rm GL}_2$ of transformations \eqref{invert} acts on polynomials of the eight coefficients $\alpha_i, \beta_i$ of the system \eqref{equ}. The corresponding infinitesimal action of 
 the Lie algebra $gl_2$ is defined by the following vector fields:
\begin{eqnarray*}
 X_{11}&=& \alpha_1 \frac{\partial}{\partial \alpha_1}- \alpha_4 \frac{\partial}{\partial \alpha_4}+\beta_2 \frac{\partial}{\partial \beta_2}
 +\beta_3 \frac{\partial}{\partial \beta_3}-2 \beta_4 \frac{\partial}{\partial \beta_4}, \\[3mm]
 X_{22}&=& \beta_1 \frac{\partial}{\partial \beta_1}- \beta_4 \frac{\partial}{\partial \beta_4}+\alpha_2 \frac{\partial}{\partial \alpha_2}
 +\alpha_3 \frac{\partial}{\partial \alpha_3}-2 \alpha_4 \frac{\partial}{\partial \alpha_4}, \\[3mm]
 X_{12}&=& -\beta_4 \frac{\partial}{\partial \alpha_1}+(\alpha_1-\beta_3) \frac{\partial}{\partial \alpha_2}+(\alpha_1-\beta_2) \frac{\partial}{\partial \alpha_3} \nonumber \\
 & &+(\alpha_2+\alpha_3-\beta_1) \frac{\partial}{\partial \alpha_4}+ (\beta_2+\beta_3) \frac{\partial}{\partial \beta_1}+\beta_4 \frac{\partial}{\partial \beta_2}+\beta_4 \frac{\partial}{\partial \beta_3}, \\[3mm]
 X_{21}&=& -\alpha_4 \frac{\partial}{\partial \beta_1}+(\beta_1-\alpha_3) \frac{\partial}{\partial \beta_2}+(\beta_1-\alpha_2) \frac{\partial}{\partial \beta_3} +(\beta_2+\beta_3-\alpha_1) \frac{\partial}{\partial \beta_4} \nonumber \\
 & &+ (\alpha_2+\alpha_3) \frac{\partial}{\partial \alpha_1}+\alpha_4 \frac{\partial}{\partial \alpha_2}+\alpha_4 \frac{\partial}{\partial \alpha_3}.
\end{eqnarray*}
The  eight-dimensional vector space spanned by the coefficients  $\alpha_i, \beta_i$ can be decomposed into a direct sum of three subspaces invariant
with respect to this action.  The first one is four-dimensional and spanned by 
\[
x_1=\alpha_4, \qquad x_2 = \alpha_2+\alpha_3-\beta_1, \qquad x_3=\beta_2+\beta_3-\alpha_1, \qquad x_4=\beta_4;  
\]
another two are two-dimensional and spanned by
\[
x_5=\alpha_1+\beta_3, \qquad x_6=\beta_1+\alpha_3;
\]
and
\[
x_7=\alpha_1+\beta_2  \qquad  x_8=\beta_1+\alpha_2.
\]
Denote these vector spaces by $W_1$, $W_2$ and $W_3$, respectively.

\begin{lemma} The vectors ${\bf V}_1=(x_5,x_6)$ and ${\bf V}_2=(x_7,x_8)$ are vector invariants.
\end{lemma}

The (scalar) semi-invariants of the transformation group \eqref{invert} have rather
simple form in the variables $x_i.$ 
There exists only  one semi-invariant of degree 1:
\[
I_0 = x_6 x_7 - x_5 x_8.
\]
In the original variables  
$$
I_0=\alpha_1 \alpha_2 - \alpha_1 \alpha_3 - \alpha_3 \beta_2 - \beta_1 \beta_2 + \alpha_2 \beta_3 + \beta_1 \beta_3.
$$

 The following polynomials are semi-invariants of degree 2:
\[\begin{array}{c}
I_1 = x_2^2 x_3^2 + 4\, x_1 x_3^3 + 4\, x_2^3 x_4 + 18\, x_1 x_2 x_3 x_4 - 27\, x_1^2 x_4^2,\\ [3mm]
 I_2 = (x_5 - x_7)^3 \,x_1 -(x_5 - x_7)^2 (x_6 - x_8) \,x_2 -(x_5 - x_7) (x_6 - x_8)^2 \, x_3 + (x_6 - x_8)^3\, x_4,\\[3mm]
 I_3 =  (x_5^3 - x_7^3) \,x_1 +(x_7^2 x_8 - x_5^2 x_6) \,x_2 + (x_7 x_8^2 - x_5 x_6^2 ) \, x_3 + (x_6^3 - x_8^3)\, x_4,\\[3mm]
 I_4 = (x_5^2 - x_7^2)\, x_2^2 + (x_6^2 - x_8^2)\,x_3^2 + 3 \,(x_5^2 - x_7^2)\,x_1 x_3 + 3\, (x_6^2 - x_8^2)\,x_2 x_4+9\, ( x_7 x_8 - x_5 x_6 )\, x_1 x_4.
 \end{array}
\]

\begin{remark} The semi-invariant $I_1$ is symmetric with respect to the
involution \eqref{star1} and 
the other four semi-invariants are skew-symmetric.
\end{remark}
\begin{proposition} Any invariant $J$ of the transformation group
\eqref{invert} is a function of the four functionally independent invariants
$\ds J_i=\frac{I_i}{I_0^2}, \,\, i=1,2,3,4.$
\end{proposition}
\begin{proof} Invariants are solutions of the PDE system $X_{ij}(J)=0, \,\,
i,j=1,2$. 
It is easy to verify that the vector fields $X_{ij}$ 
are linearly
independent over functions. Therefore, there exist only four functionally
independent invariants. A straightforward computation of the rank for the
Jacobi matrix shows that $J_i,\, i=1,2,3,4$ are functionally
independent. 
\end{proof}

\begin{remark} Actually there exist 9 linearly independent over $\C$ semi-invariants of degree 2. But only five of them are functionally independent. Moreover there exist 7 linearly independent over $\C$ vector semi-invariants of degree 1. Each semi-invariant of degree 2 can be obtained by the construction of Lemma \ref{lem31} applying to vector semi-invariants of degree 1 and vector invariants.
\end{remark}

\begin{example} The vectors
\begin{eqnarray*}
{\bf V}_3&=&\Big( -\beta_4 (\alpha_3 - \alpha_2)^2 + (\beta_3 - \beta_2) (\alpha_1 \alpha_2 - \alpha_1 \alpha_3 - \alpha_2 \beta_2 + \alpha_3 \beta_3), \\ && \qquad \alpha_4 (\beta_3 - \beta_2)^2 - (\alpha_3 - \alpha_2) (\beta_1 \beta_2 - \beta_1 \beta_3 - \beta_2 \alpha_2 + \beta_3 \alpha_3) \Big)
\end{eqnarray*}
and
\begin{eqnarray*}
{\bf V}_4&=&\Big((\beta_2 - \beta_3)  \alpha_4 \beta_4 - (\alpha_2 - \alpha_3) (\alpha_2 - \beta_1) \beta_4 + 
 \beta_3 (\alpha_1 \alpha_2 - \alpha_1 \alpha_3 - \alpha_2 \beta_2 + \alpha_3 \beta_3), \\ && \,\, - (\alpha_2 - \alpha_3)  \alpha_4 \beta_4 + (\beta_2 - \beta_3) (\beta_2 - \alpha_1) \alpha_4 - 
 \alpha_3 (\beta_1 \beta_2 - \beta_1 \beta_3 - \beta_2 \alpha_2 + \beta_3 \alpha_3) \Big)
\end{eqnarray*}
are vector semi-invariants of degree 1.
\end{example}

\begin{proposition}  System \eqref{equ} is triangular iff 
$
{\bf V}_3={\bf V}_4 =0 .
$
\end{proposition}

\begin{proposition}  System \eqref{equ} is strongly triangular iff the vectors $(\alpha_1,\alpha_2,\alpha_3,\alpha_4)$ and $ (\beta_4, \beta_3, \beta_2, \beta_1)$ are linearly dependent.  
\end{proposition}

\subsubsection{Examples}

 Some experiments with non-triangular systems \eqref{equ} having
symmetries of degree $m$ have been made in \cite{miksokcmp}.
  One of the results is:
\begin{theorem}
Any non-triangular equation \eqref{equ} possessing
 a symmetry of the form
 $$
\begin{cases} \nonumber
u_{\tau}  =\gamma _1 u\, u\, u  + \gamma _2 u\, u\, v +
\gamma _3 u\, v\, u + \gamma _4 v\, u\, u \, +  \\ \qquad \,
\gamma _5 u\, v\, v  + \gamma _6 v\, u\, v +
\gamma _7 v\, v\, u + \gamma _8 v\, v\, v ,\\[3mm]
v_{\tau}  = \delta _1 u\, u\, u  + \delta _2 u\, u\, v + \delta _3
u\, v\, u + \delta _4 v\, u\, u \, + \\ \qquad \,
\delta _5 u\, v\, v  + \delta _6 v\, u\, v +
\delta _7 v\, v\, u + \delta _8 v\, v\, v 
\end{cases}
$$
is equivalent to one of the following:
\begin{eqnarray*}
&&a): \begin{cases}
u_t=u\, u - u\, v,\\
v_t=v\, v - u\, v + v\, u,     
\end{cases} \qquad
b): \begin{cases}
u_t=u\, v,\\
v_t=v\, u,        
\end{cases} \qquad 
c):\begin{cases}
u_t=u\, u - u\, v,\\
v_t=v\, v - u\, v,        
\end{cases}
\\[3mm]&&
d):\begin{cases}
u_t= - u\, v, \\
v_t=v\, v + u\, v - v\, u,     
\end{cases} \qquad
e):\begin{cases}
u_t=u\, v - v\, u,\\
v_t=u\, u + u\, v - v\, u,     
\end{cases} \qquad 
f): \begin{cases}
u_t=v\, v,\\
v_t=u\, u.        
\end{cases}
\end{eqnarray*}
\end{theorem}
\begin{remark}
Possibly equation a) has only one symmetry while equations b)-f) have infinitely many. It is proved for equations b), c), and f).
\end{remark}

It is a remarkable fact that a requirement of the existence of just
one cubic symmetry selects a finite list of equations with no free
parameters (or more precisely, all possible parameters can be
removed by linear transformations \eqref{invert}).

A list of six equations with quartic symmetries was presented in \cite{miksokcmp}. As mentioned by A. Odesskii, two of them are equivalent. 
The five non--equivalent equations are given by 
\begin{eqnarray*}
&& \begin{cases}
u_t= - u\, v,\\
v_t=v\, v + u\, v,     
\end{cases} \qquad
 \begin{cases}
u_t=-v\, u,\\
v_t=v\, v + u\,v,        
\end{cases} \qquad 
\begin{cases}
u_t=u\, u - u\, v - 2 v\,u,\\
v_t=v\, v -v \, u - 2 u\, v,        
\end{cases}
\\[3mm]&&
\begin{cases}
u_t=u \, u - 2 v\, u, \\
v_t=v\, v  - 2 v\, u,     
\end{cases} \qquad
\begin{cases}
u_t=u\, u - 2 u\, v,\\
v_t=v\, v + 4  v\, u.     
\end{cases}  
\end{eqnarray*}
Using computer algebra system CRACK \cite{wolf}, T.Wolf verified that it is a complete list of non--triangular systems that have quartic but no cubic symmetries.  

\begin{op}
Describe non-equivalent non-triangular systems \eqref{equ} that possess infinitely many symmetries and/or first integrals.  
\end{op}

\section{PDEs on free associative algebra}

  In this section we consider the so called non-abelian evolution equations,
which are natural generalizations of the evolution matrix equations.  

\subsection{Matrix integrable equations}
The matrix Burgers equation is given by
\begin{equation}\label{matbur}
{\bf U}_{t}={\bf U}_{xx}+2\, {\bf U} {\bf U}_{x},
\end{equation}
where ${\bf U}(x,t)$ is unknown $m\times m$-matrix.
The matrix KdV equation has the following form
\begin{equation}\label{matkdv}
{\bf U}_{t}={\bf U}_{xxx}+3\, ({\bf U} {\bf U}_{x}+{\bf U}_{x} {\bf U}). 
\end{equation}
 It is known that  this equation has infinitely many 
higher symmetries for arbitrary $m$. All of them can also be written in the
matrix form.
The simplest higher symmetry of \eqref{matkdv} is given by
\[
{\bf U}_{\tau}={\bf U}_{xxxxx}+5\, ({\bf U} {\bf U}_{xxx}+{\bf U}_{xxx} {\bf U})+10\,({\bf U}_{x} {\bf U}_{xx}+{\bf U}_{xx} {\bf U}_{x})+
+10\,({\bf U}^{2} {\bf U}_{x}+{\bf U} {\bf U}_{x} {\bf U}+{\bf U}_{x}{\bf U}^{2}).
\]
For $m=1$ this matrix hierarchy coincides with the usual KdV hierarchy.

In general, we may consider \cite{OlvSok98} matrix equations of the form
$$
{\bf U}_{t}=F({\bf U},\, {\bf U}_{1},\, \dots, \,{\bf U}_{n}), \qquad {\bf U}_i=\frac{\partial^i {\bf U}}{\partial x^i},
$$
where $F$ is a (non-commutative) polynomial. As usual in this chapter, the criterion of integrability is the existence of matrix  higher symmetries 
$$
{\bf U}_{\tau}=G({\bf U},\, {\bf U}_{1},\, \dots, \,{\bf U}_{m}).
$$
Such equations can be regarded as matrix generalizations of scalar integrable equations \eqref{eveq}.

Below we present a list of integrable matrix equations. The main goal is to demonstrate that the matrix KdV equation is not a single accident.  Many of known integrable models have matrix generalizations \cite{march}--\cite{kuper}.

In particular, the mKdV equation \eqref{mkdv} has two different matrix generalizations:
\begin{equation}\label{matmkdv}
{\bf U}_t={\bf U}_{xxx}+3 {\bf U}^2 {\bf U}_x+3 {\bf U}_x {\bf U}^2,
\end{equation}
and (see \cite{kuper})
$${\bf U}_t={\bf U}_{xxx}+3 [{\bf U}, {\bf U}_{xx}]-6 {\bf U} {\bf U}_x {\bf U}.$$
The matrix generalization of the NLS equation \eqref{NLS} is given by
\begin{equation}\label{matnls}
{\bf U}_t={\bf U}_{xx}-2\, {\bf U} {\bf V} {\bf U}, \qquad {\bf V}_t=-{\bf V}_{xx}+2\, {\bf V} {\bf U} {\bf V}.
\end{equation}

The following integrable matrix system of the derivative NLS type 
\begin{equation}\label{olsokder}
{\bf U}_t={\bf U}_{xx}+2\, {\bf U} {\bf V} {\bf U}_x, \qquad {\bf V}_t=-{\bf V}_{xx}+2\, {\bf V}_x {\bf U} {\bf V}
\end{equation}
was found in \cite{OlvSok99}.

The Krichever-Novikov equation \eqref{KN} with $P=0$ is called the {\it Shwartz KdV equation}. Its matrix generalization is given by
\begin{equation}\label{matskdv}
{\bf U}_t={\bf U}_{xxx}-\frac32\, {\bf U}_{xx} {\bf U}_x^{-1} {\bf U}_{xx}.
\end{equation}
The Krichever-Novikov equation \eqref{KN} with the generic $P$  has probably no matrix generalizations.

The matrix Heisenberg equation has the form
\begin{equation}\label{matheis}
{\bf U}_t = {\bf U}_{xx} -2\, {\bf U}_x ({\bf U}+{\bf V})^{-1} {\bf U}_{x},
\qquad {\bf V}_t = -{\bf V}_{xx} +2\, {\bf V}_x ({\bf U}+{\bf V})^{-1} {\bf V}_{x}. 
\end{equation}

One of the most renowned hyperbolic matrix integrable equations is the principle chiral $\sigma$-model
\begin{equation}\label{prchir}
{\bf U}_{xy}=\frac12\, ({\bf U}_x {\bf U}^{-1} {\bf U}_y+ {\bf U}_y {\bf U}^{-1} {\bf U}_x).
\end{equation}

The following, maybe, new integrable system 
$$
\begin{array}{c}
\begin{cases}
{\bf U}_t = \lambda_1 {\bf U}_{x} + (\lambda_2-\lambda_3)\, {\bf W}^t  {\bf V}^t , \\[2mm]
{\bf V}_t = \lambda_2 {\bf V}_{x} + (\lambda_3-\lambda_1)\, {\bf U}^t  {\bf W}^t , \\[2mm]
{\bf W}_t =\lambda_3 {\bf W}_{x} + (\lambda_1-\lambda_2)\, {\bf V}^t  {\bf U}^t
\end{cases}
\end{array}
$$
is a matrix generalization of the 3-wave model. In contrast with the previous equation it contains the matrix transposition denoted by $^t$.

Let ${\bf e}_1,\dots, {\bf e}_N$ be a basis of some associative algebra ${\cal B}$ and $U$ be the element \eqref{UU}.
Then, any of the matrix equations presented above gives rise to an integrable system for unknown functions $u_1,
\dots, u_N$ associated with ${\cal B}$. Indeed, only the associativity of the product in $\cal B$ is needed to verify that symmetries of a matrix equation remain symmetries of the corresponding system for $u_i$. 

Interesting examples of integrable multi-component systems  are produced by the
Clifford algebras and by the group algebras of associative rings. 

However, the most fundamental so called non-abelian equations correspond to free associative algebras (cf. Section 3.1).

\subsection{Non-abelian evolution equations}\label{naev}

In order to formalize the concept of matrix equations, let us consider
evolution equations on a free associative algebra ${\cal A}$. In the case of
one-field non--abelian equations the generators of ${\cal A}$ are
\begin{equation}\label{generat}
U, \quad U_{1}=U_{x}, \quad \dots, \quad U_{k}=\frac{\partial^k U}{\partial x^k}, \quad \dots .
\end{equation}
Since ${\cal A}$ is assumed to be a free algebra, no algebraic relations between
the generators \eqref{generat} are allowed.  
 All definitions can be easily generalized to the case of several non-abelian variables.

The formula 
\begin{equation}\label{nonabel}
U_{t}=F(U,\, U_{1},\, \dots, \,U_{n}), \qquad F\in {\cal A} 
\end{equation} 
does not mean that we consider an element of non-associative algebra depending on time $t$. As usual, formula \eqref{nonabel}
defines a derivation $D_t$ of  ${\cal A}$, which commutes with the basic derivation
$$
D=\sum_{0}^{\infty} U_{i+1} \frac{\partial}{\partial U_{i}}.
$$
It is easy to check that $D_t$ is defined by the vector field
\[
D_t=\frac{\partial}{\partial t}+
\sum_0^{\infty} D^{i} (F) \frac{\partial}{\partial U_i}.
\]

A generalization of the symmetry approach to differential equations on
free associative algebras requires proper definitions for concepts such as symmetry,
conservation law, Fr\'echet  derivative, and formal symmetry.

As in the scalar case, a symmetry is an evolution equation
\[
U_{\tau}=G(U,\, U_{1},\, \dots, \,U_{m}),
\]
such that the vector field
\[
D_{G}=\sum_0^{\infty} D^{i} (G) \frac{\partial}{\partial u_i}
\]
commutes with $D_t$. The polynomial $G$ is called {\it the symmetry generator}.

The condition $[D_t, D_{G}]=0$ is equivalent to $D_t(G)=D_{G}(F)$. The latter
relation can be rewritten as
$$
G_{*}(F)-F_{*}(G)=0,
$$
where the Fr\'echet derivative $H_{*}$ for any  $H\in {\cal A}$ can be defined as follows.

For any $a \in {\cal A}$ we denote by $L_a$ and $R_a$ the operators of left and right multiplication by $a$:
$$
L_a(X)=a\,X, \qquad R_a(X)=X\,a, \qquad X\in {\cal A}.
$$
The associativity of ${\cal A}$ is equivalent to the identity $[L_a, R_b]=0$
for any $a$ and $b$. Moreover,
$$
L_{ab}=L_{a}\,L_{b}, \qquad R_{ab}=R_{b}\,R_{a},
\qquad L_{a+b}=L_{a}+L_{b},   \qquad R_{a+b}=R_{a}+R_{b}.
$$
\begin{definition}\label{localr}
We denote by ${\cal O}$ the associative algebra generated by
all operators of left and right multiplication by elements \eqref{generat}.
This algebra is called the {\it algebra of local operators}.
\end{definition}

Let us now extend the set of generators \eqref{generat} by adding non-commutative symbols
$V_0, V_1, \dots$ and let us prolong the derivation $D$ by $D(V_i)=V_{i+1}$.

Given $H(U, U_1, U_2, \dots, U_k) \in {\cal A},$ we find
$$L_H=\frac{\partial}{\partial \varepsilon} H(U+\varepsilon V_0, \
U_1+\varepsilon V_1, \
U_2+\varepsilon V_2, \dots) \big| _{\varepsilon=0}
$$
and represent this expression as $H_{*} (V_0)$, where $H_{*}$ is a linear
differential operator of order $k,$ whose coefficients belong to $\cal O$.
For example, $(U_2+U U_1)_{*}=D^2+L_U D+R_{U_1}$.

In contrast with the definition of the symmetry, which is a straightforward
generalization of the corresponding scalar notion, the definition of a
conserved density has to be essentially modified.

Recall that in the scalar case a conserved density is a
function $\rho \in {\cal F}$ such that $D_t(\rho)=D(\sigma)$ for some
$\sigma \in {\cal F}$. Equivalent densities define
the same functional, whose values do not depend on time (see Section \ref{conslaw}). Here, the equivalence relation is defined as follows:
$\rho_1 \sim \rho_2$ iff $\rho_1-\rho_2=D(s), \,\,\, s \in {\cal F}$.
In other words, a conserved density is an equivalence class in ${\cal F}$ such that
$D_t$ takes it to zero equivalence class.
 
In the non-abelian case we hold the same line.
The following elementary operations define an equivalence relation:
\begin{itemize} 
\item Addition of elements of the form $D(s)$, where $ s \in {\cal A}$, to
the element $\rho \in {\cal A}$;
\item Cyclic permutation of factors in any
monomial of the polynomial $\rho$.
\end{itemize}

Two densities $\rho_{1}, \rho_{2}\in {\cal A}$
related to each other through a finite sequence of the elementary operations
are called {\it equivalent}. An equivalent definition is  $\rho_{1} \sim \rho_{2}$ iff   
$$\rho_{1} - \rho_{2}\in {\cal A}/({\rm Im}\,D+[{\cal A},\, {\cal A}]).
$$
It is clear that in the scalar case this
definition coincides with Definition \ref{def28}.
\begin{definition} The equivalence class of an element $\rho$ is called the {\it trace of} $\rho$ and is denoted by ${\rm tr}\,\rho$.
\end{definition}

A motivation of the definition is that in the matrix case the
functional
$$ I({\bf U}) =\int_{-\infty}^{\infty} {\rm trace}(\rho({\bf U},{\bf U}_x,\dots))\, dx$$
is well defined on the equivalence classes if $\rho({\bf 0},{\bf 0},\dots)={\bf 0}$ and ${\bf U}(x,t) \to {\bf 0}$ as $x\to \pm \infty.$

\begin{definition} The equivalence class $\rho$ is called {conserved density} for equation \eqref{nonabel} if $D_t(\rho)\sim 0$.
\end{definition}

Poisson brackets \eqref{PDEbrop} are defined on the vector space ${\cal A}/({\rm Im}\,D+[{\cal A},\, {\cal A}])$.

A general theory of Poisson and double Poisson brackets on algebras of differential functions was developed in \cite{SKV}.  
An algebra of (non-commutative) differential functions is defined as a unital
associative algebra ${\cal D}$ with a derivation $D$ and commuting derivations $\partial_i, \,\, i\in \Z_{+}$
such that the following two properties hold:
\begin{itemize}
\item[1).] For each $f\in {\cal D}$ , $\partial_i(f)=0 $ for all but finitely many $i$;
\item[2).] $[\partial_i,\, D]=\partial_{i-1}.$
\end{itemize}

\subsubsection{Formal symmetry}
 
At least for non-abelian equations of the form \eqref{nonabel}, where
\begin{equation}\label{cononabel}
F=U_{n}+f(U,\, U_{1},\, \dots, \,U_{n-1})
\end{equation}
all definitions and results concerning the formal symmetry 
(see Section 3.1) can be easily generalized.

\begin{definition} A formal series
$$
\Lambda=D+l_{0}+l_{-1}D^{-1}+\cdots \, ,
\qquad l_k \in {\cal O} 
$$
is called a {\it formal symmetry} for equation \eqref{nonabel}, \eqref{cononabel}
if it satisfies equation
$$
D_t (\Lambda)-[F_{*},\,\Lambda]=0\, .
$$
\end{definition}
\begin{remark}
Stress that the coefficients of both  $F_{*}$ and $\Lambda$ belong not to ${\cal A}$ but to the
associative algebra ${\cal O}$ of local operators
{\rm (see Definition \ref{localr} above)}.
\end{remark}
For example, in the case of non-abelian Korteweg-de Vries equation
\eqref{matkdv} one can take $\Lambda={\cal R}^{1/2},$ where ${\cal R}$ is the following recursion
operator for \eqref{matkdv} (see \cite{OlvSok98})):
$$
{\cal R}=D^2+2(L_U+R_U)+(L_{U_{x}}+R_{U_{x}})\,D^{-1}+
(L_U-R_U)\,D^{-1}\,(L_U-R_U)\, D^{-1}.
$$
In the scalar case, this recursion operator coincides with the usual one
(see \eqref{recop}.

Analogs of Theorems \ref{tlsym}, \ref{svsok}  can be proved
by reasonings similar to the ones used for the original statements.

 \chapter{Integrable systems and  non--associative algebras}
 
 One of the most remarkable observations by S. Svinolupov is the
discovery of the fact that polynomial multi-component integrable equations
are closely related to the well-known nonassociative algebraic structures such 
as left-symmetric algebras, Jordan algebras, triple Jordan systems, etc.
This connection allows one to clarify the nature of known vector and
matrix generalizations (see, for instance \cite{For1,For2,For3}) of classical
scalar integrable equations and to construct some new examples of this
kind \cite{SviSok94}.

\section{Notation and definitions of algebraic structures}
Let ${\cal A}$ be an $N$-dimensional algebra over $\mathbb{C}$ with a product $\circ$, ${\bf e}_1, \dots, {\bf e}_N$ be a basis in ${\cal A}.$
By $C^{i}_{jk}$ we denote the   structural constants of ${\cal A}.$
For any triple system with an operation $\{\cdot,\cdot,\cdot\}$ we denote by $B^i_{jkm}$ the structural constants:
$$
\{{\bf e}_j, {\bf e}_k, {\bf e}_m\}=\sum^N_{i=1} B^i_{jkm} {\bf e}_i.
$$
We denote by $U$ the element \eqref{UU}.  Hereinafter, we use the notation \eqref{as}, \eqref{br}.

\subsection{Left-symmetric algebras}

\begin{definition}
Algebras with identity 
$[X, Y, Z]=0$ are called {\it left-symmetric} \cite{vinb}.
\end{definition}
 \begin{example} \label{ex41} Examples of left-symmetric algebras:
\begin{itemize} 
\item[\bf 1. ] Any associative algebra is left-symmetric.
\item[\bf 2. ] The operation
\begin{equation}\label{vecleft}
{\bf x}\circ {\bf y}=\langle {\bf x},{\bf c}\rangle\, {\bf y}+\langle {\bf x},{\bf y}\rangle \,{\bf c}, 
\end{equation}
where $ \langle \cdot, \cdot \rangle$ is the scalar product, ${\bf c}$ is a given vector, defines a left-symmetric algebra.
\item[\bf 3.]  
Let ${\mathfrak A}$ be associative algebra. Assume that $R: {\mathfrak A}\rightarrow {\mathfrak A}$
satisfies the modified classical Yang-Baxter equation
$$
R([R(x),y]-[R(y),x])=[x,y]+[R(x),R(y)].
$$
Then the multiplication
$$
x\circ y=[R(x),\,y] - (x y + y x)
$$
is left-symmetric.
\end{itemize}
\end{example}

\subsection{Jordan algebras}\label{Sec412}

\begin{definition} An algebra ${\cal A}$ is called Jordan if the following identities
\begin{equation}\label{jord}
X\circ Y=Y\circ X, \qquad X^2\circ (Y\circ X)=(X^2\circ Y)\circ X
\end{equation}
hold. 
\end{definition}
 
\begin{example}\label{ex42} Examples of simple Jordan algebras:

\begin{itemize}
\item[\bf 1.]  The set of all $m\times m$ matrices with respect to the operation 
\begin{equation}\label{jora1}
{\bf X}\circ {\bf Y}={\bf X}\,{\bf Y} + {\bf Y}\,{\bf X};
\end{equation}
\item[\bf 2.]  The set of all symmetric $m\times m$ matrices with the same operation \eqref{jora1};
\item[\bf 3.]  The set of all $N$-dimensional vectors with respect to the operation
\begin{equation}\label{jora2}
{\bf x} \circ {\bf y} =\langle {\bf x} ,{\bf c} \rangle\,{\bf y} +\langle {\bf y} ,{\bf c} \rangle\,{\bf x} -\langle {\bf x} ,{\bf y} \rangle\,{\bf c},
\end{equation}
where $ \langle \cdot, \cdot \rangle$ is the scalar product, ${\bf c}$ is a given vector;
\item[\bf 4.]  The special Jordan algebra $H_3(O)$ of dimension 27.
\end{itemize}
 \end{example}

\subsection{Triple Jordan systems}

\begin{definition} \label{def53} A triple system $T$ is called Jordan if the following identities hold
$$
\{X,Y,Z\}=\{Z,Y,X\},
$$
$$
\{X,Y,\{V,W,Z\}\}-\{V,W,\{X,Y,Z\}\}=
\{\{X,Y,V\},W,Z\}-\{V,\{Y,X,W\},Z\}.
$$
\end{definition}
\begin{example}\label{ex43} Examples of simple triple Jordan systems: 

\begin{itemize}
\item[\bf a)] The set of all $N\times N$ matrices with respect to the operation
\begin{equation} \label{jormat1}
\{{\bf X},\,{\bf Y},\,{\bf Z}\} = \frac{1}{2}\,\Big({\bf X}\, {\bf Y}\, {\bf Z}+{\bf Z}\,{\bf Y}\,{\bf X}\Big);
\end{equation}
\item[\bf b)]
The set of all skew-symmetric $N\times N$ matrices with the operation \eqref{jormat1};
\item[\bf c)] The set of all $N$-dimensional vectors with the operation
\begin{equation} \label{jorvec1}
\{{\bf x},{\bf y},{\bf z}\}=\langle{\bf x},{\bf y}\rangle\, {\bf z}+\langle {\bf y},{\bf z}\rangle\, {\bf x}-\langle{\bf x},{\bf z}\rangle \, {\bf y};
\end{equation}
\item[\bf d)] The set of all $N$-dimensional vectors with respect to
\begin{equation} \label{jorvec2}
\{{\bf x},{\bf y},{\bf z}\}=\langle{\bf x},{\bf y}\rangle \, {\bf z}+\langle{\bf y},{\bf z}\rangle \,{\bf x}.
\end{equation}
\end{itemize}
\end{example}
\begin{remark} There is the following generalization of the product \eqref{jorvec2}. The vector space of all $n\times m$-matrices is a triple Jordan system
with respect to the operation
\begin{equation}\label{nm} 
\{{\bf X},{\bf Y},{\bf Z}\}= {\bf X}\, {\bf Y}^t \,{\bf Z}+{\bf Z}\, {\bf Y}^t \, {\bf X}, 
\end{equation}
where ``t'' stands for the transposition.
\end{remark}

There are close relations between Jordan algebras and triple systems.

\begin{proposition}\label{pr41}
For any triple Jordan system $\{X,Y,Z\}$ the product 
\begin{equation}\label{jorC}
X\circ Y=\{X, C ,Y\},
\end{equation}
where $C$ is an arbitrary fixed element, defines a Jordan algebra. 
\end{proposition}
\begin{proposition}\label{pr42} For any Jordan algebra with a product $\circ$ the formula 
\begin{equation}\label{trjor}
\{X,Y,Z\}=(X\circ Y)\circ Z+(Z\circ Y)\circ X - Y\circ (X\circ Z)
\end{equation}
defines a triple Jordan system.
\end{proposition}

\begin{proposition}\label{pr43} For each triple Jordan system $\{X,Y,Z\}$ and any element $C$ the formula
$$ 
\sigma(X,Y,Z)=\{X,\{C,Y,C\},Z\} \label{newtr}
$$
defines a triple Jordan system.
\end{proposition}
\begin{proof} To prove the statement, we apply Proposition \ref{pr42} to multiplication \eqref{jorC} and take into account the identities of Definition \ref{def53}. 
\end{proof}

\section{Jordan KdV systems}\label{JordKdv}
 Consider equations of the form \eqref{nonabel}, 
where $F$ is a non-commutative and non-associative polynomial.  Suppose that $F$ is fixed but the multiplication in the algebra is 
unknown. 

For example,  let us consider the KdV equation
\begin{equation}\label{Ukdv}
U_t=U_{xxx}+ 3\,U \circ U_x,
\end{equation}
where $\circ$ is a multiplication in some algebra ${\cal A}$. The main question is: 
\begin{question} For which algebras ${\cal A}$ is this equation integrable?
\end{question}

If $U$ is given by \eqref{UU}, then equation \eqref{Ukdv} is equivalent  to  the following $N$-component system
\begin{equation} \label{kort}
u^i_t=u^i_{xxx}+\sum_{k,j}\, C^i_{jk} \,u^k \, u^j_x, \qquad i,j,k=1,\dots,N.
\end{equation}
Vice versa any system of the form \eqref{kort} can be written as \eqref{Ukdv} for a proper algebra ${\cal A}$.  Systems \eqref{kort} seem to be a natural multi-component generalization of the KdV equation \eqref{kdv}.

\begin{remark}\label{rem41} Let $I$ be a double--sided  ideal in ${\cal A}.$  Choose a basis such that  ${\bf e}_{M+1}, \dots, {\bf e}_{N}$ span $I$. For such a basis we have $C^i_{jk}=0$ for $i\le M$ and $j>M$ or $k>M$. This means that the equations for $u_1,\dots, u_M$ constitute a closed subsystem of the same form 
\eqref{kort} and the whole system has a ``triangular'' form.  
\end{remark}

Any linear
transformation of $\vec u=(u_1,\dots,u_N)$ preserves the class of systems \eqref{kort}.  A reasonable description of integrable cases has to be invariant under these transformations. 

The system \eqref{kort} is homogeneous \eqref{homo} with $\mu=3,\, \lambda_i=2.$ Therefore, without loss of generality we may assume that the polynomial higher symmetries for \eqref{kort} are also homogeneous. Any such a symmetry has the following form  
\begin{equation}\label{Usym}
U_{\tau}=U_m + A_1(U, U_{n-2})+ A_2(U_1, U_{n-3})+\cdots .
\end{equation}
Here, the quadratic terms are defined by unknown bi-linear operations $A_i$, the cubic terms are defined by three-linear operations and so on\footnote{In a weak version of the symmetry approach one can assume that operations in \eqref{Usym} are expressed in terms of the multiplication $\circ$ in ${\cal A}$ like $A_1(X,\,Y)=c_1\, X\circ Y+c_2\, Y\circ X$ and so on.}.  

\begin{theorem}\label{svin} {\rm \cite{svin3}}. Suppose that the algebra ${\cal A}$ is finite-dimensional and commutative. Then equation \eqref{Ukdv} has a higher symmetry of the form \eqref{Usym},
where $m \ge 5$  iff ${\cal A}$ is a { \it Jordan algebra}.
\end{theorem}

The original proof of Theorem \ref{svin} was done for the system \eqref{kort} by straightforward computations in terms of structural constants. It turns out 
that the computations can be performed in terms of algebraic operations, which define the equation and the symmetry. 

In the following theorem we do not assume that the algebra ${\cal A}$ is commutative or finite-dimensional. For simplicity, we assume that equation \eqref{Ukdv} has a symmetry \eqref{Usym} of fifth order. For the KdV equation such a symmetry is given by \eqref{kdvsym}. 

\begin{theorem}\label{Svin} Suppose that the algebra ${\cal A}$ possesses the following property: \, if for any $Z$ we have
$$
(X\circ Y-Y\circ X)\circ Z = 0,
$$
then $X\circ Y-Y\circ X=0.$  Equation \eqref{Ukdv} has a homogeneous symmetry \eqref{Usym} of fifth order iff ${\cal A}$ is a Jordan algebra. 
\end{theorem}
\begin{proof} The symmetry can be written in the form
$$
U_{\tau}=U_5+5\, A_1(U, U_3)+5\,A_2(U_1, U_2)+5\,B_1(U,U,U_1),
$$
where $A_1$ and $A_2$ are unknown bi-linear operations and $B_1$ is a three-linear operation. The symmetry condition 
$$
D_t(D_\tau(U))= D_\tau(D_t(U))
$$
is an identity with respect to the independent variables $U_5, U_4, \dots,U_0=U$. Let us perform the scaling $U_i\to z_i U_i, \,\,\, z_i\in \C$ and collect coefficients of different monomials in $z_i$. Comparing the coefficients of  $z_5 z_1,$ we get $A_1(X,Y)=X \circ Y$. The terms that involve $z_4 z_2$ gives us $A_2(X,Y)=X \circ Y+Y \circ X$.  The terms with $z_3 z_1 z_0$ implies identities
$$
B_1(X,X,Y)=\frac{1}{2} \Big( 2\, X\circ (X \circ Y)+(X \circ X)\circ Y\Big)
$$
and 
\begin{equation}\label{kdvid}
\Big(X\circ Y-Y \circ X\Big)\circ Z=0.
\end{equation}
So, we expressed the symmetry in terms of multiplication in the algebra ${\cal A}.$ Moreover, it follows from the assumption of the theorem that ${\cal A}$ is commutative.  Now, we obtain the Jordan identity from the coefficients of $z_2 z_0^3.$ The remaining identities are satisfied in virtue of \eqref{jord}. 
\end{proof}
\begin{remark}
It is easy to see that any simple algebra satisfies the requirement  of Theorem \ref{Svin}. 
\end{remark}
\begin{exercise} Prove that equation \eqref{Ukdv} has a higher homogeneous symmetry of the form \eqref{Usym},
where $m \ge 5$ iff the identities
$$
(X\circ X)\circ (X\circ Y) = X\circ \Big((X\circ X)\circ Y   \Big), 
$$
$$
2 \, (X\circ Y)\circ (X\circ Y)-2 X\circ \Big(Y\circ (X\circ Y)  \Big)+(X\circ X)\circ (Y\circ Y)- \Big((X\circ X)\circ Y\Big)\circ Y=0,
$$
and \eqref{kdvid} hold. 
\end{exercise}

According to Remark \ref{rem41}, the most interesting non-triangular integrable equations \eqref{Ukdv} correspond to simple Jordan algebras. Classification of such algebras can be found in \cite{jac}.  All simple Jordan algebras are described in Example \ref{ex42}.

\subsubsection{Integrable Jordan non-triangular KdV equations}
\begin{itemize}
\item[\bf 1)] The matrix KdV equation
$$
{\bf U}_t={\bf U}_{xxx}+{\bf U} {\bf U}_x+{\bf U}_x {\bf U},
$$ 
where ${\bf U}$ is an $m\times m$-matrix, is generated by \eqref{jora1}. It coincides with \eqref{matkdv} up to a scaling of ${\bf U}$;
\item[\bf 2)] The matrix  KdV equation under the reduction ${\bf U}^t={\bf U}$;
\item[\bf 3)] The vector KdV equation   \cite{SviSok94}:
\begin{equation}\label{VectorKdv}
{\bf u}_t={\bf u}_{xxx}+\langle {\bf c},\,{\bf u}\rangle \,{\bf u}_x+\langle {\bf c},\,{\bf u}_x\rangle \, {\bf u}-\langle {\bf u},\,{\bf u}_x\rangle \, {\bf c},
\end{equation}
where   $\bf c$ is a given constant vector, corresponds to the product \eqref{jora2}.
\end{itemize}
\begin{op} Find solitonic and finite-gap solutions for  the KdV equations related to the special Jordan algebra $H_3(O)$
\end{op}

\section{Left-symmetric algebras and Burgers type systems}

Consider the equation
\begin{equation}\label{Ubur}
U_t=U_{2} + 2\,U\circ U_1 + B(U,U,U),
\end{equation}
where $B$ is a three-linear operation. 
In the finite-dimensional case it is equivalent to a system of evolution equations of the form
\begin{equation} \label{burg}
u^i_t=u^i_{xx}+2 C^i_{jk} u^k u^j_x+B^i_{jkm} u^k u^j u^m,
\end{equation}
where $\quad i,j,k=1,\dots,N$.

The system \eqref{burg} is homogeneous \eqref{homo} with $\mu=2, \lambda_i=1.$  The Burgers equation \eqref{burgers} has a homogeneous symmetry \eqref{burgersym} of third order. The general ansatz of such symmetry in the case of equations \eqref{Ubur} is given by 
\begin{equation}\label{Ubursym}
U_{\tau}=U_3+3 A_1(U,U_2)+3 A_2(U_1,U_1)+3 B_1(U,U,U_1)+C_1(U,U,U,U)
\end{equation}
Denote $A(X,Y)$ by $X\circ Y$.
\begin{theorem}
Equation \eqref{Ubur} possesses a symmetry of the form \eqref{Ubursym} iff
$$B(X,X,X)=X\circ (X\circ X)- (X\circ X)\circ X,$$ and $\circ$ is a left-symmetric product.
The operations in \eqref{Ubursym} have the following form:
$$
A_1(X,Y)=X\circ Y, \qquad A_2(X,X)=X\circ X, \qquad B_1(X,X,Y)=X\circ (X\circ Y)+Y\circ (X\circ X)-(Y\circ X)\circ X,
$$
$$
C_1(X,X,X,X)=X\circ(X\circ (X\circ X))-X\circ ((X\circ X)\circ X)+(X\circ X)\circ (X\circ X)-((X\circ X)\circ X)\circ X.
$$
\end{theorem}

 In contrast with Jordan algebras, there is no complete classification of finite-dimensional simple left-symmetric algebras. 

Any associative algebra is left-symmetric. The system, which corresponds to the case of the matrix algebra, is the matrix 
Burgers equation \eqref{matbur}. 
 The vector Burgers equation   \cite{SviSok94}
$$
 {\bf u}_t={\bf u}_{xx}+2 \langle{\bf u},{\bf u}_x\rangle \,{\bf c}+2\langle{\bf c},{\bf u}\rangle \,{\bf u}_x+  \langle{\bf u}, {\bf u}\rangle\langle{\bf c},{\bf u}\rangle \,{\bf c}-\langle{\bf u},{\bf u}\rangle\langle{\bf c},{\bf c}\rangle \,{\bf u},
 $$
where ${\bf c}$ is a constant vector, comes from the left-symmetric algebra with the product \eqref{vecleft}.

\chapter{Integrable models associated with triple systems.}

\section{Integrability and  triple Jordan systems}\label{sec54}

Here, we consider some classes of polynomial integrable
systems of evolution equations with cubic non-linear terms. 
They are generalizations of the following famous scalar
integrable equations:  
the modified Korteweg-de Vries equation \eqref{mkdv}, 
the nonlinear Schr\"odinger equation \eqref{NLS},  and the nonlinear
derivative Schr\"odinger equation
$$
u_t = u_{xx} + 2 (u^2 v)_x, \qquad  v_t = - v_{xx} - 2 (v^2 u)_x.
$$
 
\subsection{MKdV-type systems}
Consider systems of the form
\begin{equation}\label{jormkdv}
u^i_t=u^i_{xxx}+\sum_{j,k,m} B^i_{jkm} u^k u^j u^m_x, \qquad i,j,k=1,\dots,N.
\end{equation}
Let $T$ be a  triple system with basis ${\bf e}_1,...,{\bf e}_N,$ such that 
$$
\{{\bf e}_j,{\bf e}_k,{\bf e}_m\} =\sum_i B^i_{jkm} {\bf e}_i.
$$
If $\, U=\sum_k u^k {\bf e}_k, \,$ then 
the algebraic form of the equation is given by
\begin{equation} \label{mkdvtr}
U_t = U_{xxx} + B(U, U_x, U).
\end{equation}
The triple systems $B(X,Y,Z)$ such that  $B(X,Y,Z) = B(Z,Y,X)$  are in one-to-one correspondence with the systems \eqref{jormkdv}.

Equations \eqref{mkdvtr} are homogeneous \eqref{homo} with $\mu=3, \lambda=1.$ 
They are also invariant with respect to the discrete involution $U\to -U.$ Without loss of generality we assume that all symmetries enjoy the same properties.

\begin{theorem}{\rm \cite{svin3}}\label{th54} For any  triple Jordan system $\{\cdot,\cdot,\cdot\}$, the equation \eqref{mkdvtr}, where $B(X,Y,X)=\{X,X,Y \}$  has a fifth order symmetry of the form 
\begin{equation}\label{symmkdv}
U_{\tau}=U_{5}+B_1(U,U,U_3)+B_2(U,U_1,U_2)+B_3(U_1,U_1,U_1)+C(U,U,U,U,U_1).
\end{equation}
\end{theorem}

The converse of this statement was proved by I. Shestakov and VS. 
\begin{theorem} The equation \eqref{mkdvtr} has a fifth 
order symmetry of the form \eqref{symmkdv}
iff
\begin{equation}\label{BBB}
B(X,Y,Z)=\{X,Z,Y\}+\{Z,X,Y\},
\end{equation}
where $\{\cdot,\cdot,\cdot\}$ is a triple Jordan system.
\end{theorem}
\begin{proof}
The compatibility condition
\begin{equation}\label{com}
0=(U_t)_{\tau}-(U_{\tau})_t=P(U,U_1,...,U_5)
\end{equation}
 of \eqref{mkdvtr} and \eqref{symmkdv} leads to a defining polynomial $P$ that should be identically zero. After the scaling $U_i \to z_i U_i$ in $F$ all coefficients of different monomials in $z_0,..., z_5$ have to be identically zero. Equating the coefficient of $z_0 z_1 z_5$ to zero, we find that 
\begin{equation}\label{B1}
B_1(X, X, Y)=B(X, Y, X).
\end{equation}
The coefficient of $z_0 z_2 z_4$ leads to 
\begin{equation}\label{B2}
B_2(X, Y, Z)= 2 B(X, Y, Z)+2 B(X, Z, Y).
\end{equation}
All other terms with $z_5$ and $z_4$ disappear by virtue of \eqref{B1} and \eqref{B2}.
Comparing the coefficients of $z_1 z_2 z_3$, we obtain 
\begin{equation}\label{B2}
B_3(X, X, X)= B(X, X, X),
\end{equation}
while the coefficients of $z_3 z_1 z_0^3$ give rise to 
\begin{equation}\label{CC}
C(X, X, X, X, Y)= B(X, B(X,Y,X), X) + \frac{1}{2} B(X,Y,B(X,X,X)).
\end{equation}
Thus the symmetry \eqref{symmkdv} is expressed in terms of the triple system $B.$ All fifth order identities $I_i=0, \,i=1,2,3,4$ for the triple system $B$ come from coefficients of the monomials $z_0^3 z_1 z_3, z_0^3 z_2^2, z_0^2 z_1^2 z_2$ and $z_0 z_1^4$. They are defined by the formulas
$$
I_1(X,Y,Z)=2 B(X,Z,B(X,X,Y))-3 B(X,Z,B(X,Y,X))+B(Y,Z,B(X,X,X),
$$
$$
I_2(X,Y)=2 B(X,Y,B(X,X,Y))-3 B(X,Y,B(X,Y,X))+B(Y,Y,B(X,X,X)),
$$

$$\begin{array}{l}
I_3(X,Y,Z)=-2 B(X, Y, B(X, Y, Z)) + 6 B(X, Y, B(X, Z, Y)) - B(X, Y, B(Y, X, Z)) -  \\[1.5mm]
  B(X, Y, B(Z, X, Y))+4 B(X, Z, B(X, Y, Y)) - 
 2 B(X, Z, B(Y, X, Y)) - 2 B(X, B(Y, Y, Z), X) + \\[1.5mm]
 2 B(X, B(Y, Z, Y), X) + 2 B(X, B(Z, Y, Y), X) - 
 2 B(Y, Y, B(X, X, Z)) + 3 B(Y, Y, B(X, Z, X)) - \\[1.5mm]
 4 B(Y, Z, B(X, X, Y)) + 2 B(Y, Z, B(X, Y, X)) - 
 2 B(Y, B(X, Y, X), Z) - 2 B(Y, B(X, Z, X), Y) - \\[1.5mm]
 2 B(Z, Y, B(X, X, Y)) + 3 B(Z, Y, B(X, Y, X)) - 2 B(Z, B(X, Y, X), Y),
 \end{array}
$$
and
$$
I_4(X,Y)=B(X, Y, B(Y, Y, Y)) + 2 B(Y, Y, B(X, Y, Y)) - B(Y, Y, B(Y, X, Y)) - 
 2 B(Y, B(X, Y, Y), Y).
$$
It is clear that $I_2(X,Y)=I_1(X,Y,Y).$ Using  the method of undetermined coefficient, we will show that the identity $I_4=0$ is a consequence of the identities $I_2=0$ and $I_3=0.$ First, introduce the polarizations of these identities. 
Let \begin{itemize} \item $J_2(X,Y,Z,U,V)$ be the coefficient of $k_1 k_2 k_3 $ in $I_2(k_1 X + k_2 U + k_3 V, Y, Z)$; 
\item $J_3(X,Y,Z,U,V)$ be the coefficient of $k_1 k_2 k_3 k_4$ in $I_3(k_1 X + k_2 U, k_3 Y + k_4 V, Z)$;
\item $J_4(X,Y,Z,U,V)$ be the coefficient of $k_1 k_2 k_3 k_4$ in $I_4(X, k_1 Y + k_2 Z + k_3 U + k_4 V)$.
\end{itemize}
 Consider the following ansatz
 $$
 \begin{array}{c}
 Z=J_4(X,Y,Z,U,V)-\sum_{\sigma\in S_5} b_{\sigma} J_2\Big(\sigma(X),\sigma(Y),\sigma(Z),\sigma(U),\sigma(V)\Big)-\\[2mm]
  \sum_{\sigma\in S_5} c_{\sigma} J_3\Big(\sigma(X),\sigma(Y),\sigma(Z),\sigma(U),\sigma(V)\Big),
 \end{array}
 $$
 where $\sigma$ is a permutation of the set $\{X,Y,Z,U,V\}.$ To take into account the identity $B(X,Y,Z)=B(Z,Y,X),$ we fix the ordering
 $$
 U<V<X<Y<Z< B(\cdot,\cdot,\cdot)
 $$
 and replace all expressions of the form $B(P,Q,R)$ by  $B(R,Q,P)$ if $P>R.$ After that, equating the coefficients of similar terms in the relation $Z=0$, we obtain an overdetermined system of linear equations for the coefficients $b_{\sigma}$ and $c_{\sigma}$. Solving this system by computer we find that
 $$
 \begin{array}{l}
 \ds J_4(X,Y,Z,U,V)=\frac{1}{6}\Big(J_2(U, X, V, Y, Z) + J_2(U, X, Y, V, Z) + J_2(U, X, Z, V, Y) + 
 J_2(V, X, U, Y, Z) -\\[1.5mm] J_3(U, V, X, Y, Z) - J_3(U, V, X, Z, Y) - 
 J_3(U, V, Y, X, Z) - J_3(U, V, Z, X, Y) - J_3(U, Y, V, X, Z) - \\[1.5mm]
 J_3(U, Y, X, V, Z) - J_3(V, U, X, Y, Z) - J_3(V, U, X, Z, Y) - 
 J_3(V, U, Y, X, Z) - J_3(V, U, Z, X, Y) -\\[1.5mm] J_3(V, Y, U, X, Z) - 
 J_3(X, U, V, Y, Z) - J_3(X, U, V, Z, Y) - J_3(X, U, Y, Z, V) - 
 J_3(X, U, Z, Y, V) - \\[1.5mm] J_3(X, V, U, Y, Z) -J_3(X, V, U, Z, Y) - 
 J_3(Y, U, X, Z, V)\Big)
 \end{array}
 $$

Consider a triple system 
$$
\{X,Y,Z\} = \frac{1}{2} \Big( B(Y,Z,X)+B(Y,X,Z)-B(X,Y,Z)  \Big).
$$
It is easy to verify that
\begin{equation}\label{BB}
B(X,Y,Z)= \{X,Z,Y\} + \{Z,X,Y\}.
\end{equation}
Let us prove that the identities $J_2=J_3=0$ are equivalent to one identity ${\cal J}=0$, where (see Definition \ref{def53} of Jordan triple system)
$$
{\cal J}(X,Y,Z,U,V)=\{X, Y, \{U, V, Z\}\} - \{\{X, Y, U\}, V, Z\} - 
 \{U, V, \{X, Y, Z\}\} + \{U,\{ Y, X, V\}, Z\}.
$$
The identity ${\cal J}=0$ is supposed to be rewritten in terms of the triple system $B$ by means of \eqref{BB}.
 By the same method of undetermined coefficients we have verified that the identity ${\cal J}=0$ follows from $J_2=J_3=0$ and vice versa
 each of the identities $J_2$ and $J_3=0$ follows from ${\cal J}=0.$ For example, 
 $$
 \begin{array}{c}
J_2(X,Y,Z,U,V) = -{\cal J}(U, X, V, Z, Y) + {\cal J}(U, X, Y, Z, V) - {\cal J}(U, Y, X, Z, V) + \\[1.5mm]
 {\cal J}(V, X, Y, Z, U) - {\cal J}(V, Y, U, Z, X) - {\cal J}(V, Y, X, Z, U) - \\[1.5mm]
 {\cal J}(X, U, V, Z, Y) + {\cal J}(X, U, Y, Z, V) - {\cal J}(X, V, U, Z, Y) + \\[1.5mm]
 {\cal J}(Y, U, V, Z, X) + {\cal J}(Y, V, U, Z, X) + {\cal J}(Y, V, X, Z, U).
 \end{array}
 $$
 The formulas, which express $J_3$ through ${\cal J}$ and ${\cal J}$ through $J_2, J_3$, are much more complicated.
 
 Besides the above fifth order identities there exist two identities of order 7. The coefficient of $z_0^6 z_2$ in the polynomial $P$ yields 
 $$
 B(X, B(X, Y, X),\, B(X, X, X)) - B(X, B(X, Y, B(X, X, X)), X)=0
 $$
 while the coefficient of $z_0^5 z_1^2$ leads to
 $$
 \begin{array}{l}
 2 B[X, Y, B[X, X, B[X, Y, X]]] - 2 B[X, Y, B[X, Y, B[X, X, X]]] - 
 3 B[X, Y, B[X, B[X, Y, X], X]] +\\[1.5mm] 4 B[X, B[X, Y, X], B[X, Y, X]] + 
 2 B[X, B[X, Y, Y], B[X, X, X]] - 2 B[X, B[X, Y, B[X, X, Y]], X] + \\[1.5mm]
 3 B[X, B[X, Y, B[X, Y, X]], X] - 4 B[X, B[X, B[X, Y, X], Y], X] - \\[1.5mm]
 B[X, B[Y, Y, B[X, X, X]], X] + B[B[X, X, X], Y, B[X, Y, X]]=0.
 \end{array}
 $$
 Using the method of undetermined coefficients, one can check that both identities follows from ${\cal J}=0$.  
 \end{proof}
 
\begin{remark} Since equation \eqref{mkdvtr} is expressed throught $B(X,Y,X)$, it follows from \eqref{BB} that all equations that have the fifth order symmetry are described by Theorem \ref{th54}.
\end{remark}

The  integrable vector systems  corresponding to the operations \eqref{jorvec1} and \eqref{jorvec2} are given by
\begin{equation}\label{vecmkdv1}
{\bf u}_t={\bf u}_{xxx}+\langle{\bf u}, {\bf u}\rangle \, {\bf u}_x
\end{equation}
and
\begin{equation} \label{vecmkdv2}
{\bf u}_t={\bf u}_{xxx}+\langle{\bf u}, {\bf u}\rangle \, {\bf u}_x+\langle{\bf u}, {\bf u}_x\rangle \, {\bf u},
\end{equation}
respectively.

The triple Jordan system \eqref{jormat1} generates the matrix mKdV equation \eqref{matmkdv}.

 \subsection{NLS-type systems}

Multi-component generalizations (see \cite{For1, sv15}) of the nonlinear
Schr\"odinger equation \eqref{NLS}  are systems of $2N$ equations
of the following form
\begin{equation}
u_t^i=u_{xx}^i+2 \sum_{j,k,m}\,b_{jkm}^i u^j v^k u^m, \qquad 
v_t^i=-v_{xx}^i-2 \sum_{j,k,m}\,b_{jkm}^i v^j u^k v^m,
\label{mse}
\end{equation}
where $i=1,\dots,N$, and $a^i_{jkm}$ are
constants.  In terms of the corresponding triple system it can be written in the form 
\begin{equation}
U_t=U_{xx}+2 B(U,V,U),\qquad V_t=-V_{xx}-2 B(V,U,V).        \label{nmse}
\end{equation}
The triple systems $B(X,Y,Z)$ such that  $B(X,Y,Z) = B(Z,Y,X)$  are in one-to-one correspondence with the systems \eqref{mse}.
The systems \eqref{nmse} are homogeneous \eqref{homo} with $\mu=2, \lambda_i=1$.  
\begin{remark} The system \eqref{nmse} has the following additional scaling symmetry $U\to \lambda\, U,\,$ $V\to \lambda^{-1} V$. Therefore, we can assume without loss of generality that all symmetries of this equation are homogeneous and invariant with respect to the same scaling.
\end{remark}

\begin{theorem}{\rm \cite{sv15}} The equation \eqref{nmse}  has a third order symmetry 
of the form
$$
U_{\tau}=U_{xxx}+3 B_1(U,V,U_x)+3 B_3(U,V_1,U), \qquad V_{\tau}=U_{xxx}+3 B_2(V,U,V_x)+3 B_4(V,U_1,V)
$$
iff $B_1(X,Y,Z)=B_2(X,Y,Z)=B(X,Y,Z),$ $\quad B_3(X,Y,X)=B_4(X,Y,X)=0\,\,$ and $B(\cdot,\cdot,\cdot)$ is a triple Jordan system.
\end{theorem}
 
The above examples of simple triple Jordan systems from Section 4.1.3 provide several interesting vector and matrix integrable systems. 

\begin{example} The triple system \eqref{jormat1} produces the matrix NLS-system \eqref{matnls} up to a scaling.
\end{example} 

\begin{example} The well-known vector Schr\"odinger equation  \cite{Man}
\begin{equation}\label{vecNLS}
 {\bf u}_t={\bf u}_{xx}+2\langle {\bf u},{\bf v}\rangle \, {\bf u} ,\qquad
{\bf v}_t=-{\bf v}_{xx}-2\langle {\bf v},{\bf u}\rangle \,{\bf v}     
\end{equation}
 corresponds to the triple Jordan system \eqref{jorvec2}.
\end{example}

\begin{example} One more  integrable vector nonlinear Schr\"odinger equation
\begin{equation}
 {\bf u}_t={\bf u}_{xx}+2\langle {\bf u},{\bf v}\rangle \, {\bf u}-\langle {\bf u},{\bf u}\rangle\, {\bf v},\qquad
{\bf v}_t=-{\bf v}_{xx}-2\langle {\bf v},{\bf u}\rangle \,{\bf v}+\,\langle {\bf v},{\bf v}\rangle\, {\bf u}               \label{ex3}
\end{equation}
found in \cite{kuli1981}
corresponds to the triple Jordan system \eqref{jorvec1}. \end{example} 

\subsection{Derivative NLS-type systems}\label{secderiv}
The derivative NLS equation is given by
$$
u_t = u_{xx} + 2\, (u^2 v)_x, \qquad v_t=-v_{xx} + 2\, (v^2 u)_x.  
$$
Consider its following generalization: 
\begin{equation}\label{mderiv}
U_t = U_{xx} + B(U,V,U)_x, \qquad V_t=-V_{xx} + B(V,U,V)_x,
\end{equation}
where $B(X,Y,Z)=B(Z,Y,X).$ The systems \eqref{mderiv} are homogeneous \eqref{homo} with $\mu=1, \lambda_i=\frac{1}{2}.$ 
\begin{theorem} A system \eqref{mderiv} has a third order polynomial symmetry of the form 
\begin{equation}\label{sym33}
U_{\tau}=U_{xxx} + P(U, V, U_x, V_x, U_{xx}, V_{xx}), \qquad V_{\tau}=V_{xxx} + Q(U, V, U_x, V_x, U_{xx}, V_{xx}),
\end{equation}
iff $B(\cdot,\cdot,\cdot)$ is a
triple Jordan system. 
\end{theorem}
Matrix and vector examples of integrable derivarive NLS type systems can be 
constructed using formulas \eqref{jormat1}-\eqref{jorvec2}.

\section{Some integrable systems corresponding to new classes of triple systems}

\subsection{Equations of potential mKdV type}  Consider  multi-component generalizations of known integrable equation (see \eqref{e2} and Example \ref{ex63})
\begin{equation}\label{potmkdv}
u_t=u_{xxx}+ 3\, u_x^3.
\end{equation}
This equation has the following fifth order symmetry 
$$
u_{\tau}=u_{xxxxx}+ 15 u_x^2 u_{xxx}+15 u_x u_{xx}^2+\frac{27}{2} u_x^5.
$$
The multi-component generalizations of \eqref{potmkdv} have the form 
\begin{equation} \label{potsys}
U_t= U_{xxx}+ 3 B(U_x, U_x, U_x), 
\end{equation}
where 
$$
B(X,Y,Z) = B\Big(\sigma(X), \sigma(Y), \sigma(Z)\Big), \qquad \sigma\in S_3.
$$
\begin{theorem} Equation \eqref{potsys} has a symmetry of the form 
$$
U_{\tau}=U_{5}+ 15 B_1(U_1, U_1, U_{3}) + 15 B_2(U_{2}, U_{2}, U_1) + \frac{27}{2} C( U_1, U_1, U_1, U_1, U_1)
$$
iff the symmetric triple system $B(\cdot,\cdot,\cdot)$ satisfies the identity 
$$
B(X, Y, B(Y, Y, Z)) + 
 B(Y, Y, B(X, Y, Z)) + B(Y, Z, B(X, Y, Y))  - 3 B(X, Z, B(Y, Y, Y)) =0. 
$$
\end{theorem}
\begin{remark} The symmetry is defined by 
$$
B_1(X,X,Y)= B(X, X, Y), \qquad B_{2}(X, X, Y)=B(X, X, Y), $$$$
C(X,X,X,X,X)=B(X, X, B(X, X, X)).
$$
\end{remark}

\subsection{Systems of Olver--Sokolov type}

A multi-component generalization of the matrix system \eqref{olsokder} reads as follows
\begin{equation}\label{osder}
U_t = U_{xx} + 2\, B(U,V,U_x), \qquad V_t=-V_{xx} + 2\,B(V_x,U,V),
\end{equation}
where $B(\cdot,\cdot,\cdot)$ is a triple system. The systems of the form \eqref{osder} are homogeneous with the same weights as 
in Section \ref{secderiv}.

\begin{theorem}\label{th47} A system \eqref{osder} has a third order polynomial symmetry of the form \eqref{sym33} iff $B(\cdot,\cdot,\cdot)$ satisfies the identities
\begin{equation} \label{soksh}
\begin{array}{c}
B\Big(X, Y, B(X,V,U) \Big) = B\Big(X, B(Y,X,V), U \Big), \\[2mm] 
B\Big(X, B(Y,Z,V), Z \Big) = B\Big(B(X,Y,Z), V , Z \Big), \\[2mm] 
B\Big(X, Y, B(Z,Y,V)\Big) = B\Big(B(X,Y,Z), Y , V \Big).
\end{array}
\end{equation}
\end{theorem}
\begin{remark} The latter identity means that the multiplication
$$
X \circ Y = B(X,Z,Y)
$$ 
is associative for any $Z.$ Moreover, varying $Z,$ we obtain a vector space of compatible associative products {\rm \cite{odsok1}}. 
\end{remark} 

In the paper \cite{loos} the associative triple systems with identities
\begin{equation}\label{astrip}
B\Big(X, Y, B(Z,U,V) \Big) = B\Big(X, B(Y,Z,U), V \Big) = B\Big(B(X,Y,Z),U, V \Big)
\end{equation}
were considered. It is clear that \eqref{astrip} implies \eqref{soksh}.

Since the triple matrix product 
$$
B(X,Y,Z)= X Y Z, \qquad X,Y,Z \in Mat_{n}
$$
satisfies \eqref{astrip}, the system  \eqref{olsokder} belongs to the class described in Theorem \ref{th47}. One more example of a triple system satisfying \eqref{astrip} is given by
$$
B(X,Y,Z)= X Y^t Z, \qquad X,Y,Z \in Mat_{n,m}.
$$

As far as I know, triple systems \eqref{soksh} have never been considered by algebraists.

\begin{op} Find all simple triple systems \eqref{soksh}.
\end{op}

\chapter{Integrability and deformations of non-associative structures}
 
\section{Inverse element in triple Jordan systems and rational integrable equations}

\subsection{Inverse element as a solution of PDE-system}

The function $\ds y=\frac{1}{x}$ can be defined as the homogeneous solution of the differential equation $$y'+y^2=0.$$  
It turns out that an inverse element for triple Jordan systems can be defined in a similar way. 

Let $\{X,Y,Z\}$ be a triple system, ${\bf e}_1,\dots,{\bf e}_N$ be its basis, $U=\sum_k u_k {\bf e}_k.$ Our aim is to define an  element
\begin{equation}\label{inverse}
\phi(U)=\sum_{k=1}^N \phi_k(u_1,\dots,u_N)\, {\bf e}_k
\end{equation}
 inverse for $U$ as a solution of a proper system of PDEs.
\begin{proposition} The following overdetermined system 
\begin{equation}
\frac{\partial \phi}{\partial u^k}=-\{\phi, \ {\bf e}_k, \  \phi\},
\label{vecdef} \end{equation}
where $k=1,\dots ,N$, is compatible iff $\{X,Y,Z\}$ is a triple Jordan system.
\end{proposition}
\begin{definition} For any triple Jordan system any homogeneous\footnote{This means that $\sum u_i \frac{\partial \phi}{\partial u_i}=- \phi$.} solution \eqref{inverse} of the system \eqref{vecdef} is called 
the {\it inverse element} for the element $U$.
\end{definition}
\begin{remark}
It can be easily seen that any homogeneous solution of  \eqref{vecdef} satisfies the identity
\begin{equation}\label{Uphi}
\phi(U)=\{\phi(U),\, U,\, \phi(U)\}.
\end{equation}
\end{remark}
\begin{example}
For the matrix triple Jordan system \eqref{jormat1}  
the homogeneous solution $ \phi({\bf U})$
of \eqref{vecdef} is just the matrix inverse ${\bf U}^{-1}.$
\end{example}
\begin{example} For the triple Jordan system \eqref{jorvec1} the homogeneous solution for \eqref{vecdef} is given by
$$
 \phi({\bf u})=\frac{{\bf u}}{\vert {\bf u} \vert^2}.
$$
\end{example}

The following algebraic definition of $U^{-1}$ is well known in the theory of triple Jordan systems.
Let us define a linear operator $P_X$ by the formula
$P_X(Y)=\{X,Y,X\}$. If $P_{U}$ is non-degenerate, then by definition $U^{-1}=P_{U}^{-1}(U)$.
\begin{proposition}\label{pr45}
 If $P_{U}$ is non-degenerate, then
 \begin{equation}\label{phiu}
 \phi(U) = P_{U}^{-1}(U)
 \end{equation}
 satisfies \eqref{vecdef}.
\end{proposition}
\begin{remark}
It follows from \eqref{Uphi} and from Proposition \ref{pr45} that if $P_{U}$ is non-degenerate, then there exists a unique homogeneous solution of \eqref{vecdef}.
\end{remark}
\begin{example}\label{ex29} For the triple Jordan system \eqref{jorvec2} the operator $P_{\bf u}$ is degenerate for any ${\bf u}$ and the formula \eqref{phiu} does not work. The general solution of \eqref{vecdef} is given by
$$
\phi({\bf u})=\frac{\bf c}{2 \langle \bf c,\, \bf u\rangle },
$$
where $\bf c$ is arbitrary constant vector. This formula is a special case of 
$$
\phi(U)=\frac{1}{2}\,C\,(U^t\, C)^{-1},
$$
which gives a solution in the case of \eqref{nm}. 
\end{example}
\begin{remark}
We see that in Example \ref{ex29} there are many homogeneous solutions of \eqref{vecdef}.  
\end{remark}
\begin{op} Find $\phi(U)$ for the triple Jordan system of $m\times m$ skew-symmetric matrices for odd $m$.
\end{op}

\subsection{Several classes of integrable rational Jordan  models}

In all ``rational'' integrable models 
described below  $\phi(U)$ denotes an arbitrary solution of the system \eqref{vecdef}.

Let us introduce the following notation
$$ \alpha_{U}(X,Y)=\{X,\,\phi(U),\,Y\} $$ and
$$ \sigma_{U}(X,Y,Z)=\{X,\,\{\phi(U),Y,\phi(U)\},\,Z\}. $$
\begin{remark} According to Propositions \ref{pr41} and \ref{pr43} for any fixed $U$ the operation $\alpha_{U}(X,Y)$ is a multiplication in a Jordan algebra while  $\sigma_{U}(X,Y,Z)$ defines a triple Jordan system. The coefficients $u_1,\dots, u_N$ of the element $U$ can be regarded as deformation parameters of these Jordan structures.
\end{remark}

\subsubsection{Class 1} For any triple Jordan system consider the hyperbolic equation 
\begin{equation} U_{xy}=\alpha_U(U_x, U_y). \label{chir} \end{equation}
In the matrix case, \eqref{chir} coincides with the equation of the
principal chiral field \eqref{prchir}.
For this reason we will call \eqref{chir} the Jordan chiral field equation.

It is easy to verify that \eqref{chir} admits the following Lax
re\-pre\-sen\-ta\-tion
$$ \Psi_x=\frac{2}{(1-\lambda)} L_{U_x}\Psi, \qquad
\Psi_y=\frac{2}{(1+\lambda)} L_{U_y}\Psi.$$
As usual, we denote by $L_X$ the operator of left multiplication. 
Note that this formula gives us a Lax representation for the matrix $\sigma$-model 
\eqref{prchir}:
$$\Psi_x=\frac{1}{(1-\lambda)}M\Psi,
\qquad\Psi_y=\frac{1}{(1+\lambda)}N\Psi,$$
where $\Psi$ is a matrix and
$$ M\Psi=-{\bf U}_x {\bf U}^{-1}\Psi-\Psi {\bf U}^{-1} {\bf U}_x, \qquad N\Psi=-{\bf U}_y {\bf U}^{-1}\Psi-
\Psi {\bf U}^{-1} {\bf U}_y,$$
which is different from the standard one.

\subsubsection{Class 2} 
The following evolution equation \begin{equation} \label{geoex}
U_t = U_{xxx} - 3
\alpha_{U} (U_x, U_{xx}) +
             \frac32 \sigma_{U}(U_x, U_x, U_x)
\end{equation} 
has infinitely many higher symmetries for any triple Jordan systems. The matrix and two vector equations, which correspond to the triple systems \eqref{jormat1}, \eqref{jorvec1} and \eqref{jorvec2}, have
the following form: $$ {\bf U}_t = {\bf U}_{xxx} -
\frac{3}{2}\,{\bf U}_x {\bf U}^{-1} {\bf U}_{xx}-\frac{3}{2}\, {\bf U}_{xx} {\bf U}^{-1} {\bf U}_x
+\frac{3}{2}\, {\bf U}_{x} {\bf U}^{-1} {\bf U}_{x} {\bf U}^{-1} {\bf U}_{x},
$$
$$
  {\bf u}_t =  {\bf u}_{xxx} - 3 \frac{\langle  {\bf u}, {\bf u}_x\rangle}{ \vert  {\bf u} \vert^2}  {\bf u}_{xx} -
  3\frac{\langle  {\bf u}, {\bf u}_{xx}\rangle}{\vert  {\bf u} \vert^2}  {\bf u}_x
         +3\frac{\langle  {\bf u}_x, {\bf u}_{xx}\rangle }{ \vert  {\bf u} \vert^2}  {\bf u} 
       -\frac{3}{2}\frac{\vert  {\bf u}_x \vert^2}{ \vert  {\bf u} \vert^2} {\bf u}_x + 
      6\frac{\langle  {\bf u}, {\bf u}_x\rangle ^2}{ \vert  {\bf u} \vert^4} {\bf u}_x
          -3\frac{\langle  {\bf u}, {\bf u}_x\rangle \vert  {\bf u}_x \vert^2 }{\vert  {\bf u} \vert^4}  {\bf u},
$$
and
$$
 {\bf u}_t =  {\bf u}_{xxx} - \frac{3}{2} \frac{\langle {\bf c}, {\bf u}_x\rangle }{\langle {\bf c}, {\bf u}\rangle }  {\bf u}_{xx} -
                    \frac{3}{2} \frac{\langle {\bf c}, {\bf u}_{xx}\rangle }{\langle {\bf c}, {\bf u}\rangle }  {\bf u}_x 
                  + \frac{3}{2} \frac{\langle {\bf c}, {\bf u}_x\rangle ^2}{\langle {\bf c}, {\bf u}\rangle ^2}  {\bf u}_x.$$

\subsubsection{Class 3} 
The following equations 
$$ V_t= V_{xxx}
- \frac32 \alpha_{V_x}(V_{xx},V_{xx})
             $$
of the Schwartz-KdV type have infinitely many symmetries for any triple Jordan system. They are related to the equations of Class 2 by the potentiation $U=V_x$. The
matrix equation is given by \eqref{matskdv}. 
The two 
vector Shwartz-KdV equations have the form
$$
{\bf u}_t = {\bf u}_{xxx} -  3\, \frac{\langle {\bf u}_x,{\bf u}_{xx}\rangle }{\vert {\bf u}_x \vert^2} {\bf u}_{xx} +
\frac{3}{2}\, \frac{\vert {\bf u}_{xx} \vert^2}{\vert {\bf u}_x \vert^2} {\bf u}_x,$$ and
$$
{\bf u}_t = {\bf u}_{xxx} -  \frac{3}{2} \frac{\langle {\bf c},{\bf u}_{xx}\rangle }{\langle {\bf c},{\bf u}_x\rangle } {\bf u}_{xx}.$$

\subsubsection{Class 4} The scalar representative of this class is the
Heisenberg model
$$u_t=u_{xx}-\frac{2}{u+v}u_x^2, \qquad v_t=-v_{xx}+\frac{2}{u+v}v_x^2.$$
The following coupled system 
$$
U_t = U_{xx} -2
\alpha_{U+V}(U_x, U_{x}), \qquad V_t = -V_{xx} +2
\alpha_{U+V}(V_x, V_{x}) 
$$
is a Jordan generalization of equation \eqref{matheis}. It
has the following third order symmetry:
$$
 U_t =
U_{xxx} -6 \alpha_{U+V}(U_x, U_{xx}) +
             6 \sigma_{U+V} (U_x, U_x, U_x),\quad 
V_t = V_{xxx} -6 \alpha_{U+V}(V_x,V_{xx}) +
             6 \sigma_{U+V} (V_x, V_x, V_x).
$$ 
One of the two vector equations is given by 
$$
\begin{cases}\ds {\bf u}_t = {\bf u}_{xx} -4\, \frac{\langle {\bf u}_x, {\bf u}+{\bf v}\rangle }{\vert {\bf u}+{\bf v}\vert ^2} {\bf u}_{x}+
2\, \frac{\vert {\bf u}_x \vert^2}{\vert {\bf u}+{\bf v}\vert ^2} ({\bf u}+{\bf v}), \\[5mm]
\ds {\bf v}_t = -{\bf v}_{xx}+4\, \frac{\langle {\bf v}_x, {\bf u}+{\bf v}\rangle }{\vert {\bf u}+{\bf v}\vert ^2} {\bf v}_{x}-
2\, \frac{\vert {\bf v}_x \vert^2}{\vert {\bf u}+{\bf v}\vert ^2} ({\bf u}+{\bf v}). \end{cases}$$

\begin{op}
Find Lax representations in the Tits--Kantor--Koecher superstructural algebra {\rm \cite{meyb}} for equations from Classes 2-4. 
\end{op}

\section{Deformations of non-associative algebras and integrable systems of geometric type} 
\subsection{Geometric description of deformations} 

Let $E$ be the Euclidean connection on an $N$-dimensional manifold ${\cal M}$, ${\bf u}=
(u_l,\cdots, u_N)$ be the local coordinates on ${\cal M}$. Denote by $E^i_{jk}({\bf u})$ the components
of $E$. Let us consider a connection $\Gamma$ with the components $\Gamma^i_{jk}({\bf u})=E^i_{jk}({\bf u})+C^i_{jk}({\bf u}),$
where $C^i_{jk}({\bf u})$ are components of a tensor field $C$ on ${\cal M}$.
\begin{definition} \cite{ss} The connection $\Gamma$ is called a {\it covariantly constant deformation
of the Euclidean connection} if the deformation tensor $C$ is covariantly constant
with respect to $\Gamma$.
\end{definition}
It follows from the standard formulas for recalculating the curvature and torsion
under the deformation of a connection (see, for example, \cite{helg}) that both the
curvature tensor of $\Gamma$ and the torsion can be expressed in terms of the deformation
tensor $C$ only:
\begin{equation}\label{tor}
T^i_{jk}=C^i_{jk}-C^i_{kj},
\end{equation}
\begin{equation}\label{curv}
R^i_{mjk}= \sum_{r} C^i_{rm} C^r_{jk}-C^i_{rm} C^r_{kj}+ C^i_{kr} C^r_{jm} - C^i_{jr} C^r_{km}. 
\end{equation}
\begin{remark} Since the tensor $C$ is covariantly constant, then ${\cal M}$ is  a space of covariantly
constant curvature and torsion.
\end{remark}

Rewriting in terms of the Euclidean connection $E$  the fact that $C$ is a covariantly
constant tensor, we obtain
\begin{equation}\label{nabla}
 \nabla_m\Big(C^i_{jk} \Big) = \sum_{r}  C^i_{rk} C^r_{mj}+ C^i_{jr} C^r_{mk}- C^i_{mr} C^r_{jk}.
\end{equation}
Here, we denote by $\nabla_m$ the covariant $u^m$-derivative with respect to $E$.
The relations \eqref{nabla} will be regarded as an overdetermined system of first order
PDE's with respect to unknown functions $C^i_{jk}({\bf u})$. We intend to investigate
the compatibility conditions for \eqref{nabla}. Of course, they are independent of the
choice of coordinate system. For calculation, it is natural to use the coordinate
system in which all components of $E$ are identically zero. With  this preferred local
coordinates, \eqref{nabla} takes the form 
of overdetermined system of PDEs for functions $C^i_{jk}(u^1,\dots, u^N)$:
\begin{equation} \label{defalg}
\displaystyle\frac{\partial C^i_{jk}}{\partial u^m}+C^i_{rk} C^r_{mj}+C^i_{jr} C^r_{mk}-C^i_{mr} C^r_{jk}=0.
\end{equation}

\subsection{Algebraic description of deformations}

Let ${\cal V}$ be a vector space with a basis ${\bf e}_1, \dots , {\bf e}_N$. The tensor $C({\bf u})$ gives rise to
the $N$-parameter family of multiplications on ${\cal V}$. Namely, the products of basis
vectors are defined by the formula
\begin{equation}\label{mustdef}
{\bf e}_j\circ {\bf e}_k = \sum_i C^i_{jk}({\bf u})\,{\bf e}_i.
\end{equation}

In terms of the product \eqref{mustdef} the deformation equation \eqref{defalg} takes the form
$$
\partial_X(Y \circ Z) = (X \circ Y)  \circ Z + Y  \circ (X  \circ Z) - X  \circ (Y  \circ Z).  
$$
Here and below, for any $X = \sum x^i {\bf e}_i$ we denote by $\partial_X$ the vector field $\ds \sum x^i \frac{\partial}{\partial u^i}$.
Note that \eqref{tor} and \eqref{curv} can be rewritten in the following compact form
$$
T(X, Y) = X \circ Y - Y \circ X,
$$
$$
R(X, Y, Z) = [Y, Z, X], \qquad X,Y,Z \in {\cal V},
$$
where $[\cdot,\cdot,\cdot]$ is defined by \eqref{as} and \eqref{br}.
The notation $T(X, Y)$ means the value of a tensor $T$ on vectors
$X$ and $Y$.

\begin{theorem} The system \eqref{defalg} is compatible iff for any ${\bf u}=(u^1,\dots, u^N)$ the product \eqref{mustdef} satisfies the following identity  {\rm \cite{ss, gss}}:
\begin{equation}\label{SS}
[V, X, Y \circ Z] - [V, X, Y] \circ Z - Y \circ [V, X, Z] = 0.
\end{equation}
\end{theorem}
\begin{theorem}{\rm \cite{ss}} \quad {\bf a)} The class of algebras with identity \eqref{SS} contains:
\begin{itemize}
\item[1)] Associative algebras;
\item[2)] Left-symmetric algebras;
\item[3)] Lie algebras;
\item[4)] Jordan algebras;
\item[5)] LT-algebras. Any commutative algebra with identity  \eqref{SS} is an LT-algebra.
\end{itemize}

{\bf b)} Let $C^i_{jk}({\bf u})$ be the solution of system \eqref{defalg} with an initial data $C^i_{jk}(0)$. If $C^i_{jk}(0)$ are structural constants of an algebra from  one of the classes 1--5,  then the algebra with multiplication \eqref{mustdef} belongs to the same class for any ${\bf u}$. In other words, all these classes of algebras are invariant with respect to the deformation \eqref{defalg}. 
\end{theorem}
\begin{op} Suppose that $C^i_{jk}({\bf u})$ satisfies \eqref{defalg}. Prove that for small ${\bf u}$ the algebra with the stuctural constants $C^i_{jk}({\bf u})$ is isomorphic to the algebra with the structural constants $C^i_{jk}(0)$.

\end{op}
\begin{proposition} Let $\{X,Y,Z\}$ be a triple Jordan system and $\phi({\bf u})$ be a solution of \eqref{vecdef}. Then the structural constants of the product
$$
X\circ Y=\{X,\, \phi,\, Y\}
$$
satisfy the system \eqref{defalg}.
\end{proposition}

\subsection{Equations of geometric type}

Consider evolution systems of the form
\begin{equation}
u^{i}_{t}=u^i_{xxx}+\alpha^{i}_{jk}({\bf u})\,u^{j}_{x}\,u^{k}_{xx}+
\gamma^{i}_{jks}({\bf u})\,u^{j}_{x}\, u^{k}_{x}u^{s}_{x}, \qquad i=1,\dots,N.
\label{geom} 
\end{equation}
For such systems unknown coefficients are not constants but functions of $u^1,\dots,u^N$. 
Here and below, we assume that the {\it summation is carried out over repeated indices}. 

It is convenient to
rewrite \eqref{geom} in the following way
\begin{equation}
u^{i}_{t}=u^{i}_{xxx}+3 \alpha^{i}_{jk}\, u^{j}_{x}u^{k}_{xx}+\Big(
\frac{\partial \alpha^{i}_{km}}{\partial u^j}+2\alpha^{i}_{jr}\alpha^{r}_{km}-
\alpha^{i}_{rj}\alpha^{r}_{km}+
\beta^{i}_{jkm}\Big)\,u^{j}_{x}u^{k}_{x}u^{m}_{x}.\label{geeq}
\end{equation}
The class of systems \eqref{geeq} is invariant under the arbitrary
point transformations
${\bf u} \rightarrow  {\bf \Phi}({\bf u})$, where ${\bf u}=(u^1,\dots, u^N)$.
It is  easy to see that under such a change of coordinates,
$\alpha^{i}_{jk}$ and $\beta^{i}_{jkm}$ are transformed just as
components of an affine connection $\Gamma$ and a tensor $\beta$, respectively. 

\begin{example} In the case $N=1$ the equation \eqref{geeq} has the form
$$
u_t=u_{xxx} + 3 \alpha(u)\, u_x u_{xx} + \Big(\alpha'(u)+\alpha(u)^2+\beta(u)\Big)\, u_x^3. 
$$
Using the symmetry approach (see Section 2.2), one can verify that this equation possesses higher symmetries iff $\beta'=2 \alpha \beta.$ By a proper point transformation of the form $ u \rightarrow \Phi(u)$ the function $\alpha$ can be reduced to zero (for $N=1$ any affine connection is flat) and the function $\beta$ becomes a constant. The equation 
$u_t=u_{xxx}+ {\rm const} \,u_x^3$ is known to be integrable. It is related to the mKdV equation by a potentiation.
\end{example}

Without loss of generality we assume that the tensor $\beta$ is symmetric:
$$
\beta(X,Y,Z)=\beta(Y,X,Z)=\beta(X,Z,Y)
$$
for any vectors  $X,Y,Z$.

\begin{question} For which affine connections $\Gamma$ and which tensors $\beta$ is the equation \eqref{geeq} integrable?
\end{question}

Let $R$ and $T$ be the curvature and the torsion tensors of $\Gamma$. In order to formulate the classification results, we introduce the following
tensor:
\begin{equation}\label{bsd1}
\sigma(X,Y,Z)=\beta(X,Y,Z)-\frac{1}{3} \delta(X,Y,Z)+
\frac{1}{3} \delta(Z,X,Y),
\end{equation}
\begin{equation}\label{bsd2}
\delta(X,Y,Z)=T(X,T(Y,Z))+R(X,Y,Z)-\nabla_X(T(Y,Z)). 
\end{equation}
Using Bianchi's identity $R(X,Y,Z)+R(Y,Z,X)+R(Z,X,Y)=0$, one can check that
\begin{equation} \sigma(X,Y,Z)=\sigma(Z,Y,X). \label{jord1}
\end{equation}
\begin{remark} Taking into account \eqref{jord1}, the symmetricity of $\beta$ and the identity $\delta(X,Y,Z)=-\delta(X,Z,Y),$ one can check that 
\begin{equation}\label{bsd3}
\beta(X,Y,Z)=\frac{1}{3} \Big(\sigma(X,Y,Z)+\sigma(Y,Z,X)+\sigma(Z,X,Y)\Big).
\end{equation}

\end{remark}

\begin{theorem}\footnote{This theorem was proved by S.~Svinilupov and V.~Sokolov (unpublished) and was formulated in the survey \cite{habsokyam} dedicated to the memory of Sergey Svinolupov.}  The equation \eqref{geeq} possesses a higher symmetry of the
form
$$ {\bf u}_{\tau}={\bf u}_{n}+ {\bf G}({\bf u}, {\bf u}_x,\cdots , {\bf u}_{n-1}),
\qquad n>3,$$
iff the following identities for the tensors $T, R$ and $\sigma$ hold:
\begin{equation} 
\nabla_X(R(Y,Z,V))=R(Y,X,T(Z,V)),\label{Cond1}
\end{equation}
\begin{equation} 
\nabla_X\Big(\nabla_Y\big(T(Z,V)\big)-T\big(Y,T(Z,V)\big)-R(Y,Z,V)\Big)=0,\label{Cond2}
\end{equation}
\begin{equation} 
\nabla_X\Big(\sigma(Y,Z,V)\Big)=0, \label{Cond3}
\end{equation}
\begin{equation}  T\big(X,\sigma(Y,Z,V)\big)+T\big(Z,\sigma(Y,X,V)\big)    
+ T\big(Y,\sigma(X,V,Z)\big)+T\big(V,\sigma(X,Y,Z\big)=0, \label{Cond4}
 \end{equation}
\begin{equation}\begin{array}{l}  \sigma\big(X,\sigma(Y,Z,V),W\big)-\sigma\big(W,V,\sigma(X,Y,Z)\big)     
 + \\[1.5mm] \sigma\big(Z,Y,\sigma(X,V,W)\big)-\sigma\big(X,V,\sigma(Z,Y,W)\big)=0, \label{jord2} 
 \end{array}
 \end{equation}
 and relation \eqref{bsd1}, where the tensors $\beta$ and $\delta$ are eliminated by means of \eqref{bsd2} and \eqref{bsd3}.
  \end{theorem}
 \begin{remark} The identities \eqref{jord1} and \eqref{jord2} mean that $\sigma^{i}_{jkm}({\bf u})$
are the structural constants of a triple Jordan system for any ${\bf u}$.
\end{remark}

\begin{remark} \label{rem53} In the case $T=0$ only the following conditions remain:
\begin{equation}\label{Tcon1}
  \nabla_X\big(\sigma(Y,Z,V)\big)=0,
\end{equation}
\begin{equation}\label{fou}
R(X, Y, Z) = \sigma(X,Z,Y) - \sigma(X ,Y, Z),
\end{equation}
\eqref{jord1} and \eqref{jord2}. 
This means that we are dealing with a symmetric space and a special covariantly constant deformation of a triple Jordan system. 
\end{remark}
  
It can be verified that for equations \eqref{geoex} we have $T=0$ and the integrability conditions of Remark \ref{rem53} are satisfied. More generally, 
any equation of the form
\begin{equation} \label{class1}  U_t = U_{xxx} - 3 \,A_U(U_x, U_{xx}) +
             \frac32\, B_U(U_x, U_x, U_x),
\end{equation}
where the product $-A_U(\cdot, \cdot)$ is defined as the deformation \eqref{defalg} of any Jordan algebra and $B_U$ is the corresponding triple system \eqref{trjor}, satisfy the conditions of Remark \ref{rem53}.

The following two equations  
$$
{\bf u}_t={\bf u}_{xxx}+\frac32\, \big(P({\bf u},{\bf u}_x)
({\bf c}-\vert {\bf c}\vert ^2 {\bf u})\big)_x+3\, \mu\, \vert {\bf c}\vert ^2
P({\bf u},{\bf u}_x)\, {\bf u}_x,
$$
where $\ds \mu=-\frac{1}{2}$ or $\mu=0$,  ${\bf u}$ is an $N$-dimensional vector, $\bf c$ is a given constant vector, and
$$P({\bf u},{\bf u}_x)=\left|{\bf u}_x+\frac{\langle{\bf c},{\bf u}_x\rangle}{1-\langle{\bf c},{\bf u}\rangle}\, {\bf u}\right|^2,
$$
satisfy the integrability conditions but do not belong to the class \eqref{class1}.

\begin{theorem}
 For any triple Jordan system $\{\cdot,\cdot,\cdot\}$
with the structural
constants $s^{i}_{jkm}$, there exists a unique (up to point transformations)
solution of equations \eqref{jord1}, \eqref{jord2}, \eqref{Tcon1} and  \eqref{fou}  such that 
\begin{equation} \sigma^{i}_{jkm}(0)=s^{i}_{jkm}. \label{data1} \end{equation}
\end{theorem}

In the case $T \ne 0$, a generalization of the symmetric
spaces arises. We do not know if such affine connected spaces have been considered by geometers.

\chapter{Vector integrable evolution equations}

In Sections \ref{JordKdv}--\ref{sec54}  we have seen that many examples of integrable models related to non-associative algebras are polynomial either matrix or vector equations. A first attempt for classification of such equations was made in \cite{sokwolf}.  

\section{Integrable polynomial homogeneous vector systems}
In this section we consider several types of vector polynomial $\lambda$-homegeneous systems (see Theorem \ref{schomo}). 
By analogy with the scalar case  $\lambda$ is supposed to be 2,1, or $\ds \frac{1}{2}$. Vector NLS type systems \eqref{vecNLS} and \eqref{ex3} give us examples of such integrable systems. 
Here, we present several interesting vector systems from  \cite{sokwolf}, which have higher symmetries.  
\begin{example} The vector analogue of the Ibragimov-Shabat equation \eqref{IbShab} is given by 
$$
\label{ibsh}
{\bf u}_{t}={\bf u}_{xxx}+3 \langle {\bf u}, {\bf u}\rangle \,{\bf u}_{xx}+6 \langle {\bf u},{\bf u}_x\rangle\, {\bf u}_x+             3 \langle {\bf u},{\bf u}\rangle^2 \, {\bf u}_x+3 \langle {\bf u}_x, {\bf u}_x\rangle \, {\bf u}. 
$$
\end{example}
\begin{example} There are two following systems  of the derivative NLS type with two vectors ${\bf u}$ and {\bf v}:
$$
\begin{cases}
{\bf u}_t = \;\;\,{\bf u}_{xx} + 2 \alpha \langle {\bf u}, {\bf v}\rangle \,{\bf u}_x
                   + 2 \alpha \langle {\bf u},  {\bf v}_x\rangle \,{\bf u}
                   - \alpha \beta \langle {\bf u},  {\bf v}\rangle^2 \,{\bf u}, \\[2mm]
 {\bf v}_t =   -    {\bf v}_{xx} + 2 \beta \langle {\bf u},  {\bf v}\rangle  \,{\bf v}_x
                   + 2 \beta \langle  {\bf v}, {\bf u}_x\rangle  \,{\bf v}
                   + \alpha \beta \langle {\bf u}, {\bf v}\rangle^2  \,{\bf v}  
\end{cases}
$$
and \cite{japan1}
$$
\begin{cases}
{\bf u}_t = \;\;\,{\bf u}_{xx} + 2 \alpha \langle {\bf u}, {\bf v}\rangle {\bf u}_x
                   + 2 \beta \langle {\bf u}, {\bf v}_x\rangle {\bf u}
                   + \beta (\alpha -2 \beta) \langle {\bf u},{\bf v}\rangle^2 {\bf u},\\[2mm]
{\bf v}_t =   -   {\bf v}_{xx} + 2 \alpha \langle {\bf u}, {\bf v}\rangle {\bf v}_x
                   + 2 \beta \langle {\bf v}, {\bf u}_x\rangle {\bf v}
                   - \beta (\alpha -2 \beta) \langle {\bf u},{\bf v}\rangle^2 {\bf v} . 
\end{cases}
$$
Here, $\alpha$ and $\beta$ are arbitrary constants. One of them can be normalized by a scaling.  In the above examples $\ds \lambda=\frac{1}{2}$.
 \end{example}
\begin{example} Four systems with a vector ${\bf u}$ and a scalar function $u,$ where $\lambda=2,$ were found in  \cite{sokwolf}:
\begin{equation}\label{svs}
\begin{cases}
 u_t = u_{xxx} + u u_x - \langle {\bf u}, {\bf u}_x\rangle, \\[2mm]
{\bf u}_t = {\bf u}_{xxx} + u {\bf u}_x + u_x {\bf u}, 
\end{cases}
\end{equation}
\begin{equation}
\label{eeq2V}
\begin{cases}
u_t = u_{xxx} + 3 u u_x + 3 \langle {\bf u},  {\bf u}_x\rangle, \\[2mm]
 {\bf u}_t = u \,  {\bf u}_x + u_x \,  {\bf u},
\end{cases}
\end{equation}
\begin{equation}
\label{eeq3V}
\begin{cases}
u_t = \langle {\bf u}, \,  {\bf u}_x\rangle, \\[2mm]
{\bf u}_t = {\bf u}_{xxx} + 2 u \, {\bf u}_x + u_x \, {\bf u},
\end{cases}
\end{equation}
\begin{equation}
\label{eeq4V}
\begin{cases}
u_t = u_{xxx} + u u_x + \langle {\bf u}, {\bf u}_x\rangle, \\[2mm]
{\bf u}_t = -2\,{\bf u}_{xxx} - u \, {\bf u}_x.
\end{cases}
\end{equation}
System \eqref{svs} is just \eqref{VectorKdv}, where ${\bf c}=(1,0,\dots,0)$. System \eqref{eeq2V} has been considered in 
\cite{kuper1}.  It is a vector generalization of the Ito
equation.
Systems (\ref{eeq3V}) and (\ref{eeq4V}) are vector generalizations of the corresponding
scalar systems from the paper \cite{drsok}. 
\begin{op}
Find a Lax representation for systems \eqref{svs}-\eqref{eeq4V} in the frames of general approach of {\rm \cite{kac2}}.
\end{op}

\begin{remark}
The fifth order symmetry
\begin{equation}
\label{kuper}
\begin{cases}
\begin{array}{rcl}
u_\tau&=&u_{xxxxx} + 10 u u_{xxx}+25 u_x u_{xx}+20 u^2 u_x-\\[2mm]
   & &10 \langle {\bf u}, {\bf u}_{xxx}\rangle- 15 \langle {\bf u}_x, {\bf u}_{xx}\rangle
      -10 u_x \langle {\bf u},{\bf u}\rangle-20 u \langle {\bf u}, {\bf u}_x\rangle,\\[4mm]
{\bf u}_\tau&=&-9\,{\bf u}_{xxxxx} - 30 u \, {\bf u}_{xxx}-45 u_x \, {\bf u}_{xx}
      -(35 u_{xx}+20 u^2+5 \langle {\bf u},{\bf u}\rangle)\,{\bf u}_x\\[2mm]
   & &-(10 u_{xxx}+20 u u_x+5 \langle {\bf u},{\bf u}_x\rangle)\,{\bf u}
\end{array}
\end{cases}
\end{equation}
of system \eqref{eeq3V} is nothing but a vector
generalization of the Kaup-Kuperschmidt equation. Indeed,
if the vector part is absent {\rm (}i.e., ${\bf u}=0${\rm )}, then system
\eqref{eeq3V} becomes trivial and \eqref{kuper} turns out
to be the Kaup-Kuperschmidt equation \eqref{kk}.
\end{remark}
\end{example}

A big list of integrable systems with one scalar and one vector unknown functions and $\lambda=1$ can be found in \cite{TsWolf}. 
 
\section{Symmetry approach to integrable vector equations}

It turns out that for vector models the symmetry approach can be developed \cite{meshsokCMP} in the same universality as for the scalar equations (see Section 2.2). 
In particular, we do not assume anymore that the right-hand side of the equation is polynomial. 

\begin{example} The following
vector Harry Dym equation
\begin{equation}\label{HD}
{\bf u}_t = \langle {\bf u}, {\bf u}\rangle ^{3/2} \, {\bf u}_{xxx},
\end{equation}
where $\bf u$ is an $N$-component vector, has infinitely many symmetries and conservation laws for any $N$.
\end{example}

Equations \eqref{vecmkdv1}, \eqref{vecmkdv2} and \eqref{HD} belong to the following
class of vector equations:
\begin{equation}
{\bf u}_t=f_n \, {\bf u}_n+f_{n-1}\, {\bf u}_{n-1}+\cdots+f_1\,{\bf u}_1+f_0\,{\bf u}, \qquad 
{\bf u}_{i}=\frac{\partial^i  {\bf u}}{\partial x^i}.
\label{gensys}
\end{equation}
Here, the (scalar) coefficients $f_i$ depend on 
scalar products between vectors $ {\bf u},{\bf u}_{x},...,{\bf u}_{n-1}$.   The assumption
that the vector space is finite-dimensional is inessential for us.
For instance, the vectors could be functions of $t, x$, and $y$ and the scalar product 
$$
\langle U, \,V\rangle  = \int_{-\infty}^{\infty} U(x,t,y)\, V(x,t,y) \,dy.
$$
In this case we arrive at a special class of $(2+1)$-dimensional non-local evolution equations.

The crucial point is that we regard the variables 
\begin{equation}\label{izvar}
u[i,j]=\langle {\bf u}_i,\, {\bf u}_j\rangle , \qquad i,j=0,1,\dots \qquad j\ge i
\end{equation}
as independent. We denote by $\cal{F}$ a field of functions depending on   variables \eqref{izvar}. 

It is clear that in the case of the standard scalar product in $N$-dimensional vector space all such equations \eqref{gensys} are {\it isotropic}, i.e. they are invariant with respect to the group $SO(N)$.

\subsubsection{Examples of integrable  non-isotropic vector equations}

\begin{example}\label{ex52}  The integrable vector KdV equation \eqref{VectorKdv} (see \cite{SviSok94} and Section \ref{JordKdv}).
\end{example}
\begin{example}\label{ex53}
Consider the equation \cite{golsok}:
\begin{equation}
{\bf u}_t=\Big({\bf u}_{xx}+\frac{3}{2} \langle {\bf u}_x, \, {\bf u}_x \rangle \, {\bf u}\, \Big)_x +\frac{3}{2} \langle {\bf u}, \,R({\bf u})\rangle \, {\bf u}_x, \qquad \vert {\bf u} \vert = 1,
\label{LL}
\end{equation}
where $R={\rm diag} (r_1,...,r_N)$ is a constant matrix.  
If $N=3$, then \eqref{LL} is a commuting flow of the Landau-Lifshitz equation.  
\end{example} 

In the case of Example \ref{ex52}
we have to take for ${\cal F}$ functions of a finite number of independent variables
\begin{equation}\label{cizvar}
\langle {\bf u}_i,\, {\bf u}_j\rangle , \quad j\ge i,  \qquad \langle {\bf c},\, {\bf c}\rangle, \qquad \langle {\bf u}_i,\, {\bf c}\rangle.
\end{equation}
\begin{op} Solve several simple classification problems for such type equations.
\end{op}

For equations like \eqref{LL} we may consider $(X,\,Y)=\langle X,\, R(Y)\rangle $ as a new scalar product and take for ${\cal F}$ functions of independent variables
\begin{equation}\label{anizvar}
u[i,j]=\langle {\bf u}_i,\, {\bf u}_j\rangle , \qquad v[i,j]=({\bf u}_i,\, {\bf u}_j), \qquad i,j=0,1,\dots \qquad j\ge i. 
\end{equation}
\begin{remark}
It is clear that linear transformations of the scalar products
$$
\bar u[i,j] = c_1 u[i,j]+c_2 v[i,j], \qquad \bar v[i,j] = c_3 u[i,j]+c_4 v[i,j],\qquad c_1 c_4 \ne c_2 c_3
$$
 preserve the class of all integrable equations of such type.
\end{remark}
The symmetry approach we are developing in the next sections can be easily extended to equations \eqref{gensys} with coefficients that depend on variables \eqref{anizvar}.

\subsection{Canonical densities}
Consider the equations \eqref{gensys} with coefficients from ${\cal F}$, where ${\cal F}$ is a field of functions, which depend on variables \eqref{izvar}, \eqref{cizvar} or \eqref{anizvar}.

\begin{theorem}\label{th51} {\rm (\cite{meshsokCMP})}. 
If the equation \eqref{gensys} possesses an infinite series of
vector commuting flows of the form
\begin{equation}  \label{symvec}
{\bf u}_{\tau}=g_m \, {\bf u}_m+g_{m-1}\, {\bf u}_{m-1}+\cdots+g_1\,{\bf u}_1+g_0\,{\bf u}, \qquad g_i\in {\cal F},  \end{equation}
then \begin{itemize}
\item[i).] there exists a formal Lax pair $L_t=[A, \, L]$, where 
\begin{equation}\label{forLax}
L=a_1 \, D+a_0 \, + a_{-1}\, D^{-1} +\cdots   \,, \qquad  A=\sum_0^n  f_i\, D^i.
\end{equation}
Here, $f_i$ are the coefficients of equation \eqref{gensys} and $a_i\in {\cal F}$.

\item[ii).]The following functions
\begin{equation} \label{kanon}
\rho_{-1}=\frac{1}{a_1},  \qquad  \rho_0=\frac{a_0}{a_1},
\qquad  \rho_i={\rm res} \, L^i,
\qquad i \in \N     
\end{equation}
are conserved densities for the equation \eqref{gensys}.
\end{itemize}

The conservation laws with densities \eqref{kanon} are called  {\it canonical}. 
\end{theorem}

\begin{proof} Idea of the proof. 
i). Let us rewrite the equation \eqref{gensys} and its commuting flow  
\eqref{symvec} in the form 
\begin{equation}
{\bf u}_t=A({\bf u}), \qquad {\bf u}_{\tau}=B({\bf u}), \qquad \mbox{where} \quad B=\sum_0^m
g_i\, D^i. \label{compat}
\end{equation}
The compatibility of  \eqref{compat} leads to the operator identity  
$$
B_t-[A,\, B]=A_{\tau}.
$$
For large  $m$ we may ``ignore'' the right-hand side, since it has a small order comparing with the other terms. In other words, the operator $L=B$ satisfies $L_t=[A, \, L]$ approximately. Then the series of first order $L_m=B^{1/m}$ is an approximate solution as well. The gluing of the first order approximate solutions corresponding to different commuting flows into an exact formal Lax operator $L$ is similar to the scalar case \cite{sokshab}.

ii). It follows from  Adler's theorem \ref{adler}.
\end{proof}

\begin{theorem}\label{th52} {\rm (\cite{meshsokCMP})}.
 If the equation \eqref{gensys} possesses an infinite series of
conserved densities from ${\cal F}$,
then \begin{itemize} \item[i).] there exist a formal Lax operator $L$ 
and a series $S$ of the form
$$
S=s_1 \, D+s_0 \, +s_{-1}\, D^{-1}+s_{-2}\,
D^{-2}+\cdots   \,,
$$
such that
$$
S_t+A^{+}\, S+S \, A=0 ,   \qquad S^{+}=-S,  \qquad  L^{+}=-S^{-1} L S,    
$$
where the upper index $+$ stands for the formal conjugation.

\medskip

\item[ii).] Under the conditions of Item i) the densities \eqref{kanon},
with $i=2k,$ are of the form $\rho_{2 k}=D(\sigma_k)$ for
some functions $\sigma_k.$
\end{itemize} 
\end{theorem}

\begin{proof} 
For the proof  see \cite{sokshab, meshsokCMP}.
\end{proof}

\begin{remark} In contrast with \eqref{Lambdaeq} the differential operator $A$ in \eqref{forLax} is not the Fr\'echet derivative of the right-hand side of \eqref{gensys}. For this reason we avoid the names ``formal symmetry'' or ``formal recursion operator'' for the series $L$.
\end{remark}

\subsection{Euler operator and  Fr\'echet derivative}

 To define the canonical conserved densities we have considered differential operators and pseudo-differential series with scalar coefficients from 
${\cal F}$. However, similar to the matrix case (see Section \ref{naev}), it is not sufficient for constructing of Hamiltonian and recursion operators, and we should change the ring of coefficients. 

Let ${\cal F}$ be a field of functions that depend on \eqref{izvar}. Denote by $R_{i,j}$ an ${\cal F}$-linear operator that acts on vectors  by the rule   
$$R_{i,j}(\bf v)= \langle {\bf u}_j, \bf v\rangle \,{\bf u}_i.$$ 
It is easy to see that 
$$
R_{i,j} R_{p,q}=\langle {\bf u}_j, {\bf u}_p\rangle  R_{i,q}, \qquad R_{i,j}^t=R_{j,i}, \qquad {\rm tr}\, R_{i,j}=\langle {\bf u}_i, {\bf u}_j\rangle ,$$
$$  D\circ R_{i,j}=R_{i,j} D+R_{i+1,j}+R_{i,j+1}.
$$
Denote by $\cal O$ the algebra over $\cal F$, generated by operators $R_{i,j}$ and by the identity operator. 

The Fr\'echet derivative of an element from ${\cal F}$ is a differential operator with coefficients from $\cal O.$ Such differential operators are called {\it local}.  For instance,  the Fr\'echet derivative of the right-hand side  $F$ of an equation  \eqref{gensys} is equal to
\begin{equation} \label{Freshe}
F_{*}= \sum_{k} f_k D^k+\sum_{i,j,k} \frac{\partial f_k}{\partial u_{[i,j]}} \left(R_{k,i}\, D^j+R_{k,j}\, D^i\right), 
\end{equation}
where $\,\, i,j,k=0,\dots,n.$

The Euler operator (or the variational derivative) is given by 
$$
\frac{\delta }{\delta {\bf u}}=\sum_{i\le j} (-D)^i \circ  {\bf u}_j
\Big( \frac{\partial }{\partial u_{[i,j]}}\Big)+ (-D)^j \circ {\bf u}_i
\Big( \frac{\partial }{\partial u_{[i,j]}} \Big).
\label{euler}
$$

Possibly most of the Hamiltonian structures for vector integrable equations are non-local. For example, the
Hamiltonian operator  $\cal H$ and the symplectic operator  $\cal T$  for the vector MKdV-equation 
$$
{\bf u}_t={\bf u}_{xxx}+\langle{\bf u}, {\bf u}\rangle \, {\bf u}_x
$$
are given by
 $$
 {\cal H}({\bf w})=D({\bf w})+\langle {\bf u}, \, D^{-1}\circ {\bf u}\rangle \, {\bf w} - \langle{\bf u},\,D^{-1}\circ {\bf w}\rangle {\bf u},$$
 $$ 
 {\cal T}({\bf w})=D({\bf w})+{\bf u}\, D^{-1}\circ\langle {\bf u},\,{\bf w}\rangle. 
$$
It is easy to see that these operators, written as pseudo-differential series, have coefficients from $\cal O$.
\begin{remark} One can define a formal symmetry as a pseudo-differential series with coefficients from ${\cal O}$ that satisfies \eqref{Lambdaeq} and develop the symmetry approach related to the existence of such a formal symmetry. But the Fr\'echet derivative \eqref{Freshe} is much more complicated than the operator $A$ defined by \eqref{forLax}. 
Moreover Theorem \ref{th51} is true for non-isotropic equations also while the definition of $\cal O$ essentially depends on the choice 
of ${\cal F}$ {\rm(see Examples \ref{ex52} and \ref{ex53})}.

\end{remark}

\subsection{Vector isotropic equations of KdV-type}

Consider vector equations of the form
\begin{equation}\label{evec}
  {\bf u}_{t} =  {\bf u}_{xxx}+f_2  {\bf u}_{xx}+f_1  {\bf u}_x+f_0  {\bf u}, 
\end{equation}
where  coefficients $f_i$  of the equation are  scalar 
functions of the following six  independent variables:
\begin{equation}\label{dynvar}
\langle{\bf u},\, {\bf u}\rangle,\quad  \langle{\bf u}, {\bf u}_x\rangle,\quad  \langle {\bf u}_x, {\bf u}_x\rangle,\quad  \langle{\bf u}, {\bf u}_{xx}\rangle,\quad  \langle{\bf u}_x, {\bf u}_{xx}\rangle,\quad  \langle{\bf u}_{xx}, {\bf u}_{xx}\rangle.
\end{equation}

\begin{theorem}\label{canvec} {\rm \cite{meshsokCMP}} For equations \eqref{evec} the canonical densities are defined by the following recurrence  formula:
\begin{eqnarray}
\rho_{n+2}&=&\frac{1}{3}\biggl[\sigma_n-f_0\,\delta_{n,0} -2\,f_2\,\rho_{n+1
}-
f_2\,D\rho_{n} - f_1\,\rho_{n}\biggr] \nonumber\\[2mm]
&&-\frac{1}{3}\biggl[f_2\,\sum_{s=0}^{n} \rho_{s}\,\rho_{n-s}+\sum_{0\le s+k\le n}\rho_{s}\,\rho_{k}\,
\rho_{n-s-k}+3\sum_{s=0}^{n+1}\rho_{s}\, \rho_{n-s+1}\biggr] \nonumber \\[2mm]
&&-D\biggl[\rho_{n+1}+\frac{1}{2}\sum_{s=0}^{n}\rho_{s}\,\rho_{n-s}+\frac{1}{3} D \rho_{n}\biggr],\quad n\ge 0, \label{rekkur}
\end{eqnarray}
where $\delta_{i,j}$ is the Kronecker delta and $\rho_0,\ \rho_1$ are defined by  
$$
\rho_0 = -\frac{1}{3}\,f_2,
$$
$$ \label{ro1}
\rho_1 =\frac{1}{9}\,f_2^2-\frac{1}{3}\,f_1+\frac{1}{3}\,D(f_2).
$$
\end{theorem}

\begin{remark}
The same formulas for the canonical densities  hold for the non-isotropic case. The most general recurrence  formula for equations \eqref{gensys} with $n=3$ is presented in Section \ref{sub627}.

\end{remark}
\begin{op} Perform a complete classification of integrable equations \eqref{evec} in isotropic or/and non-isotropic cases.
\end{op}

Some particular results were obtained in \cite{M-Sok2, Bal2}. Here, we formulate one of them. 
\subsubsection{Shift-invariant equations}

Consider equations of the form
\begin{equation}\label{poteq}
  {\bf u}_{t} =  {\bf u}_{xxx}+f_2  {\bf u}_{xx}+f_1  {\bf u}_x,
\end{equation}
where  $f_i$ depend only on
$\,\, \langle {\bf u}_x, {\bf u}_x\rangle ,\,$ 
$\langle {\bf u}_x, {\bf u}_{xx}\rangle ,\,$ 
$\langle {\bf u}_{xx}, {\bf u}_{xx}\rangle \,.$
It is clear that such equations are invariant with respect to translations of the form ${\bf u} \to {\bf u}+{\bf c}$.
\begin{theorem} {\rm \cite{M-Sok2}} Any equation \eqref{poteq} with infinite sequence of higher symmetries or local conservation laws\footnote{We call a conservation law $D_t(\rho)=D(\sigma)$ local if $\rho,\sigma \in {\cal F}$.} belongs to the following list:
$$
{\bf u}_t={\bf u}_{xxx} +\frac{3}{2}\, \Big(\frac
{{a^2\,u_{[1,2]}}^2}{1+a\,u_{[1,1]}}
-a\,u_{[2,2]} \Big)\,{\bf u}_x,   $$
$${\bf u}_t={\bf u}_{xxx}-3\,{\frac {u_{[1,2]}}{u_{[1,1]}}}\,{\bf u}_{xx}
+
 \frac{3}{2}\,{\frac {u_{[2,2]}}{u_{[1,1]}}}\,{\bf u}_x, $$
$${\bf u}_t={\bf u}_{xxx}-3\,\frac {u_{[1,2]}}{u_{[1,1]}}\,{\bf u}_{xx} +
\frac{3}{2}\,\left({\frac {u_{[2,2]}}{u_{[1,1]}}}+{\frac {{u_{[1,2]}}^{2}}
{{u_{[1,1]}}^{2} ( 1+a\,u_{[1,1]}) }}\right) {\bf u}_x,$$

$$
{\bf u}_t={\bf u}_{xxx}- \frac{3}{2}\,(p+1)\,\frac
{u_{[1,2]}}{p\,u_{[1,1]}}\,{\bf u}_{xx}+
\frac{3}{2}\,(p+1)\left( {\frac {u_{[2,2]}}{u_{[1,1]}}}
-{\frac {a\,{u_{[1,2]}}^{2}}{{p}^{2}\,u_{[1,1]}}} \right){\bf u}_x.
$$
Here, $a$ is an arbitrary constant and $p=\sqrt{\mathstrut 1+a\,u_{[1,1]}}$. Notice that, if $a=0,$ 
the last equation of the list is reduced to
$$
{\bf u}_t={\bf u}_{xxx}-3\,{\frac {u_{[1,2]}}{u_{[1,1]}}}\,{\bf u}_{xx}+
{3}\,{\frac
{u_{[2,2]}}{u_{[1,1]}}}\,{\bf u}_x.
$$
\end{theorem}
\begin{remark} We verified that each of these equations possesses a symmetry of fifth order and non-trivial local conservation laws. To prove their integrability auto-B\"acklund transformations of first order with a spectral parameter were found in {\rm \cite{M-Sok2}}.  
\end{remark}

\subsection{Auto-B\"acklund transformations}

An auto-B\"acklund transformation of the first order is defined by  
$$
{\bf u}_x=h\, {\bf v}_x+f\, {\bf u}+g\, {\bf v},
$$
where ${\bf u}$  and ${\bf v}$ are solutions of the same vector equation. The functions $f,g$ and $h$ are (scalar) functions of variables  
$$
u_{[0,0]}{=} \langle{\bf u}, {\bf u}\rangle,\qquad v_{[i,j]}\stackrel{def}{=} \langle{\bf
v}_i, {\bf v}_j\rangle,\qquad
w_i\stackrel{def}{=}\langle{\bf u}, \, {\bf v}_i\rangle,\qquad  i,j\ge 0.
$$
\begin{example} The auto-B\"acklund transformation for the vector Swartz-KdV equation 

$${\bf u}_t={\bf u}_{xxx}-3\,{\frac {u_{[1,2]}}{u_{[1,1]}}}\,{\bf u}_{xx}
+
 \frac{3}{2}\,{\frac {u_{[2,2]}}{u_{[1,1]}}}\,{\bf u}_x$$
is given by
$$
{\bf u}_x=\frac {2\,\mu}{{\bf v}_x^{2}}\,\langle \bf u-\bf v,\,{\bf v}_{x}\rangle\,
({\bf u}-{\bf v})-\frac{\mu}{{\bf v}_x^{2}}\,\vert\bf u-\bf v\vert^{2}\,{\bf v}_x,
$$
where $\mu$ is an arbitrary (spectral) parameter.
The superposition formula 
$$\bf z=\bf u +(\mu- \nu)\,\frac{\nu\, (\bf u- {\bf v'})^2\,(\bf u-\bf v)
-\mu\,(\bf u-\bf v)^2\,(\bf u-{\bf v'})}{\big(\mu\,(\bf u-\bf v)-
\nu\,(\bf u-{\bf v'})\big)^2},
$$
corresponding to this auto-B\"acklund transformation connects 4 different solutions 
$$
\begin{CD}
\bf v' &@>\mu>> & \bf z \\
@A\nu AA & & @AA\nu A \\
\bf u& @>>\mu >& \bf v
\end{CD}
$$
of the vector Schwartz-KdV equation. It defines a known 
integrable vector discrete model. 
\end{example}

\subsection{Equations on the sphere}
\subsubsection{Isotropic equations}

Consider the equations \eqref{evec} on the sphere $\langle \bf u, \, \bf u\rangle =1$. In this case, the coefficients are functions of only three  independent variables
\begin{equation}\label{dynvarsph}
u_{[1,1]} = \langle{\bf u}_x, {\bf u}_x\rangle, \qquad  u_{[1,2]} = \langle {\bf u}_x, {\bf u}_{xx}\rangle,\qquad  u_{[2,2]} = \langle{\bf u}_{xx}, {\bf u}_{xx}\rangle
\end{equation}
instead of six variables \eqref{dynvar}. Differentiating the constraint $\langle \bf u, \, \bf u\rangle =1$, one can express the remaining scalar products through \eqref{dynvarsph}. Moreover, 
the relation $\langle \bf u, \,\bf u_t\rangle =0$ implies that 
$$f_0 = f_2\,u_{[1,1]} + 3\,u_{[1,2]}$$ 
and any equation \eqref{evec} on the sphere has the following form
\begin{equation}\label{eqsf}
{\bf u}_t={\bf u}_3+f_2\,{\bf u}_2+f_1\,{\bf u}_1+(f_2\,u_{[1,1]} +3\,u_{[1,2]})\, {\bf u}, \qquad \vert {\bf u}\vert = 1.
\end{equation}

\begin{theorem} {\rm \cite{meshsokCMP}} Any equation \eqref{eqsf} with an infinite sequence of higher symmetries or local conservation laws belongs to the following list:
\begin{align*}
&{\bf u}_t={\bf u}_{xxx}-3\,\frac{u_{[1,2]}}{u_{[1,1]}}\,{\bf u}_{xx}+\frac{3}{2}\,\biggl(\frac{u_{[2,2]}}{u_{[1,1]}}+\frac{u_{[1,2]}^2}{u_{[1,1]}^2\,(1+a\,u_{[1,1]})}\biggr)\,{\bf u}_x, \\[2mm]
&{\bf u}_t={\bf u}_{xxx}+\frac{3}{2}\biggl(\frac{a^2\,u_{[1,2]}^2} {1+a\,u_{[1,1]}}-a\,(u_{[2,2]}-u_{[1,1]}^2)+u_{[1,1]}\biggr)\,{\bf u}_x +3\,u_{[1,2]}\,{\bf u},  \\[2mm]
&{\bf u}_t= {\bf u}_{xxx}-3\,\frac{u_{[1,2]}}{u_{[1,1]}}\,{\bf u}_{xx}+ \frac{3}{2}\frac{u_{[2,2]}}{u_{[1,1]}}\,{\bf u}_x,\\[2mm]
&{\bf u}_u={\bf u}_{xxx}-3\,\frac{u_{[1, 2]}}{u_{[1, 1]}}\,{\bf u}_{xx}+3\,\frac{u_{[2, 2]}}{u_{[1, 1]}}\, {\bf u}_x,\\[2mm]
&{\bf u}_t=\,{\bf u}_{xxx}-3\,{\frac {\left (q+1\right )\, u_{[1,2]}}{2\,q\,u_{[1,1]}}}\,{\bf u}_{xx}+3\,{\frac {\left (q-1\right )\,u_{[1,2]}}{2\,q}} \, {\bf u}\\[1mm]
&\qquad+\frac{3}{2}\left ({\frac {\left (q+1\right ) u_{[2,2]}}{u_{[1,1]}}}-{\frac {\left (q+1\right )a\, u_{[1,2]}}{2}}{{q}^{2} u_{[1,1]}}+u_{[1,1]}\left(1-q\right )\right ){\bf u}_x,\\[2mm]
&{\bf u}_t={\bf u}_3+3\,u_{[1,1]} {\bf u}_1+3\,u_{[1,2]}\,{\bf u}, \\[2mm]
&{\bf u}_t={\bf u}_3+\frac{3}{2}\,u_{[1,1]} {\bf u}_1+3\,u_{[1,2]} {\bf u},
\end{align*}
where $a$ and $c$ are arbitrary constants, $q=\sqrt{\mathstrut 1+a\,\bf u_{[1,1]}}$.
\end{theorem}

\subsection{Equations with two scalar products}

Suppose that the coefficients of the equation  \eqref{evec} are functions of the variables \eqref{anizvar}, where $i\le j\le 2$.
Under the assumption $\langle {\bf u},\,{\bf u}\rangle = 1\, $ and up to transformations $\bar v[i,j] = u[i,j]+{\rm const}\, v[i,j],$ all equations that have 
an infinite series of higher symmetries 
 have been found \cite{meshsokCMP, BalMesh}.  The equations  of the classification list that have rational coefficients are:
$${\bf u}_t={\bf u}_3+\Big(\frac{3}{2}\,u_{[1,1]}+v_{[0,0]}\Big)\,{\bf u}_1
+3\,u_{[1,2]}\,{\bf u}_0, 
$$
$$
{\bf u}_t={\bf u}_3 - 3\,\frac{u_{[1,2]}}{u_{[1,1]}}\,{\bf u}_2+
\frac{3}{2}\biggl(\frac{u_{[2,2]}}
{u_{[1,1]}}+\frac{u_{[1,2]}^2}{u_{[1,1]}^2}+\frac{v_{[1,1]}}{u_{[1,1]}}
\biggr)\,{\bf u}_1,
$$
$$
{\bf u}_t={\bf u}_3-3\frac{u_{[1,2]}}{u_{[1,1]}}{\bf u}_2 +
\frac{3}{2}\biggl(\frac{u_{[2,2]}}{u_{[1,1]}}+
\frac{u_{[1,2]}^2}{u_{[1,1]}^2}-\frac{(v_{[0,1]}+u_{[1,2]})^2}{q\,u_{[1,1]}}
+\frac{v_{[1,1]}}{u_{[1,1]}}\biggr) {\bf u}_1,
$$

where $q=u_{[1,1]}+ v_{[0,0]}+a,\,\,$ $a$ is an arbitrary constant,
$$
\boldsymbol u_t=\boldsymbol u_3-3\frac{v_{[0,1]} }{v_{[0,0]}} \boldsymbol u_2
-3\left(\frac{v_{[0,2]} }{v_{[0,0]}}-2\frac{ v_{[0,1]}^2 }{v_{[0,0]}^2}\right)\boldsymbol u_1
+3 \left(u_{[1,2]} -\frac{v_{[0,1]} }{v_{[0,0]}}u_{[1,1]}\right)\boldsymbol u ,  $$
$$
\boldsymbol u_t=\boldsymbol u_3-3\frac{v_{[0,1]} }{v_{[0,0]}} \boldsymbol u_2 
-3\left(\frac{2v_{[0,2]}+v_{[1,1]}+a}{2v_{[0,0]}}-\frac{5}{2}\frac{v_{[0,1]}^2 }{v_{[0,0]}^2}\right)\boldsymbol u_1 
+3 \left(u_{[1,2]} -\frac{v_{[0,1]} }{v_{[0,0]}}u_{[1,1]}\right)\boldsymbol u ,  
$$
$$
\boldsymbol u_t= \boldsymbol u_3- 3\frac{v_{[0,1]} }{v_{[0,0]}}\left( \boldsymbol u_2+u_{[1,1]}\boldsymbol u \right) +3 u_{[1,2]}\,\boldsymbol u+ 
\frac{3}{2}\Big(-\frac {u_{[2,2]}}{v_{[0,0]}}+\frac {(u_{[1,2]}+v_{[0,1]})^{2}}{v_{[0,0]}
(v_{[0,0]}+u_{[1,1]})}+ $$$$\qquad \quad +\frac {(v_{[0,0]}+u_{[1,1]})^{2}}{v_{[0,0]}}
+\frac {v_{[0,1]}^{2}-v_{[0,0]}\,v_{[1,1]}}{v_{[0,0]}^{2}}\Big)\, \boldsymbol u_1. 
$$
The first of these equations is just  \eqref{LL}. To justify the integrability for all equations auto-B\"acklund transformations with the spectral parameter were found.

\subsubsection{Vector hyperbolic equations on the sphere with integrable third-order symmetries}
In \cite{MeshSokHy} vector hyperbolic equations of the form 
$$
{\bf u}_{xy}= h_0 {\bf u}+h_1  {\bf u}_x+h_2  {\bf u}_y, \qquad {\bf u}^2=1
$$
on the sphere were considered. Here,  $h_{i}$ are some  scalar-valued  functions 
depending on two different scalar products  $(\cdot\, ,\cdot)$ and $\langle\cdot\, ,\cdot\rangle$ between vectors ${\bf u}, {\bf u}_x$ and ${\bf u}_y$.  All such equations that have integrable vector $x$ and $y$-symmetries were found (cf. Section \ref{hyphyp}).

\begin{example} A hyperbolic integrable equation on the sphere is given by 
$$
\bf u_{xy}=\frac{\bf u_x}{\langle\bf u,\bf u \rangle}\Big(\langle\bf u,\bf u_y \rangle+\sqrt{\vphantom{\bf u_y^2}1+
\langle\bf u,\bf u \rangle (\bf u_x, \bf u_x)^{-2}}\ \phi \Big)-(\bf u_x,\bf u_y)\bf u,  $$$$
{\rm where} \qquad \phi =\sqrt{\langle\bf u,\bf u_y \rangle^2+\langle\bf u,\bf u\rangle(1-\langle\bf u_y,\bf u_y\rangle)}.$$
In the case $N=2$ this equation is equivalent to (cf. \eqref{borzyk} )
$$
u_{xy}= {\rm sn}(u) \sqrt{\vphantom{u_y^2}u_x^2+1}\sqrt{u_y^2+1}.  
$$
\end{example}

\subsection{Equations with constant vector}\label{sub627}

In Section \ref{JordKdv} the   integrable vector KdV equation \eqref{VectorKdv} appears. 
The simplest conserved densities for this equation are given by
$$ \rho_1=({\bf c}, {\bf u}), \qquad  \rho_2={\bf c}^2 {\bf u}^2-2{(\bf c,\bf u)}^2,
$$ and 
$$ \rho_3=6 ({\bf c},{\bf u}_x)^2-3 {\bf c}^2 {\bf u}_x^2+3 {\bf c}^2 {\bf u}^2({\bf c},{\bf u})-4({\bf c},{\bf u})^3.$$
This equation belongs to the class of equations of the form 
\begin{equation}\label{vecc}
  {\bf u}_{t} =f_3 {\bf u}_{xxx}+f_2  {\bf u}_{xx}+f_1  {\bf u}_x+f_0  {\bf u}+h \, {\bf c}.
\end{equation}
In this case the coefficients depend on 
\begin{equation}\label{verC}
 ({\bf u}_i,\, {\bf u}_j), \qquad ({\bf u}_i,\, {\bf c}), \qquad i\le j.
 \end{equation}
Without loss of generality we may assume that $({\bf c},{\bf c})=1.$
 
All statements of theorem \ref{th51} can be easily generalized to the case of equations of the form
\begin{equation}
{\bf u}_t=f_n \, {\bf u}_n+f_{n-1}\, {\bf u}_{n-1}+\cdots+f_1\,{\bf u}_1+f_0\,{\bf u} + h\, {\bf c} = A({\bf u})+h \, {\bf c},  
\label{gensyscc}
\end{equation}
where the coefficients $f_i$ and $h$ depend on \eqref{verC}. 

\begin{theorem} Equations \eqref{gensyscc}, which have infinitely many symmetries, possess a formal Lax operator
$$
L=a_1 \, D_x+a_0 \, +a_{-1}\, D_x^{-1} +\cdots   \,
$$
with scalar coefficients such that 
$$L_t=[A, \, L].$$ 
\end{theorem}

For equation \eqref{vecc} the recurrent formula for densities of canonical conservation laws  
$$
D_t\,\rho_i=D_x\,\theta_i, \qquad i=-1,0,1,2,\dots
$$
is given by
\begin{align*}
\rho_{n+2}&=\frac{a}{3}\biggl(\theta_n-f_0\,\delta_{n,0} -2\,af_2\,\rho_{n+1}-f_2\,\frac{d}{dx}\rho_{n} - f_1\,\rho_{n}\biggr) \nonumber\\[2mm]
&-\frac{a}{3}\,f_2\,\sum_{i+j=n} \rho_{i}\,\rho_{j}-\frac{1}{3}\,a^{-2}\sum_{i+j+k= n}\rho_{i}\,\rho_{j}\,\rho_{k}- a^{-1}\sum_{i+j=n+1}\rho_{i}\, \rho_{j}
\nonumber \\[2mm]
&-a^{-2} D\,a\,\rho_{n+1}-\frac{1}{2}\,a^{-2}D \sum_{i+j=n}\rho_{i}\,\rho_{j}-\frac{1}{3}\,a^{-2}\,D^2\, \rho_{n},\qquad n\ge 0, \label{rekkur1}
\end{align*}

 where $f_3=a^{-3}$,  $\delta_{i,j}$ is the Kronecker delta,  and 
$$
\rho _{-1}=a,
$$
$$\rho _0=-\frac{1}{3}\,a^3 f_2-D \ln a,
$$
$$
\rho _1=\frac{a}{3}\theta _{-1}+a^{-1}\rho_0^2-\frac{1}{3}\,a^2 f_1+a^{-3}\left(D \,a\right)^2+2\,a^{-2}\rho _0 
D\,a-a^{-1}\,D\,\rho _0
-\frac{1}{3}\,a^{-2} D^2\,a.
$$

\begin{remark} Formula \eqref{rekkur} is valid for integrable equations \eqref{gensys} with $n=3$ in the cases, when the coefficients depend on one or two scalar products. If $f_3=a=1$ it coincides with \eqref{rekkur}.
\end{remark}

Although the recurrent formula for the canonical densities are obtained, no classification results for equations \eqref{vecc} are known.

\chapter{Appendices}

\section{Appendix 1. Hyperbolic equations with integrable third order symmetries}
\begin{theorem}
\begin{align*}
u_{xy}&=c_1e^u+c_2e^{-u}; \qquad  \qquad 
u_{xy} = f(u)\sqrt{\vphantom{u_y^2}u_x^2+1}, \quad {\rm where} \quad f''=c f;  \\[1.5mm]
u_{xy}&= \sqrt{\vphantom{u_y^2}u_x}\sqrt{u_y^2+1}; \qquad \qquad 
u_{xy} = \sqrt{P(u)-\mu}\,\sqrt{\vphantom{u_y^2}u_x^2+1}\sqrt{u_y^2+c};  \\[1.5mm]
u_{xy}&= 2\,uu_x; \qquad \qquad 
u_{xy} = 2\,u_x\,\sqrt{u_y};  \qquad \qquad 
u_{xy} = u_x\,\sqrt{u_y^2+1};  \\[1.5mm]
u_{xy}& = \sqrt{\rule{0pt}{1.4ex}u_xu_y};  \qquad \qquad 
u_{xy} = \frac{u_x(u_y+c)}{u},\ \ c\ne0;  \qquad \qquad 
u_{xy} = (c_1 e^u+c_2 e^{-u})\,u_x;   \\[1.5mm]
u_{xy}&= \frac{u_y \eta}{\sinh(u)}\Big(\eta e^u - 1\Big); \qquad \qquad u_{xy}= \frac{2 u_y \eta}{\sinh(u)}\Big(\eta \cosh(u) - 1\Big); \\[1.5mm]
u_{xy}&= \frac{2 \eta \xi}{\sinh(u)}\Big((\eta \xi+1) \cosh(u) - \xi - \eta \Big); \qquad \qquad u_{xy}= \frac{u_y}{u} \eta (\eta - 1)+
c\,u\, \eta (\eta+1);  \\[1.5mm]
 u_{xy}&= \frac{2 u_y}{u} \eta (\eta - 1); \qquad  u_{xy} = \frac{2 \eta \xi}{u} (\eta - 1) (\xi - 1); \qquad \qquad u_{xy}=\frac{u_x u_y}{u}-2 u^2 u_y;  \\[1.5mm]
 u_{xy}&=\frac{u_x(u_y+c)}{u}-u u_y; \qquad \qquad u_{xy} = \sqrt{u_y} + c\, u_y; \qquad \qquad u_{xy} = c\, u.
\end{align*}
 Here, $\mu$ is a solution of the equation $4 \mu^3-g_2 \mu-g_3=0$,  $c,c_1,c_2,g_2,g_3$ are constants, and  $\xi =\sqrt{u_y+1}\,,\ \eta =\sqrt{u_x+1}.\,$
 \end{theorem}
 
\section{Appendix 2. Scalar hyperbolic equations of Liouville type}
In this section   {\it nonlinear} equations
(\ref{hyper}) of Liouville type known to the author  are collected \cite{zibsok,Laine}. A list is
presented up to the involution
$x \leftrightarrow y$ and  
transformations of the form
$$
x \rightarrow \zeta(x), \qquad y \rightarrow \xi(y), \qquad
u \rightarrow \theta(x,y,u).
$$

{\bf Class 1.} Equations of this class have the form
\begin{equation}\label{EQU1}
u_{xy}=-\frac{W_y}{W_{u_x}}.
\end{equation}
Here the function $W(x,y,u_x)$ is defined from the equation
$$
u_x=q_0(x,y)+\sum_{i=1}^n\alpha_i(y) \, q_i(x,W),
$$
where $\alpha_i, q_i$ are arbitrary functions. $\square$

{\bf Class 2.} (see Section \ref{sec351})
\begin{equation}\label{EQU2}
u_{xy}=-\frac{P_{u}}{P_{u_{x}}}\, u_y,
\end{equation}
where
$$
u_x=\alpha(x,P) u^2+\beta(x,P) u+\gamma(x,P),
$$
and $\alpha, \beta, \gamma$ are arbitrary functions.
\begin{remark} The following well known
equation:
$$
u_{xy}= e^u \, u_y
$$
gives us an example of equation from Class 2.
\end{remark}

{\bf Class 3.} Equations of the form
$$
u_{xy}=A_n(x,y) \, \sqrt{u_x\, u_y}, \qquad n\ge 1,
$$
where $A_n(x,y)$ is defined as follows, are Liouville integrable. Let 
\begin{equation}\label{ch2}
  h_{0}=\frac 14 A_n^2, \qquad h_{1}=\frac 14 A_n^2-
  \frac{\pa^2}{\pa x\pa y}\ln A_n.
\end{equation}
Define $h_i, \, i>1$ by the formula 
\begin{equation}\label{ch1}
  h_{k+1}=2h_k-h_{k-1}-\frac{\pa^2}{\pa x\pa y}\ln h_k,
  \qquad k=1,\ldots,  
\end{equation} 
Then, by definition, $A_n$ is any solution of the equation 
$
  h_{n}=0.
$
 In particular, Class 3 contains equations with
$$
A_n=\frac{2 n \lambda}{\lambda (x+y)-xy}
$$
and their degenerations with 
$$
A_n=\frac{2 n}{(x+y)}. \qquad \square
$$

Apart from these three series of Liouville type equations the following equations of the form 
\begin{equation} \label{annz}
u_{xy}= A(x,y,u) \, B(u_x) \bar B(u_y)
\end{equation}
are known: \eqref{liou},  \eqref{pryamoy}, \eqref{kosliu},

\begin{equation} \label{he1} 
u_{xy}=\left(\frac
1{u-x}+\frac 1{u-y}\right)\,u_x\,u_y; 
\end{equation}
\begin{equation} \label{he2}  u_{xy}=f(u) \, b(u_x), \qquad {\rm where} \qquad f f''-f'^2=0, \qquad bb^\prime+u_1=0;
\end{equation}
 \begin{equation} \label{he3}   u_{xy}=f(u)\, b(u_x) \, \bar b(u_y),  \quad {\rm where} \qquad (\ln f)^{\prime\prime}-f^2=0, \quad
bb^\prime+u_1=0, \quad \bar b\bar b^\prime+\bar u_1=0;
\end{equation}
\begin{equation} \label{he4} 
u_{xy}=\frac 1u \, b(u_x) \, \bar b(u_y),   \quad {\rm where} \qquad \ bb^\prime+c b+u_1=0,
\qquad \quad \bar b\bar b^\prime+c \bar b+\bar u_1=0;
\end{equation}
 \begin{equation} \label{he5} 
u_{xy}=\frac 1{(x+y)\, b(u_x)\, \bar b(u_y)},  \quad {\rm where} \qquad b^\prime=b^3+b^2, \qquad \qquad \bar b^\prime=\bar b^3+\bar b^2; 
\end{equation}
\begin{equation} \label{he6} 
u_{xy}=\frac{1}{u} \, b(u_x)\, \bar b(u_y), \quad {\rm where} \qquad bb^\prime+b-2u_1=0, \qquad \quad \bar b\bar b^\prime+\bar
b-2\bar u_1=0;
\end{equation}
The equations of the form more complicated than \eqref{annz} are \cite{Laine,zibsok2}
\begin{equation}
\label{L1} u_{xy}=\left(\frac{u_{y}}{u-x}+\frac{u_{y}}{u-y}\right)u_{x}+\frac{u_{y}}{u-x}\sqrt{u_{x}};
\end{equation}
 \begin{equation}
\label{L2}
u_{xy}=2\left[(u+Y)^{2}+u_{y}+(u+Y)\sqrt{(u+Y)^{2}+u_{y}}\right]\times\left[\frac{\sqrt{u_{x}}+u_{x}}{u-x}-\frac{u_{x}}{\sqrt{(u+Y)^{2}+u_{y}}}\right],
\end{equation}
where $Y=Y(y)$ is an arbitrary function, and \eqref{Eq1}.  
  
 \section{Appendix 3. Integrable scalar evolution equations}
 \subsubsection{Admissible point transformations}
 
Let us describe point transformations we use in the classification of   equations \eqref{scalar}.

\begin{itemize}
\item[1)] The transformations\footnote{Hereinafter, once  transformation rules for some of the variables $t,$ $x$, or $u$ are not indicated in the formulas, this means that the corresponding variables are not changed.}
$$
 \tilde u = \phi(u);
$$
\item[2)] the scalings
$$
\tilde x = a x,\qquad \ \tilde t = a^n t ;
$$
\item[3)] the Galilean transformation
$$
 \tilde x = x + c t;
$$
\item[4)] If the function $F$ is independent of $u$, then the transformation
$$
\tilde u = u + c_1 x + c_2 t 
$$
is admissible;  
\item[5)] if $F(\lambda u, \lambda u_1, \dots,\lambda u_{n-1}) =\lambda F(u, u_1, \dots,u_{n-1})$, then for arbitrary constants $a$ and $b$ the transformation
$$
\tilde u = u \exp(at + bx) \label{t6}
$$
is applicable.
\end{itemize}

The equations related by the above transformations are called
{\it equivalent}. It is important to note that the classification is purely algebraic. We are not interested in properties of the solutions of the studied equations such as being real and the functions and constants that are involved in the transformations can be complex. For instance, the equations $u_t=u_3-u_1^3$ and
$u_t=u_3+u_1^3$ are regarded as equivalent.

\subsection{Third order equations}
 
  \begin{theorem} {\rm \cite{soksvin2, sokshab, meshsok} 
 Up to transformations of the form 1)--5) each non--linear equation 
$$
u_t=u_3+F(u,\, u_1, \, u_2).
$$}
possessing infinitely many higher symmetries 
  belongs to the list
\begin{align}
&u_t=u_{xxx}+uu_x,\ \nonumber \\[1.5mm]
&u_t=u_{xxx}+u^2u_x,\nonumber\\[1.5mm]
&u_t=u_{xxx}+u_x^2, \nonumber\\[1.5mm]
&u_t=u_{xxx}-\frac{1}{2}u_x^3+(c_1e^{2u}+c_2e^{-2u})u_x, \nonumber\\[1.5mm]
&u_t=u_{xxx}-\frac{3u_xu_{xx}^2}{2(u_x^2+1)}+ c_1 (u_x^2+1)^{3/2}+ c_2 u_x^3, \nonumber\\[1.5mm]
&u_t=u_{xxx}-\frac{3u_{xx}^2}{2u_x}+\frac{1}{ u_x}-\frac{3}{2} \wp(u)u_x^3, \label{list7}\\[1.5mm]
&u_t=u_{xxx}-\frac{3u_xu_{xx}^2}{2(u_x^2+1)}-\frac{3}{2} \wp(u)u_x(u_x^2+1), \label{list6}\\[1.5mm]
&u_t=u_{xxx}-\frac{3u_{xx}^2}{2u_x}, \nonumber\\[1.5mm]
&u_t=u_{xxx}-\frac{3u_{xx}^2}{4u_x}+c_1u_x^{3/2}+c_2u_x^2, \nonumber\\[1.5mm]
&u_t=u_{xxx}-\frac34\,\frac{u_{xx}^{2}}{{u_{x}}+1}+3\,u_{xx}u^{-1}(\sqrt{u_{x}+1}-u_{x}-1)\nonumber\\
&\qquad-6 \,u^{-2} u_{x}(u_{x}+1)^{3/2}+3 \,u^{-2} u_{x}\,(u_{x}+1)(u_{x}+2),\nonumber \\[3.5mm]
&u_t=u_{xxx}-\frac34\,\frac {u_{xx}^{2}}{{u_{x}}+1}-3\,\frac{u_{xx}\,(u_{x}+1)\cosh {u}}{\sinh{u}}+3\,\frac {u_{xx}\sqrt {u_{x}+1}}{\sinh{u}}\nonumber \\
&\qquad -6\,\frac {u_{x}(u_{x}+1)^{3/2}\cosh{u}}{\sinh^{2}u}+3\,\frac {u_{x}\,(u_{x}+1)(u_{x}+2)}{\sinh^{2}u}+u_{x}^2(u_{x}+3),\nonumber \\[3.5mm]
&u_t=u_{xxx}+3\,u^2u_{xx}+9\,uu_x^2+3\,u^4u_x, \nonumber\\[1.5mm]   
&u_t=u_{xxx}+3\,uu_{xx}+3\,u_x^2+3\,u^2 u_x.  
\end{align}
Here, $\wp(u)$ is any solution of the equation $ ( \wp')^2=4 \wp^3-g_2 \wp-g_3$, and $c_1,c_2, g_2, g_3$ are arbitrary constants. 
 \end{theorem}

\begin{remark} Quite often instead of  equations \eqref{list7} and \eqref{list6} one considers point equivalent to them equations \eqref{KN} and \eqref{CD2}. 
If $P'\ne0$, then one can make the transformation $u=f(v)$, where $(f')^2= P(f)$, in  equations \eqref{KN} and \eqref{CD2}. Then for $v$ we get   equations \eqref{list7} and \eqref{list6}, respectively.  
\end{remark}

\subsection{Fifth order equations} \begin{theorem} {\rm \cite{meshsok} Suppose nonlinear equation 
$$
u_t=u_5+F(u, u_x, u_2, u_3, u_4)
$$
 satisfies two conditions}:
\begin{itemize}
\item[ 1)] there exists an infinite sequence of higher symmetries
\begin{equation} \label{symm}
u_{\tau_{i}}=G_i(u,...,u_{n_{i}}), \qquad i=1,2,... , \qquad n_{i+1}>n_{i}>\cdots>5;
\end{equation}
\item[ 2)] there exist no symmetries \eqref{symm} of orders $1<n_i<5.$
\end{itemize}
Then the equation is equivalent to one in the  list

\begin{align*}
&u_t=u_{5}+5 u u_{3}+5 u_1 u_2+5 u^2 u_1,  
\\[2mm]
\label{tst2}
&u_t=u_{5}+5 u u_{3}+\frac{25}{2} u_1 u_2+5 u^2 u_1, \\[2mm]
 &u_t=u_{5}+5 u_1 u_{3}+\frac{5}{3} u_1^3, \\[2mm]
 &u_t=u_{5}+5 u_1 u_{3}+\frac{15}{4}u_2^2 + \frac{5}{3} u_1^3, \\[2mm]
 &u_t=u_{5}+5 (u_1-u^2) u_{3}+5 u_2^2-20 u u_1 u_2-5 u_1^3+5 u^4 u_1, \\[2mm]
 &u_t=u_{5}+5 (u_2-u_1^2) u_{3}-5 u_1 u_2^2+u_1^5, \\[3mm]
 &\begin{aligned}
u_t&=u_{5}+5 (u_2-u_1^2+\lambda_1 e^{2u}-\lambda_2^2 e^{-4u}) \, u_{3}-5 u_1 u_2^2+15 (\lambda_1 e^{2u}+4 \lambda_2^2 e^{-4u})\, u_1 u_2 
\\[1.5mm]
&\qquad +u_1^5-90 \lambda_2^2 e^{-4u}\, u_1^3 +5(\lambda_1 e^{2u}-\lambda_2^2 e^{-4u})^2\, u_1,
\end{aligned}
\\[3mm]
 &\begin{aligned}
u_t&=u_{5}+5 (u_2-u_1^2-\lambda_1^2 e^{2u}+\lambda_2 e^{-u}) \, u_{3}-5 u_1 u_2^2-15 \lambda_1^2 e^{2u} \, u_1 u_2 \\[1.5mm]
&\qquad +u_1^5+5(\lambda_1^2 e^{2u}-\lambda_2 e^{-u})^2 \, u_1, \qquad \lambda_2\ne 0,
\end{aligned}
\\[3mm]
 &
\begin{aligned}
u_t&=u_{5}-5\frac{u_2 u_{4} }{ u_1}+ 5\frac{u_2^2 u_{3}}{ u_1^2}+5\left(\frac{ \mu_1}{ u_1}+\mu_2 u_1^2\right)u_{3}-5\left(\frac{ \mu_1}{ u_1^2}+\mu_2 u_1\right)u_2^2
\\[2mm]
&\qquad \quad -5 \frac{\mu_1^2}{ u_1}+ 5 \mu_1\mu_2 u_1^2 +\mu_2^2 u_1^5,
\end{aligned}
\\[3mm]
 &\begin{aligned}
u_t&=u_{5}-5\frac{u_2 u_{4}}{u_1}-\frac{15}{4 }\,\frac{u_{3}^2}{u_1}+\frac{ 65}{4}\,\frac{u_2^2 u_{3}}{ u_1^2} +5\left(\frac{\mu_1}{u_1}+\mu_2 u_1^2\right)\, u_{3}
-\frac{135}{16 }\frac{u_2^4}{u_1^3}
\\
&\qquad -5\left(\frac{7 \mu_1}{4 u_1^2}-\frac{\mu_2 u_1}{2}\right) u_2^2-5 \frac{\mu_1^2}{u_1}+ 5 \mu_1\mu_2 u_1^2 +\mu_2^2 u_1^5,
\end{aligned}
\\[3mm]
\end{align*}
\begin{align*}
&\begin{aligned}
u_t&= u_{5}-\frac{ 5 }{2}\,\frac{ u_2u_{4} }{ u_1}-\frac{5}{4}\,\frac{u_{3}^2}{ u_1}+5\frac{u_2^2 u_{3}}{u_1^2} +\frac{5\, u_2 u_{3}}{2 \sqrt{u_1}}
-5 (u_1-2 \mu u_1^{1/2}+\mu^2) \, u_{3} -\frac{ 35}{16}\,\frac{u_2^4}{ u_1^3} \\
 &\qquad  -\frac{5}{3}\,\frac{u_2^3}{u_1^{3/2}} +5\Big (\frac{ 3 \mu^2}{ 4 u_1} -\frac{\mu }{\sqrt{u_1}}+\frac{1}{4}\Big)\, u_2^2
+ \frac{ 5}{ 3}\, u_1^3
 - 8 \mu u_1^{5/2}+15 \mu^2 u_1^2-\frac{40}{3}\,\mu^3 u_1^{3/2},
\end{aligned} \\[4mm]
&\begin{aligned}
u_t&=u_{{{ 5}}}+\frac52\,{\frac { f-u_{1}}{{f}^{2}}}\,u_{{2}}u_{{{ 4}}}+\frac54\,{\frac {2\,f -u_{1}}{{f}^{2}}}\,u_3^{2} +5\,\mu\, ( u_{1}+f ) ^{2}u_{{{3}}}+\frac54\,{\frac {4\,{u_{1}}^{2}-8\,u_{1}f+{f}^{2}}{{f}^{4}}}\,u_2^{2}u_{{{ 3}}} +\\
&{\frac {5}{16}}\,{\frac {2-9\,u_{1}^{3}+18\,u_{1}^{2}f}{{f}^{6}}}\,u_2^{4} +\frac {5\mu}{4}\,{\frac {( 4\,f-3\,u_{1} )( u_{1}+f )^{2}}{{f}^{2}}}\,u_2^{2}+{\mu}^{2} ( u_{1}+f ) ^{2}\big( 2\,f ( u_{1}+f ) ^{2}-1\big),
\end{aligned}
\\[4mm]
&\begin{aligned}
u_t&=u_{{{ 5}}}+\frac52\,\frac {f- u_{1}}{f^{2}}\, u_2u_4+\frac54\,\frac {2\,f- u_1}{f^2}\, u_3^{2}-5\,\omega\, ( {f}^{2}+u_1^{2} ) u_3 +\frac54\,{\frac {4\,u_1^{2}-8\,u_{1}f+{f}^{2}}{{f}^{4}}}\,u_2^{2}u_3+\\
&{\frac {5}{16}}\,\frac {2-9\,u_1^{3}+18\,u_1^{2}f}{f^6}\,u_2^{4}
 +\frac54\,\omega\,\frac {5\,u_1^{3} -2\,u_1^{2}f-11\,u_{1}{f}^{2}-2}{f^2}\, u_2^{2}-\frac52\,{\omega'}\, ( u_1^{2}-2\,u_{1}f+5\,{f}^{2} )u_{1}u_2 \\
&+5\,{\omega}^{2}u_{1}{f}^{2} ( 3\,u_{1}+f ) ( f-u_{1}),
\end{aligned}
\\[4mm]
&\begin{aligned}
u_{{t}}&=u_5+\frac52\,{\frac {f-u_{1}}{{f}^{2}}}\,u_2u_4+\frac54\,{\frac {2\,f -u_{1}}{{f}^{2}}}\,u_3^{2}
+\frac54\,{\frac {4\,u_1^{2}-8\,u_{1}f+{f}^{2}}{{f}^{4}}}\,u_2^{2}u_3
\\
&+{\frac {5}{16}}\,{\frac {2-9\,{u_{1}}^{3}+18\,{u_{1}}^{2}f}{{f}^{6}}}\,u_2^{4}+5\,\omega\,{\frac {2\,{u_{1}}^{3}+{u_{1}}^{2}f-2\,u_{1}{f}^{2}+1}{{f}^{2}}}\,u_2^{2}
\\
& -10\,\omega\,u_{{{3}}} ( 3\,u_{1}f+2\,u_1^{2}+2\,{f}^{2} )-10\,{\omega'} ( 2\,{f}^{2}+u_{1}f+{u_{1}}^{2} )\,u_{1}u_{{2}}\\[2mm]
& +20\,{\omega}^{2}u_{1} ( {u_{1}}^{3}-1 )( u_{1}+2\,f ),
\end{aligned}
\\[4mm]
&\begin{aligned}
u_t&=u_5+\frac52\,\frac {f-u_1}{f^2}\,u_2u_4+\frac54\,{\frac {2\,f-u_{1}}{{f}^{2}}}\,u_3^{2}-5\,c\frac {{f}^2+u_1^2}{{\omega}^{2}}\,u_3
\\
&+\frac54\,{\frac {4\,u_1^{2}-8\,u_{1}f+{f}^{2}}{{f}^{4}}}\,u_2^{2}u_3 +\frac {5}{16}\,\frac {2-9\,{u_{1}}^{3}+18\,u_1^{2}f }{{f}^{6}}\,u_2^{4}
\\
&-10\,\omega\, ( 3\,u_{1}f+2\,u_1^{2}+2\,{f}^{2} )\,u_3-\frac54\,c\,\frac { 11\,u_{1}{f}^{2}+2\,u_1^{2}f +2-5\,u_1^{3} }{{\omega}^{2}{f}^2}\,u_2^{2}
\\
&+5\,\omega\,\frac {2\,u_1^{3}+u_1^{2}f-2\,u_{1}{f}^{2}+1}{f^2}\,u_2^{2}+5\,c\,{\omega'}\,\frac {u_1^{2}+5\,{f}^{2}-2\,u_{1}f }{{\omega}^3}\,u_{1}u_2
\\
&-10\,{\omega'}\,(2\,{f}^{2}+u_{1}f+{u_{1}}^{2})\,u_1u_2 +20\,{\omega}^{2}u_{1} ( u_1^{3}-1 ) ( u_{1}+2\,f )
\\[2mm]
&+40\,{\frac {c\,u_{1}{f}^{3} ( 2\,u_{1}+f ) }{\omega}}+5\,{\frac {{c}^{2}u_{1}{f}^{2} ( 3\,u_{1}+f )( f-u_{1} ) }{{\omega}^{4}}},\qquad c\ne0.
\end{aligned}
\end{align*}
{\it Here}, $\lambda_1$, $\lambda_2$, $\mu$, $\mu_1$, $\mu_2$, {\it and} $c$ {\it are parameters},
{\it the function} $f(u_1)$ {\it solves the algebraic equation}
\begin{equation}\label{alg1}
(f+u_1)^2(2f-u_1)+1=0,
\end{equation}
{\it and} $\omega(u)$ {\it is any nonconstant solution to the differential equation}
\begin{equation}\label{Waier}
\omega'^2=4\, \omega^3+c. \qquad
\end{equation}
\end{theorem}

\section{Appendix 4. Systems of two equations}

\begin{theorem}{\rm \cite{MikShaYam87,SokWol99}}
Any nonlinear nontriangular system \eqref{kvazgen},
having a symmetry \eqref{kvazsym},
up to scalings of $t$, $x$, $u$, $v$, shifts of $u$ and $v$, and the
involution
$$\label{inv}
u \leftrightarrow v, \qquad t \leftrightarrow -t
$$
belongs to the following list:
$$
 \begin{cases}
u_t=u_{xx} + (u + v) u_x + u v_x, \\[1.5mm]
v_t=-v_{xx}+ (u + v) v_x + v u_x, 
\end{cases}
$$
$$
 \begin{cases}
u_t= u_{xx} - 2 (u + v) u_x - 2 u v_x +2 u^2 v + 2 u v^2 +
\alpha u + \beta v + \gamma, \\[1.5mm]
v_t=-v_{xx} + 2 (u + v) v_x + 2 v u_x - 2 u^2 v - 2 u v^2  -
\alpha u - \beta v - \gamma, 
\end{cases}
$$
$$
 \begin{cases}
u_t=u_{xx} + v u_x + u v_x, \\[1.5mm]
v_t=-v_{xx} + v v_x + u_x, 
\end{cases}
$$
$$
 \begin{cases}
u_t=u_{xx}+2 v u_x+2 u v_x+2 u v^2 + u^2 +\alpha u + \beta v + \gamma, \\[1.5mm]
v_t=-v_{xx}- 2 v v_x - u_x, 
\end{cases}
$$
$$
 \begin{cases}
u_t=u_{xx}+\alpha v_x + (u + v)^2+ \beta (u+v) + \gamma, \\[1.5mm]
v_t=-v_{xx}+\alpha u_x - (u + v)^2 - \beta (u+v) - \gamma,   
\end{cases}
$$
$$
 \begin{cases}
u_t=u_{xx} + (u + v) u_x + 4 \alpha v_x +
\alpha (u+v)^2 +\beta (u+v) +\gamma, \\[1.5mm]
v_t=-v_{xx} + (u + v) v_x + 4 \alpha u_x -
\alpha (u+v)^2 - \beta (u+v)-\gamma,  
\end{cases}
$$
$$
 \begin{cases}
u_t=u_{xx}+2 \alpha u^2 v_x+2 \beta u v u_x  +
\alpha (\beta - 2 \alpha) u^3 v^2 + \gamma u^2 v+\delta u, \\[1.5mm]
v_t=-v_{xx}+2 \alpha v^2 u_x+2 \beta u v v_x  -
\alpha (\beta - 2 \alpha) u^2 v^3 - \gamma u v^2-\delta v, 
\end{cases}
$$
$$
 \begin{cases}
u_t=u_{xx} + 2 u v u_x + (\alpha + u^2) v_x, \\[1.5mm]
v_t=-v_{xx} + 2 u v v_x +(\beta + v^2) u_x,  
\end{cases}
$$
$$
 \begin{cases}
u_t=u_{xx}+ 2 \alpha u v u_x + 2 \alpha u^2 v_x - \alpha \beta u^3 v^2
+\gamma u, \\[1.5mm]
v_t=-v_{xx}+2 \beta v^2 u_x + 2 \beta u v v_x + \alpha \beta u^2 v^3
-\gamma v, 
\end{cases}
$$
$$
 \begin{cases}
u_t=u_{xx}+ 2 u v u_x + 2 (\alpha + u^2) v_x+ u^3 v^2 + \beta u^3 +
 \alpha u v^2 +\gamma u, \\[1.5mm]
v_t=-v_{xx}-2 u v v_x - 2 (\beta + v^2) u_x- u^2 v^3-\beta u^2 v -
\alpha v^3  -\gamma v, 
\end{cases}
$$
$$
 \begin{cases}
u_t=u_{xx}+4 u v u_x+ 4 u^2 v_x + 3 v v_x + 2 u^3 v^2 + u v^3
+ \alpha u, \\[1.5mm]
v_t=-v_{xx} -2 v^2 u_x -2 u v v_x - 2 u^2 v^3 - v^4 - \alpha v, 
\end{cases}
$$
$$
 \begin{cases}
u_t=u_{xx}+4 u u_x + 2 v v_x, \\[1.5mm]
v_t=-v_{xx}-2 v u_x-2 u v_x- 3 u^2 v - v^3 + \alpha v, 
\end{cases}
$$
$$
 \begin{cases}
u_t=u_{xx}+v v_x, \\[1.5mm]
v_t=-v_{xx}+ u_x, 
\end{cases}
$$
$$
 \begin{cases}
u_t=u_{xx}+ 6 (u + v) v_x - 6 (u+v)^3-\alpha (u+v)^2-
\beta (u+v)-\gamma, \\[1.5mm]
v_t=-v_{xx}+ 6 (u + v) u_x + 6 (u+v)^3+\alpha (u+v)^2+
\beta (u+v)+\gamma, 
\end{cases}
$$
$$
 \begin{cases}
u_t=u_{xx}+v v_x, \\[1.5mm]
v_t=-v_{xx}+u u_x. 
\end{cases}
$$
\end{theorem}
Here, $\alpha, \beta, \gamma, \delta$ are arbitrary constants.  We omitted the term $(c u_x, c v_x)^t$ on the right-hand sides of
all systems. It is a Lie symmetry, corresponding to the invariance
of our classification problem with respect to the shift of $x$.

\end{document}